\documentclass[reqno,11pt]{amsart}
\usepackage[a4paper,height=20cm,top=26mm,bottom=26mm,left=26mm,right=26mm]{geometry}

\usepackage{subfiles}
\usepackage{cleveref}

\usepackage{etex}
\usepackage{graphicx}
\usepackage{epic,eepic}
\usepackage{xypic}
\usepackage{amsmath,amssymb,amscd}
\usepackage[mathscr]{eucal}

\usepackage{here}

\usepackage{bm}
\usepackage[all]{xy}
\usepackage{tikz}
\usetikzlibrary{calc, math}

\numberwithin{equation}{section}

\newtheorem{thm}{Theorem}[section]
\newtheorem{cor}[thm]{Corollary}
\newtheorem{lem}[thm]{Lemma}
\newtheorem{prop}[thm]{Proposition}

\theoremstyle{definition}
\newtheorem{definition}[thm]{Definition}

\newtheorem{conjecture}[thm]{Conjecture}
\theoremstyle{remark}
\newtheorem{remark}[thm]{Remark}

\crefname{thm}{Theorem}{Theorems}
\crefname{cor}{Corollary}{Corollaries}
\crefname{lem}{Lemma}{Lemmas}
\crefname{prop}{Proposition}{Propositions}
\crefname{definition}{Definition}{Definitions}
\crefname{example}{Example}{Examples}
\crefname{claim}{Claim}{Claims}
\crefname{conjecture}{Conjecture}{Conjectures}
\crefname{remark}{Remark}{Remarks}
\crefname{figure}{Figure}{Figures}
\crefname{section}{Section}{Sections}
\crefname{subsection}{Section}{Sections}

\crefname{introthm}{Theorem}{Theorems}
\crefname{introcor}{Corollary}{Corollaries}
\crefname{introconj}{Conjecture}{Conjectures}

\def\e{e}
\def\C{{\mathbb C}}
\def\Z{{\mathbb Z}}

\newcommand\R{{\mathbb{R}}}

\newcommand\LL{{\mathcal{L}}}
\newcommand\MM{{\mathcal{M}}}

\newcommand\Tr{{\mathrm{Tr}}}

\newcommand{\uu}{{\mathsf u}}
\newcommand{\ww}{{\mathsf w}}
\newcommand{\pp}{{\mathsf p}}
\newcommand{\qq}{{\mathsf q}}

\newcommand{\maru}[1]{\raise0.2ex\hbox{\textcircled{\scriptsize{#1}}}}

\newfont{\bg}{cmr9 scaled\magstep4}
\newcommand{\bigzerol}{\smash{\lower1.0ex\hbox{\bg 0}}}

\tikzset{
  mid arrow/.style={postaction={decorate,decoration={
        markings,
        mark=at position .5 with {\arrow[#1]{latex}}
      }}},
}

\begin{document}
\title[Quantized six-vertex model on a torus]
{Quantized six-vertex model on a torus}

\author[Rei Inoue]{Rei Inoue}
\address{Rei Inoue, Department of Mathematics and Informatics,
   Faculty of Science, Chiba University,
   Chiba 263-8522, Japan.}
\email{reiiy@math.s.chiba-u.ac.jp}

\author[Atsuo Kuniba]{Atsuo Kuniba}
\address{Atsuo Kuniba, Institute of Physics, Graduate School
of Arts and Sciences, University of Tokyo, Komaba, Tokyo, 153-8902, Japan.}
\email{atsuo.s.kuniba@gmail.com}

\author[Yuji Terashima]{Yuji Terashima}
\address{Yuji Terashima, Graduate school of science, Tohoku University,
6-3, Aoba, Aramaki-aza, Aoba-ku, Sendai, 980-8578, Japan}
\email{yujiterashima@tohoku.ac.jp}

\author[Junya Yagi]{Junya Yagi}
\address{Junya Yagi, Yau Mathematical Science Center, Tsinghua University, China}
\email{junyagi@tsinghua.edu.cn}

\date{May 13, 2025}


\begin{abstract}
We study the integrability of the quantized six-vertex model with four parameters on a torus. 
It is a three-dimensional integrable lattice model in which a layer transfer matrix, 
depending on two spectral parameters associated with the homology cycles of the torus, 
can be defined not only on the square lattice but also on more general graphs.
For a class of graphs that we call admissible, 
we establish the commutativity of the layer transfer matrices 
by introducing four types of tetrahedron equations and two types of inversion relations. 
Expanding in the spectral parameters yields a family of commuting quantum Hamiltonians.
The quantized six-vertex model can also be reformulated in terms of
(quantized) dimer models,
and encompasses known integrable systems as special cases,
including the free parafermion model and the relativistic Toda chain.
\end{abstract}



\maketitle


\section{Introduction}

The quantized six-vertex model (q-6v model, for short) with four
parameters was introduced by Kuniba, Matsuike and Yoneyama in
\cite{KMY23}, generalizing \cite{BS06, BMS10}. The model is given by
an operator $\mathscr{L}=\mathscr{L}(r,s,f,g;q)$ in
$\mathrm{End}(V \otimes V) \otimes \mathcal{W}(q)$, where $V$ is a
two-dimensional vector space, $\mathcal{W}(q)$ is the $q$-Weyl
algebra, and $r,s,f,g \in \mathbb{C}$ are parameters. Its graphical
representation is Figure \ref{fig:6v}, where one sees the weight
conservation, $i+j = a+b$, in $V \otimes V$ as in the (original)
six-vertex model \cite{B82}.

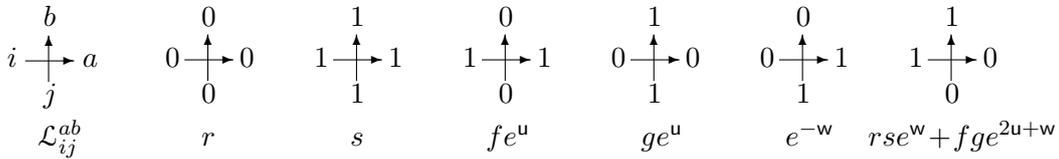
\begin{figure}[H]
\begin{picture}(330,60)(20,30)
{\unitlength 0.011in
\put(6,80){
\put(-11,0){\vector(1,0){23}}\put(0,-10){\vector(0,1){22}}
}
\multiput(81,80.5)(70,0){6}{
\put(-11,0){\vector(1,0){23}}\put(0,-10){\vector(0,1){22}}
}
\put(-74,0){
\put(60.5,77){$i$}\put(77.5,60){$j$}
\put(96,77){$a$}\put(77.5,96.5){$b$}
}
\put(61,77){0}\put(78,60){0}\put(96,77){0}\put(78,96.5){0}
\put(70,0){
\put(61,77){1}\put(78,60){1}\put(96,77){1}\put(78,96.5){1}
}
\put(140,0){
\put(61,77){1}\put(78,60){0}\put(96,77){1}\put(78,96.5){0}
}
\put(210,0){
\put(61,77){0}\put(78,60){1}\put(96,77){0}\put(78,96.5){1}
}
\put(280,0){
\put(61,77){0}\put(78,60){1}\put(96,77){1}\put(78,96.5){0}
}
\put(350,0){
\put(61,77){1}\put(78,60){0}\put(96,77){0}\put(78,96.5){1}
}
\put(78,40){
\put(-77,0){$\mathscr{L}^{ab}_{ij}$}
\put(0,0){$r$} \put(70,0){$s$} \put(134,0){$f e^\uu$}
\put(207,0){$g e^\uu$} \put(275,0){$e^{-\ww}$} \put(314,0){$rs e^\ww\!+\! fg e^{2\uu+\ww}$}
}}
\end{picture}
\caption{The operator $\mathscr{L}(r,s,f,g;q)$.}
\label{fig:6v}
\end{figure}

\noindent
The precise description of the model is provided in \S 2.1. 
This leads to a three-dimensional lattice model, and the first step 
in establishing its integrability is to find the conjugation operator $R$ satisfying a version of the tetrahedron equation \cite{Z80} known as the $RLLL$ relation \cite{BS06,BMS10,K22}:
\begin{align}\label{RLLL}
  R_{456}\mathscr{L}_{236}\mathscr{L}_{135}\mathscr{L}_{124} 
  = \mathscr{L}_{124}\mathscr{L}_{135}\mathscr{L}_{236}R_{456}.
\end{align}
In \cite{KMY23}, the operator $R$ was constructed by solving the recursion relations 
of matrix elements implied by \eqref{RLLL} for several infinite dimensional representations 
of $\mathcal{W}(q)^{\otimes 3}$. 
Further, in \cite{IKSTY} the cluster algebraic aspect of the $RLLL$ relation was unveiled. 
The operator $R$ comes from the quantum cluster algebra \cite{FG09} associated with 
the symmetric butterfly quiver (SB-quiver, for short).
It is constructed from a mutation sequence of the quiver and contains the quantum dilogarithm functions.

In this paper we study the integrability of the q-6v model, 
not only on the two-dimensional square lattice but also on more general wiring diagrams on a torus. 
A wiring diagram $G$ on a torus consists of a finite number of directed closed wires, 
such that exactly two wires intersect at each crossing.
The central object of our study is the layer transfer matrix $T_G(x, y)$,  
which depends on two spectral parameters $x$ and $y$ associated with the  
homology cycles of the torus. 
It is a nontrivial Laurent polynomial  
in $x$ and $y$, with coefficients valued in the tensor product of many copies  
of the $q$-Weyl algebra $\mathcal{W}(q)$.

To investigate the integrability of the model on general diagrams, 
we introduce four types of tetrahedron equations \eqref{LM-o}--\eqref{LM-t} 
and two types of inversion relations \eqref{M-I} and \eqref{M-I'}.  
For a given graph $G$,  
we define the notion of {\it admissibility} (Definition \ref{def:adm}),  and prove that
it provides a sufficient condition for the transfer matrix $T_G(x,y)$ to form a two parameter commuting family (Theorem \ref{thm:TT2}).  

Interestingly, by choosing the representation of the $q$-Weyl algebra and a graph $G$, 
we can relate the q-6v model to some known integrable models. In this paper we reproduce the free parafermion model \cite{B89} and the relativistic quantum Toda chain \cite{Suris90} from this viewpoint. Once we reformulate a model in terms of the q-6v model on some graph $G$, the commuting family of Hamiltonians for the model is generated by $T_G(x,y)$.   
We also explain the relation between the q-6v model and a dimer model (Table \ref{tab:6vd}) 
on wiring diagrams, which can be regarded as a quantization of the correspondence between 
the free fermion vertex models and dimers \cite{FW70,W68}. 
This relation can be naturally reduced to that between the quantized five vertex model (q-5v model for short) and another dimer model, 
by taking the limit of parameters (Table \ref{q5vdim} and \ref{5vd2}).
Combined with the results in \cite{IKSTY}, these relations are summarized as follows. 

\begin{align}\label{6v-5v}
\begin{CD}
\text{dimer model (Table \ref{5vd2})} @<<< \text{q-5v model} \\
@A{s \to 0}AA @AA{s \to 0}A \\
\text{dimer model (Table \ref{tab:6vd})} @<<< \text{q-6v model} @<{\text{\cite[\S 7]{IKSTY}}}<< \text{$R$ for SB-quiver}\\
@V{f \to 0}VV @VV{f \to 0}V @V{\text{\cite[\S 8]{IKSTY}}}VV\\
\text{dimer model (Table \ref{q5vdim})} @<<< \text{q-5v model} @<<< \text{$R$ for FG-quiver}\\[3mm]
\end{CD}
\end{align}
Here `$R$ for FG-quiver' is defined in \cite{IKT1}, associated to a mutation sequence for the Fock-Goncharov quiver.

The relation between the q-6v model and dimer models reminds us of the cluster integrable systems by Goncharov and Kenyon \cite{GK13}, but they seem different.
In the case of the q-6v model, the spectral parameters $x$ and $y$ associated with the homology cycles on a torus are commuting, whereas in \cite{GK13} they are not.
We also note that for the q-6v model, the Yang-Baxter move of the graph $G$ 
is realized by the adjoint action of the operator $R$ due to the $RLLL$ relation \eqref{RLLL}. 
(See Theorem \ref{th:RLLL}, which is due to \cite{IKSTY}.)

This paper is organized as follows. In \S 2, we recall the definition of the q-6v model and 
an infinite dimensional representation of the $q$-Weyl algebra $\mathcal{W}(q)$.
We also recall the $p$-oscillator algebra $\mathrm{Osc}(p)$ and its embedding in $\mathcal{W}(q)$, 
which is important for the construction of the invertible $R$-matrix \eqref{tred}. 
For the q-6v model we introduce the four tetrahedron equations 
and the two inversion relations, and recall the $RLLL$ relation in \cite{IKSTY}. 
In \S 3, we study the commutativity of transfer matrices $T_G(x,y)$. 
First we present the detailed proof for the commutativity for square grid case (Theorem \ref{thm:TT1}), 
where the invertibility of the $R$-matrix constructed in \S 2 is essential. 
After introducing the notion of admissibility for wiring diagrams, 
we prove commutativity in the admissible case (Theorem~\ref{thm:TT2}).
We also discuss the symmetry of $T_G(x, y)$, including its invariance 
under translations and the natural action of $SL(2, \mathbb{Z})$ on the torus 
(Proposition~\ref{prop:inv}), as well as the Yang--Baxter move on $G$ 
arising from the $RLLL$ relation (Corollary~\ref{cor:YB}).
Section 4 is devoted to some examples of commuting transfer matrices for admissible wiring diagrams. 
In \S 5, we relate the q-6v model to the free parafermion model and the relativistic quantum Toda chain.
In \S 6, we discuss the relation among q-6v model, q-5v model and dimer models on wiring diagrams.

\subsection*{Acknowledgement}
The authors would like to thank Gen Kuroki, Koji Hasegawa  and
Yasuhiko Yamada for stimulating discussions.
RI is supported by JSPS KAKENHI Grant Number 23K03048. 
AK is supported by JSPS KAKENHI Grant Number 24K06882.
YT is supported by  JSPS KAKENHI Grant Number 21K03240, 22H01117 and 25K06969.
JY is supported by NSFC Grant Number 12375064.

\section{Quantized six-vertex model}

\subsection{Operators $\LL$ and $\MM$}

We recall the quantized six-vertex (q-6v) model due to \cite{KMY23}. 
Let $V = \C v_0 \oplus \C v_1$ be a two-dimensional vector space.
For $\hbar \in \C$, let $\uu$ and $\ww$ be a canonical pair satisfying $[\uu,\ww] = \hbar$, hence $e^{\pm \uu}$ and $e^{\pm \ww}$ generate the $q$-Weyl algebra $\mathcal{W}(q)$ with a relation 
\begin{align}\label{q-comm}
e^\uu e^\ww = q e^\ww e^\uu; ~~q = e^\hbar.
\end{align}

We consider an operator $\LL=\mathcal{L}(r,s,f,g;q)$ given by 
\begin{align}
&\mathcal{L}(r,s,f,g;q) = \sum_{a,b,i,j=0,1} E_{ai}\otimes E_{bj} \otimes \mathscr{L}^{ab}_{ij} 
\in \mathrm{End}(V \otimes V)  \otimes \mathcal{W}(q),
\label{L1}
\\
&\mathscr{L}^{ab}_{ij}=0\; \text{unless}\; a+b=i+j,
\label{L2}
\\
&\mathscr{L}^{00}_{00} = r,\;\;  \mathscr{L}^{11}_{11} = s,\;\; \mathscr{L}^{10}_{10} = f e^\uu,\;\;
\mathscr{L}^{01}_{01} = g e^\uu, 
\;\;  \mathscr{L}^{10}_{01} = e^{-\ww}, \;\; \mathscr{L}^{01}_{10} = rs e^\ww+ fg e^{2\uu+\ww}.
\label{L3}
\end{align} 
Here $r,s,f,g \in \C$ are parameters, and the symbol $E_{ij}$ denotes the matrix unit on $V$ acting on the basis as $E_{ij}v_k = \delta_{jk}v_i$.  
Note that $\mathscr{L}^{ab}_{ij}$ depends also on $q$ via \eqref{q-comm}. 
The operator $\mathscr{L}(r,s,f,g;q)$ may be regarded as a quantized six-vertex model where the Boltzmann weights are $\mathcal{W}(q)$-valued. 
See Figure \ref{fig:6v} for a graphical representation.

We introduce a companion $\MM=\MM(r',s',f',g';q)$ of $\LL$ by
\begin{align}\label{eq:M} 
\MM=\sum_{a,b,i,j=0,1}E_{ai}\otimes E_{bj}\otimes \MM^{ab}_{ij}
=\LL(r',s',f',g';-q) \in \mathrm{End}(V\otimes V) \otimes \mathcal{W}(-q),
\end{align}
which means the operator obtained from $\LL$ by replacing the parameters 
$r,s,f,g$ by $r',s',f',g'$,  and $e^{\pm \uu}, e^{\pm \ww}$ by  generators 
$e^{\pm \uu'}, e^{\pm \ww'}$ of the algebra $\mathcal{W}(-q)$ 
satisfying $e^{\uu'}e^{\ww'}= -qe^{\ww'}e^{\uu'}$.

The operators $\LL^{ab}_{ij}$ and $\MM^{ab}_{ij}$ are depicted as follows:
\begin{align}\label{fig:LM}
\begin{tikzpicture}
\begin{scope}[>=latex]
\draw (-1.2,0.5) node{$\LL_{ij}^{ab}:$}; 
\draw[->] (0,0.5) node[left]{$i$}--(1,0.5) node[right]{$a$};
\draw[->] (0.5,0) node[below]{$j$}--(0.5,1) node[above]{$b$};
{\color{blue}
\draw[->] (1,0.9) --(0,0.1)[thick];
}
\end{scope}
\begin{scope}[>=latex,xshift=130]
\draw (-1.2,0.5) node{$\MM_{ij}^{ab}:$}; 
\draw[->] (0,0.5) node[left]{$i$}--(1,0.5) node[right]{$a$};
\draw[->] (0.5,0) node[below]{$j$}--(0.5,1) node[above]{$b$};
{\color{green}
\draw[->] (1,0.9) --(0,0.1)[thick] ;
}
\end{scope}
\end{tikzpicture}
.
\end{align}
The three arrows corresponding to the tensor components 1, 2, 3 constitute the right handed coordinate system, 
and the blue (resp. green) arrow indicates the multiplication of the element
$\LL^{ab}_{ij} \in \mathcal{W}(q)$
(resp. $\MM^{ab}_{ij} \in \mathcal{W}(-q)$).

We introduce the operator $h$ satisfying 
\begin{align}\label{h}
[h,e^{\uu}]=0, \; [h,e^{-\ww}]=e^{-\ww},\quad
[h,e^{\uu'}]=0, \; [h,e^{-\ww'}]=e^{-\ww'},
\quad
hv_k=k v_k \;(k=0,1).
\end{align}
It is  the ``number" operator counting  $v_1$ in $V$ and the power of $e^{-\ww}$ in $\mathcal{W}(q)$
(resp.~$e^{-\ww'}$ in $\mathcal{W}(-q)$).
Let $h_i$ denote the $h$ that acts on the $i$-th component of a tensor product.
From Figure \ref{fig:6v},  it is evident that the operator $\LL$ preserves 
the numbers counted by $h_1+h_2$ and $h_2+h_3$.
The same conservation law holds also for $\MM$.
We summarize this fact in
\begin{lem}\label{lem:L-h}
For arbitrary parameters $x$ and $y$, we have
\begin{align}\label{cl}
[x^{h_1} y^{h_2} (y/x)^{h_3}, ~ \LL] =0, \qquad 
[x^{h_1} y^{h_2} (y/x)^{h_3}, ~ \MM]=0.
\end{align}
\end{lem}

We will use the second relation later, which is depicted as 
\begin{align}\label{fig:Lh}
\begin{tikzpicture}
\begin{scope}[>=latex]
\draw[->] (0,0.5)--(1,0.5); \draw[fill=blue!20] (-0.15,0.5) circle[radius=0.15] node[left=2pt] {$x^{h}$}; 
\draw[->] (0.5,0)--(0.5,1); \draw[fill=blue!20] (0.5,-0.15) circle[radius=0.15] node[below=2pt] {$y^{h}$}; 
{\color{green}
\draw[->] (1,0.9)--(0,0.1)[thick]; 
}
\draw[fill=blue!20] (1.11,1.01) circle[radius=0.15] node[above=2pt] {$(y/x)^{h}$}; 
\draw (2.2,0.5) node {$=$};
\end{scope}
\begin{scope}[>=latex,xshift=100]
\draw[->] (0,0.5)--(1,0.5); \draw[fill=blue!20] (1.15,0.5) circle[radius=0.15] node[right=2pt] {$x^{h}$}; 
\draw[->] (0.5,0)--(0.5,1); \draw[fill=blue!20] (0.5,1.15) circle[radius=0.15] node[above=2pt] {$y^{h}$}; 
{\color{green}
\draw[->] (1,0.9)--(0,0.1)[thick]; 
}
\draw[fill=blue!20] (-0.11,-0.01) circle[radius=0.15] node[below=2pt] {$(y/x)^{h}$}; 
\end{scope}
\end{tikzpicture}
.
\end{align}

\subsection{Trace construction}

Set $\mathcal {V} = \bigoplus_{m \in \Z}\C(q) |m\rangle$, and 
let $\pi: \mathcal{W}(-q) \rightarrow \mathrm{End}(\mathcal{V})$ 
be a representation of $\mathcal{W}(-q)$ defined by 
\begin{align}\label{uwp}
e^{\pm \uu'}|m\rangle = (-q^{-1})^{\pm m}|m\rangle,
\quad
e^{\pm \ww'}|m\rangle = |m \mp 1\rangle,
\end{align}
where $\pi(\gamma)\, (\gamma \in \mathcal{W}(-q))$ is denoted by $\gamma$ for simplicity.
Consider the corresponding representation of $\MM$ (\ref{eq:M}) with parameters specialized 
depending on $\alpha \in \mathbb{C}$ as follows:
\begin{align}\label{Malpha}
(1 \otimes 1 \otimes \pi)(\MM(1,1,\alpha,q^{-1}\alpha^{-1};q))
\in \mathrm{End}(V \otimes V \otimes \mathcal{V}).
\end{align}

For $p \in \mathbb{C}$, let Osc$(p)$ be the $p$-oscillator algebra generated by $\mathbf{k}, \,\mathbf{a}^{+}$ and $\mathbf{a}^{-}$ with the relations
\begin{align}
\mathbf{k} \,\mathbf{a}^{\pm} = p^{\pm 1} \mathbf{a}^{\pm} \mathbf{k},
\quad
\mathbf{a}^+\mathbf{a}^- = 1 - \mathbf{k}^2,
\quad 
\mathbf{a}^-\mathbf{a}^+ = 1 - p^2\mathbf{k}^2.
\end{align}
We set $p = -q^{-1}$ and define a homomorphism of noncommuting algebras $\psi : \mathrm{Osc}(p) \to \mathcal{W}(-q)$ by 
\begin{align}
\mathbf{k} \mapsto e^{\uu'}, \quad \mathbf{a}^+ \mapsto e^{-\ww'}, \quad \mathbf{a}^- \mapsto e^{\ww'} +q^{-1}e^{2\uu'+\ww'} = e^{\ww'}( 1-e^{2\uu'}),
\end{align} 
where the last equality is due to the Baker-Campbell-Hausdorff formula.
For simplicity we identify $\gamma \in \mathrm{Osc}(-q^{-1})$ with its image $\psi(\gamma)$.  
From (\ref{uwp}), the generators of $\mathrm{Osc}(-q^{-1})$ linearly act on $\mathcal{V}$ as 
\begin{align}\label{posc}
\mathbf{k}|m \rangle = p^m |m\rangle,
\quad
\mathbf{a}^+|m\rangle = |m+1\rangle,
\quad
\mathbf{a}^-|m\rangle = (1-p^{2m})|m-1\rangle.
\end{align}
The representation $\pi$ of Im$(\psi)$ can be restricted to the subspace with non-negative ``mode" 
$\mathcal{V}_+ := \bigoplus_{m \in \Z_{\ge 0}}\mathbb{C}(q)|m \rangle \subset \mathcal{V}$.
We write $\pi_{\mathrm{res}}$ for this restriction.

It turns out that $\mathcal{M}_{ij}^{ab}$ belongs to Im${\psi}$, thus we write $M_{ij}^{ab}$ for $\pi_{\mathrm{res}}(\mathcal{M}_{ij}^{ab})$. The operators $M_{ij}^{ab}$ are graphically represented as Figure \ref{fig:6vM}. 
We write the corresponding specialization/restriction of $\MM$ as $M=M(\alpha)$. Namely,
\begin{align}\label{M}
M = \sum_{a,b,i,j=0,1}
E_{ai}\otimes E_{bj} \otimes M^{ab}_{ij} \in \mathrm{End}(V \otimes V \otimes \mathcal{V}_+).
\end{align}

\begin{figure}[H]
\begin{picture}(330,60)(20,30)
{\unitlength 0.011in
\put(6,80){
\put(-11,0){\vector(1,0){23}}\put(0,-10){\vector(0,1){22}}
}
\multiput(81,80.5)(70,0){6}{
\put(-11,0){\vector(1,0){23}}\put(0,-10){\vector(0,1){22}}
}
\put(-74,0){
\put(60.5,77){$i$}\put(77.5,60){$j$}
\put(96,77){$a$}\put(77.5,96.5){$b$}
}
\put(61,77){0}\put(78,60){0}\put(96,77){0}\put(78,96.5){0}
\put(70,0){
\put(61,77){1}\put(78,60){1}\put(96,77){1}\put(78,96.5){1}
}
\put(140,0){
\put(61,77){1}\put(78,60){0}\put(96,77){1}\put(78,96.5){0}
}
\put(210,0){
\put(61,77){0}\put(78,60){1}\put(96,77){0}\put(78,96.5){1}
}
\put(280,0){
\put(61,77){0}\put(78,60){1}\put(96,77){1}\put(78,96.5){0}
}
\put(350,0){
\put(61,77){1}\put(78,60){0}\put(96,77){0}\put(78,96.5){1}
}
\put(78,40){
\put(-77,0){$M^{ab}_{ij}$}
\put(0,0){$1$} \put(70,0){$1$} \put(134,0){$\alpha \mathbf{k}$}
\put(190,0){$-p\alpha^{-1}\mathbf{k}$} \put(278,0){$\mathbf{a}^+$} \put(350,0){$\mathbf{a}^-$}
}}
\end{picture}
\caption{The operators $M^{ab}_{ij}$, where $p=-q^{-1}$.}
\label{fig:6vM}
\end{figure}

Let us proceed to a trace construction based on $M$, which can be performed for any positive integer $N$.
We prepare the notations ($V$ is defined in the beginning of this section):
\begin{align}
\mathbf{i}& =(i_1,\ldots, i_N) \in \{0,1\}^N,\quad  |{\bf i}| = i_1+\cdots + i_N,
\label{bi}
\\
\mathbb{V}_k & = \bigoplus_{\mathbf{i} \in \{0,1\}^N,  |\mathbf{i}|=k} \C v_{\mathbf{i}},
\quad 
 v_{\mathbf{i}} = v_{i_1} \otimes \cdots \otimes v_{i_N},
\label{Vk}\\
\mathbb{V} &=V^{\otimes N} = \mathbb{V}_0 \oplus \cdots \oplus \mathbb{V}_N.
\label{Vdec}
\end{align}
The last equality is a direct consequence of the definition.
Now we introduce a matrix
$R(z) =R(z;p)\in \mathrm{End}(\mathbb{V} \otimes \mathbb{V})$ by
\begin{align}
R(z) (v_{\mathbf{i}} \otimes v_{\mathbf{j}}) & = 
\sum_{\mathbf{a},\mathbf{b},\mathbf{i},\mathbf{j} \in \{0,1\}^N}
R(z)^{{\bf a}, {\bf b}}_{{\bf i}, {\bf j}} v_{\mathbf{a}} \otimes v_{\mathbf{b}},
\\ 
R(z)^{{\bf a}, {\bf b}}_{{\bf i}, {\bf j}} &=\mathrm{Tr}_{\mathcal{V}_+}
\left(z^h M^{a_1 b_1}_{i_1 j_1}\cdots M^{a_N b_N}_{i_N j_N}\right).
\label{tred}
\end{align}
The trace is convergent for $|z| < 1$.
The element (\ref{tred}) is depicted as Figure \ref{fig:trm}.

\begin{figure}[ht]
\[
\begin{tikzpicture}
\begin{scope}[>=latex,xshift=0pt]
\draw [->] (1,1) node[below]{$j_N$}--(1,2) node[above]{$b_N$};
\draw [->] (0.5,1.3) node[left=-2pt]{$i_N$}--(1.5,1.7) node[right=-2pt]{$a_N$};
\draw (2.2,0.9) node[above]{\rotatebox{-24}{$\cdots$}};
\draw [->] (3,0.2) node[below]{$j_2$}--(3,1.2) node[above]{$b_2$};
\draw [->] (2.5,0.5) node[left=-2pt]{$i_2$}--(3.5,0.9) node[right=-2pt]{$a_2$};
\draw [->] (4.3,-0.32) node[below]{$j_1$}--(4.3,0.68) node[above]{$b_1$};
\draw [->] (3.8,-0.02) node[left=-2pt]{$i_1$}--(4.8,0.38) node[right=-2pt]{$a_1$};

{\color{green}
\draw [->] (0.5,1.7)--(4.8,-0.02)[thick];
\draw [-] (4.8,-0.02) to [out=-20,in=-90] (5.6,0.4) to [out=90,in=-20] (5.4,0.78)[thick]; 
\draw [-] (5.4,0.78)--(1.1,2.5) to [out=160,in=90] (0.3,2.08) to [out=-90,in=160] (0.5,1.7)[thick];
}

\draw[fill=blue!20] (5.5,0.15) circle[radius=0.15] node[right]{~$z^h$}; 
\end{scope}
\end{tikzpicture}
\]
\caption{Trace construction (\ref{tred}).
The diagram is a concatenation of the right one in (\ref{fig:LM}), 
where $\MM^{ab}_{ij}$ is specialized to $M^{ab}_{ij}$.
The green arrow is closed cyclically reflecting the trace.} 
\label{fig:trm}
\end{figure}
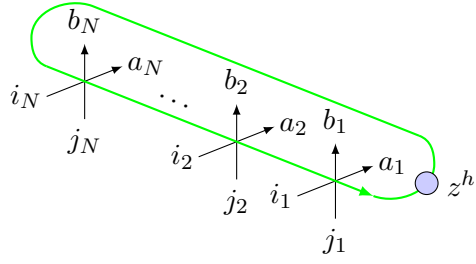

From the conservation law (\ref{cl}), one can deduce
\begin{align}\label{cons}
R(z)^{{\bf a}, {\bf b}}_{{\bf i}, {\bf j}} =0\; \text{ unless }\; 
\mathbf{a}+\mathbf{b}=\mathbf{i} + \mathbf{j}, \;
|\mathbf{a}| = |\mathbf{i}|, \;  |\mathbf{b}| = |\mathbf{j}|.
\end{align}
Thus $R(z)$ is decomposed as 
\begin{align}\label{Rz}
R(z) = \bigoplus_{0 \le k,l  \le N}R_{k,l}(z),\quad 
R_{k,l}(z) \in \mathrm{End}(\mathbb{V}_k \otimes \mathbb{V}_l).
\end{align}
From (\ref{cons}) and Figure \ref{fig:6vM}, 
it follows that the dependence of $R(z)$ on the parameter $\alpha$ 
appears solely through an overall factor $\alpha^{k-l}$ in $R_{k,l}(z)$.

 In \cite{BS06}, it was observed that, up to conventional adjustment, $R_{k,l}(z)$ coincides with 
the $U_p(\widehat{sl}_N)$ quantum $R$ matrix on $\mathbb{V}_k \otimes \mathbb{V}_l$, 
where $\mathbb{V}_1, \ldots, \mathbb{V}_{N-1}$ are 
regarded as the antisymmetric tensor representations 
(commonly referred to as Kirillov-Reshetikhin modules corresponding to 
the fundamental representations).
The precise details are elaborated in \cite[Chap.11]{K22}.
In simpler terms, the trace construction $R(z)$ serves as  a ``generating function" of the
quantum $R$ matrices for the fundamental representations, including the trivial cases where 
$k$ or $l$ is $0$ or $N$.  
An important consequence from this identification, which we will require later, is as follows.

\begin{lem}\label{le:rinv}
The $R$ matrix $R(z)$ in Figure \ref{fig:trm} is invertible for generic $z$.
\end{lem}
 
\begin{proof}
The claim reduces to the invertibility of each diagonal block $R_{k,l}(z)$.
It is a standard property that the quantum $R$ matrices satisfy the inversion relation
$PR_{l,k}(z^{-1}) PR_{k,l}(z) = \varrho_{k,l}(z)\mathrm{Id}_{\mathbb{V}_k \otimes \mathbb{V}_l}$,
where $P(u \otimes v) = v \otimes u$ denotes the transposition, and $\varrho_{k,l}(z)$ is a rational function of 
$z$ and $p$.
\end{proof}

\subsection{Tetrahedron equations and inversion relations}

Set $\mathbf{V}:=V^{\otimes 4} \otimes \mathcal{W}(q) \otimes \mathcal{W}(-q)$, 
and extend the actions of $\LL$ and $\MM$ to those on $\mathbf{V}$:
for $1 \le i < j  \le 4$ and $k=5,6$, we let $\LL_{ijk}$ and $\MM_{ijk}$ denote 
the operators acting as 
$\LL$ and $\MM$ on the $(i,j,k)$-th components of $\mathbf{V}$
and as identity elsewhere. 

\begin{prop}\label{pr:te4}
The operators $\LL$ and $\MM$ satisfy the following tetrahedron equations of four types:
\begin{itemize}
\item[(o)] Ordinary type: 
$\MM_{126} \MM_{346} \LL_{135} \LL_{245} 
= \LL_{245} \LL_{135} \MM_{346} \MM_{126}$. 
\begin{align}\label{LM-o}
\begin{tikzpicture}
\begin{scope}[>=latex,xshift=0pt]
\draw [-] (0,0) coordinate(A) to [out = 0, in = -135] (2,0.5) coordinate(B);
\draw [-] (1,1) coordinate(C) to [out = 0, in = 135] (B);
\draw [-] (C) to [out = 90, in = -45] (0.5,2) coordinate(D);
\draw [-] (A) to [out = 90, in = -135] (D);
\draw [-] (-0.5,0) node[left]{$1$} --(A); \draw [-] (0,-0.5) node[below]{$3$}--(A);
{\color{blue} 
\draw [<-] (-0.5,-0.5)--(A)--(C)--(1.5,1.5) node[above right]{$5$}[thick]; %
}
{\color{green}
\draw [<-] (2,0)--(B) to [out = 90, in = 0] (D) -- (0,2) node[left]{$6$}[thick];}
\draw [<-] (2.5,0)--(B); \draw [<-] (2.5,1)--(B);
\draw [-] (1,0.5) node[below]{$4$}--(C); \draw [-] (0.5,1) node[left]{$2$}--(C);
\draw [<-] (1,2.5)--(D); \draw [<-] (0,2.5)--(D);
\draw (3.3,1) node {$=$}; 
\end{scope}
\begin{scope}[>=latex,xshift=135pt]
\draw [-] (2,2) coordinate(A) to [out = 180, in = 45] (0,1.5) coordinate(B);
\draw [-] (1,1) coordinate(C) to [out = 180, in = -45] (B);
\draw [-] (C) to [out = -90, in = 135] (1.5,0) coordinate(D);
\draw [-] (A) to [out = -90, in = 45] (D);
\draw [<-] (2.5,2)--(A); \draw [<-] (2,2.5)--(A);
{\color{blue} 
\draw [->] (2.5,2.5) node[above right]{$5$}--(A)--(C)--(0.5,0.5) [thick]; %
}
{\color{green}
\draw [->] (0,2) node[above]{$6$}--(B) to [out = -90, in = 180] (D)--(2,0)[thick]; 
}
\draw [-] (-0.5,2) node[left]{$2$}--(B); \draw [-] (-0.5,1) node[left]{$1$}--(B);
\draw [<-] (1.5,1)--(C); \draw [<-] (1,1.5)--(C);
\draw [-] (1,-0.5) node[below left]{$3$}--(D); \draw [-] (2,-0.5) node[below right]{$4$}--(D); 
\end{scope}
\end{tikzpicture}
\end{align}

\item[(h)] Horizontally reversed type:
$[\MM_{346} \LL_{315} \LL_{425} \MM_{126}]_{\ast 6} 
= [\MM_{126} \LL_{425} \LL_{315} \MM_{346}]_{\ast 6}$ if $r'=s'$ and $g' = qf'$, 
where the subscript $\ast 6$ means that the multiplication order in $\mathcal{W}(-q)$ acting on the $6$-th component of $\mathbf{V}$ is reversed. 
\begin{align}\label{LM-h}
\begin{tikzpicture}
\begin{scope}[>=latex,xshift=0pt]
\draw [-] (0,0) coordinate(A) to [out = 0, in = -135] (2,0.5) coordinate(B);
\draw [-] (1,1) coordinate(C) to [out = 0, in = 135] (B);
\draw [-] (C) to [out = 90, in = -45] (0.5,2) coordinate(D);
\draw [-] (A) to [out = 90, in = -135] (D);
\draw [<-] (-0.5,0) node[left]{$1$} --(A); \draw [-] (0,-0.5) node[below]{$3$}--(A);
{\color{blue} 
\draw [-] (A)--(C) [thick]; %
\draw [<-] (-0.5,-0.5)--(A)[thick];%
\draw [-] (1.5,1.5) node[above right]{$5$}--(C)[thick]; %
}
{\color{green}
\draw [-] (B) to [out = 90, in = 0] (D)[thick]; 
\draw [<-] (2,0)--(B); [thick] 
\draw [-] (0,2) node[left]{$6$}--(D)[thick]; 
}
\draw [-] (2.5,0)--(B); \draw [-] (2.5,1)--(B);
\draw [-] (1,0.5) node[below]{$4$}--(C); \draw [<-] (0.5,1) node[left]{$2$}--(C);
\draw [<-] (1,2.5)--(D); \draw [<-] (0,2.5)--(D);
\draw (3.3,1) node {$=$}; 
\end{scope}
\begin{scope}[>=latex,xshift=135pt]
\draw [-] (2,2) coordinate(A) to [out = 180, in = 45] (0,1.5) coordinate(B);
\draw [-] (1,1) coordinate(C) to [out = 180, in = -45] (B);
\draw [-] (C) to [out = -90, in = 135] (1.5,0) coordinate(D);
\draw [-] (A) to [out = -90, in = 45] (D);
\draw [-] (2.5,2)--(A); \draw [<-] (2,2.5)--(A);
{\color{blue} 
\draw [-] (A)--(C) [thick]; %
\draw [-] (2.5,2.5) node[above right]{$5$} --(A) [thick];%
\draw [<-] (0.5,0.5)--(C)[thick];%
}
{\color{green}
\draw [-] (B) to [out = -90, in = 180] (D)[thick]; 
\draw [-] (0,2) node[above]{$6$}--(B); [thick] 
\draw [<-] (2,0)--(D)[thick]; 
}
\draw [<-] (-0.5,2) node[left]{$2$}--(B); \draw [<-] (-0.5,1) node[left]{$1$}--(B);
\draw [-] (1.5,1)--(C); \draw [<-] (1,1.5)--(C);
\draw [-] (1,-0.5) node[below left]{$3$}--(D); \draw [-] (2,-0.5) node[below right]{$4$}--(D); 
\end{scope}
\end{tikzpicture}
\end{align}
\item[(v)] Vertically reversed type: 
$\MM_{126} \LL_{315} \LL_{425} \MM_{346} 
= \MM_{346} \LL_{425} \LL_{315} \MM_{126}$ if $r'=s'$ and $g'=q^{-1}f'$.
\begin{align}\label{LM-v}
\begin{tikzpicture}
\begin{scope}[>=latex,xshift=0pt]
\draw [-] (0,0) coordinate(A) to [out = 0, in = -135] (2,0.5) coordinate(B);
\draw [-] (1,1) coordinate(C) to [out = 0, in = 135] (B);
\draw [-] (C) to [out = 90, in = -45] (0.5,2) coordinate(D);
\draw [-] (A) to [out = 90, in = -135] (D);
\draw [-] (-0.5,0) node[left]{$1$} --(A); \draw [<-] (0,-0.5) node[below]{$3$}--(A);
{\color{blue} 
\draw [-] (A)--(C) [thick]; %
\draw [<-] (-0.5,-0.5)--(A)[thick];%
\draw [-] (1.5,1.5) node[above right]{$5$}--(C)[thick]; %
}
{\color{green}
\draw [-] (B) to [out = 90, in = 0] (D)[thick]; 
\draw [<-] (2,0)--(B); [thick] 
\draw [-] (0,2) node[left]{$6$}--(D)[thick]; 
}
\draw [<-] (2.5,0)--(B); \draw [<-] (2.5,1)--(B);
\draw [<-] (1,0.5) node[below]{$4$}--(C); \draw [-] (0.5,1) node[left]{$2$}--(C);
\draw [-] (1,2.5)--(D); \draw [-] (0,2.5)--(D);
\draw (3.3,1) node {$=$}; 
\end{scope}
\begin{scope}[>=latex,xshift=135pt]
\draw [-] (2,2) coordinate(A) to [out = 180, in = 45] (0,1.5) coordinate(B);
\draw [-] (1,1) coordinate(C) to [out = 180, in = -45] (B);
\draw [-] (C) to [out = -90, in = 135] (1.5,0) coordinate(D);
\draw [-] (A) to [out = -90, in = 45] (D);
\draw [<-] (2.5,2)--(A); \draw [-] (2,2.5)--(A);
{\color{blue} 
\draw [-] (A)--(C) [thick]; %
\draw [-] (2.5,2.5) node[above right]{$5$} --(A) [thick];%
\draw [<-] (0.5,0.5)--(C)[thick];%
}
{\color{green}
\draw [-] (B) to [out = -90, in = 180] (D)[thick]; 
\draw [-] (0,2) node[above]{$6$}--(B); [thick] 
\draw [<-] (2,0)--(D)[thick]; 
}
\draw [-] (-0.5,2) node[left]{$2$}--(B); \draw [-] (-0.5,1) node[left]{$1$}--(B);
\draw [<-] (1.5,1)--(C); \draw [-] (1,1.5)--(C);
\draw [<-] (1,-0.5) node[below left]{$3$}--(D); \draw [<-] (2,-0.5) node[below right]{$4$}--(D); 
\end{scope}
\end{tikzpicture}
\end{align}
\item[(t)] Totally reversed type: 
$\LL_{135} \LL_{245} \MM_{126} \MM_{346} 
= \MM_{346} \MM_{126} \LL_{245} \LL_{135}$. 
\begin{align}\label{LM-t}
\begin{tikzpicture}
\begin{scope}[>=latex,xshift=0pt]
\draw [-] (0,0) coordinate(A) to [out = 0, in = -135] (2,0.5) coordinate(B);
\draw [-] (1,1) coordinate(C) to [out = 0, in = 135] (B);
\draw [-] (C) to [out = 90, in = -45] (0.5,2) coordinate(D);
\draw [-] (A) to [out = 90, in = -135] (D);
\draw [<-] (-0.5,0) node[left]{$1$} --(A); \draw [<-] (0,-0.5) node[below]{$3$}--(A);
{\color{blue} 
\draw [-] (A)--(C) [thick]; %
\draw [<-] (-0.5,-0.5)--(A)[thick];%
\draw [-] (1.5,1.5) node[above right]{$5$}--(C)[thick]; %
}
{\color{green}
\draw [-] (B) to [out = 90, in = 0] (D)[thick]; 
\draw [<-] (2,0)--(B); [thick] 
\draw [-] (0,2) node[left]{$6$}--(D)[thick]; 
}
\draw [-] (2.5,0)--(B); \draw [-] (2.5,1)--(B);
\draw [<-] (1,0.5) node[below]{$4$}--(C); \draw [<-] (0.5,1) node[left]{$2$}--(C);
\draw [-] (1,2.5)--(D); \draw [-] (0,2.5)--(D);
\draw (3.3,1) node {$=$}; 
\end{scope}
\begin{scope}[>=latex,xshift=135pt]
\draw [-] (2,2) coordinate(A) to [out = 180, in = 45] (0,1.5) coordinate(B);
\draw [-] (1,1) coordinate(C) to [out = 180, in = -45] (B);
\draw [-] (C) to [out = -90, in = 135] (1.5,0) coordinate(D);
\draw [-] (A) to [out = -90, in = 45] (D);
\draw [-] (2.5,2)--(A); \draw [-] (2,2.5)--(A);
{\color{blue} 
\draw [-] (A)--(C) [thick]; %
\draw [-] (2.5,2.5) node[above right]{$5$} --(A) [thick];%
\draw [<-] (0.5,0.5)--(C)[thick];%
}
{\color{green}
\draw [-] (B) to [out = -90, in = 180] (D)[thick]; 
\draw [-] (0,2) node[above]{$6$}--(B); [thick] 
\draw [<-] (2,0)--(D)[thick]; 
}
\draw [<-] (-0.5,2) node[left]{$2$}--(B); \draw [<-] (-0.5,1) node[left]{$1$}--(B);
\draw [-] (1.5,1)--(C); \draw [-] (1,1.5)--(C);
\draw [<-] (1,-0.5) node[below left]{$3$}--(D); \draw [<-] (2,-0.5) node[below right]{$4$}--(D); 
\end{scope}
\end{tikzpicture}
\end{align}
\end{itemize} 
\end{prop}

\begin{proof}
These are proven by direct calculation. 
We just demonstrate the calculation in a few examples for (o) and (h). 
Generators of $\mathcal{W}(q)$ (resp.~$\mathcal{W}(-q)$) acting on the $5$-th (resp.~$6$-th) component of 
$\mathbf{V}$  will be denoted by $e^{\pm \uu}, e^{\pm \ww}$ (resp. $e^{\pm \uu'}, e^{\pm \ww'}$)
as in the previous subsection. 
\\[1mm]
(o) For an in-state $\mathbf{i}=(i_1,i_2,i_3,i_4) \in V^{\otimes 4}$ and an out-state $\mathbf{j}=(j_1,j_2,j_3,j_4) \in V^{\otimes 4}$, \eqref{LM-o} is written as 
\begin{align}\label{mLM-o}
\sum_{k_1,k_2,k_3,k_4=0,1} \MM^{j_1 j_2}_{k_1 k_2} \MM^{j_3 j_4}_{k_3 k_4} 
\LL^{k_1k_3}_{i_1i_3} \LL^{k_2 k_4}_{i_2i_4}
=
\sum_{k_1,k_2,k_3,k_4=0,1} \LL^{j_2 j_4}_{k_2 k_4} \LL^{j_1 j_3}_{k_1 k_3} 
\MM^{k_3 k_4}_{i_3 i_4} \MM^{k_1 k_2}_{i_1 i_2}.
\end{align}
Let us see the case of $(\mathbf{i},\mathbf{j}) = (0,1,1,0,1,0,1,0)$. We have for the LHS of \eqref{mLM-o}
\begin{align*}
\MM^{10}_{01} \MM^{10}_{10} \LL^{01}_{01} \LL^{10}_{10} + \MM^{10}_{10} \MM^{10}_{01} \LL^{10}_{01} \LL^{01}_{10}
&= e^{-\ww'}f' e^{\uu'} g e^\uu f e^\uu + f'e^{\uu'} e^{-\ww'} e^{-\ww}(rse^\ww+fg e^{2\uu+\ww})
\\
&= f'e^{\uu'} e^{-\ww'} rs, 
\end{align*}
and for the RHS of \eqref{mLM-o}
\begin{align*}
\LL^{00}_{00} \LL^{11}_{11} \MM^{10}_{10} \MM^{10}_{01}
= rsf'e^{\uu'} e^{-\ww'}
\end{align*}
which equals to the LHS.
\\[1mm]
(h) For an in-state $\mathbf{i}=(i_1,i_2,i_3,i_4) \in V^{\otimes 4}$ and an out-state $\mathbf{j}=(j_1,j_2,j_3,j_4) \in V^{\otimes 4}$, \eqref{LM-h} is written as 
\begin{align}\label{mLM-h}
\sum_{k_1,k_2,k_3,k_4=0,1}  \MM^{k_1 k_2}_{i_1i_2} \MM^{j_3 j_4}_{k_3 k_4} \LL^{k_3 j_1}_{i_3 k_1} \LL^{k_4 j_2}_{i_4 k_2}  
=
\sum_{k_1,k_2,k_3,k_4=0,1}  \LL^{j_4 k_2}_{k_4 i_2} \LL^{j_3 k_1}_{k_3 i_1} \MM^{k_3 k_4}_{i_3 i_4} \MM^{j_1 j_2}_{k_1 k_2},
\end{align}
where we have applied the operation $[~~]_{\ast 6}$. 
For $(\mathbf{i},\mathbf{j}) = (0,1,1,0,1,0,1,0)$, the LHS of \eqref{mLM-h} becomes
\begin{align*}
\MM^{10}_{01} \MM^{10}_{10} \LL^{11}_{11} \LL^{00}_{00} + \MM^{01}_{01} \MM^{10}_{01} \LL^{01}_{10} \LL^{10}_{01}
&= e^{-\ww'} f' e^{\uu'} s r + g'e^{\uu'} e^{-\ww'} (rs e^\ww + fg e^{2\uu+\ww}) e^{-\ww} 
\\
&= sr e^{\uu'}e^{-\ww'}(g'-qf') + fg g'q^{-1}e^{\uu'}e^{-\ww'} e^{2\uu}
\end{align*}
and the RHS becomes 
\begin{align*}
\LL^{01}_{01} \LL^{10}_{10} \MM^{10}_{10} \MM^{10}_{01}
= g e^\uu f e^\uu f'e^{\uu'} e^{-\ww'}.
\end{align*}
Thus \eqref{mLM-h} reduces to $g' = q f'$.
For $(\mathbf{i},\mathbf{j}) = (0,1,0,1,0,0,1,1)$, \eqref{mLM-h} leads to
$$
\MM^{10}_{01} \MM^{11}_{11} \LL^{10}_{01} \LL^{10}_{10} = \LL^{10}_{01} \LL^{10}_{10} \MM^{10}_{01} \MM^{00}_{00} ,
$$
which reduces to $s' = r'$.
\end{proof}

\begin{remark}
One sees that the tetrahedron equation graphically interchanges two vertices corresponding to 
the operator $\LL$ for all types of (o), (h), (v) and (t) in \eqref{LM-o}--\eqref{LM-t}: 
an intersection of the $(1,3,5)$-th wires and the $(2,4,5)$-th wires are 
interchanged by moving the green arrow from northeast (NE) to southwest (SW).
We will use their two-dimensional projection diagrams as follows.
\begin{align}\label{2D-LM}
\begin{tikzpicture}
\begin{scope}[>=latex,xshift=0pt]
\draw (0,2.5) node[left] {(o)};
\draw[->] (0,1)--(2,1);
\draw[->] (1,0)--(1,2);
{\color{green}
\draw[->] (0.5,2)--(2,0.5)[thick];
}
\draw (2.5,1) node {$=$}; 
\draw[->] (3,1)--(5,1);
\draw[->] (4,0)--(4,2);
{\color{green}
\draw[->] (3,1.5)--(4.5,0)[thick];
}
\end{scope}
\begin{scope}[>=latex,xshift=200pt]
\draw (0,2.5) node[left] {(h)};
\draw[<-] (0,1)--(2,1);
\draw[->] (1,0)--(1,2);
{\color{green}
\draw[->] (0.5,2)--(2,0.5)[thick];
}
\draw (2.5,1) node {$=$}; 
\draw[<-] (3,1)--(5,1);
\draw[->] (4,0)--(4,2);
{\color{green}
\draw[->] (3,1.5)--(4.5,0)[thick];
}
\end{scope}
\begin{scope}[>=latex,yshift=-100]
\draw (0,2.5) node[left] {(v)};
\draw[->] (0,1)--(2,1);
\draw[<-] (1,0)--(1,2);
{\color{green}
\draw[->] (0.5,2)--(2,0.5)[thick];
}
\draw (2.5,1) node {$=$}; 
\draw[->] (3,1)--(5,1);
\draw[<-] (4,0)--(4,2);
{\color{green}
\draw[->] (3,1.5)--(4.5,0)[thick];
}
\end{scope}
\begin{scope}[>=latex,xshift=200pt,yshift=-100]
\draw (0,2.5) node[left] {(t)};
\draw[<-] (0,1)--(2,1);
\draw[<-] (1,0)--(1,2);
{\color{green}
\draw[->] (0.5,2)--(2,0.5)[thick];
}
\draw (2.5,1) node {$=$}; 
\draw[<-] (3,1)--(5,1);
\draw[<-] (4,0)--(4,2);
{\color{green}
\draw[->] (3,1.5)--(4.5,0)[thick];
}
\end{scope}
\end{tikzpicture}
\end{align}
\end{remark}

\begin{lem}\label{le:inv}
The operator $\MM = \mathcal{M}(r',s',f',g';q)$ \eqref{eq:M} satisfies the following inversion relations of two types:
\begin{itemize}
\item[(I)] Ordinary type: $\mathcal{M}(s',r',f',q^{-1}f';q) \mathcal{M}(s',r',f',q^{-1}f';q) = r' s' \cdot \mathrm{Id}$.
\begin{align}\label{M-I}
\begin{tikzpicture}
\begin{scope}[>=latex,xshift=0pt]
{\color{green} 
\draw [->] (-0.5,1.5) node[left]{$3$}-- (0,1.5) coordinate(A) to [out = 0, in = 90] (1.5,0) coordinate(B) -- (1.5,-0.5)[thick];
}
\draw [->] (-0.5,2) node[above left]{$2$} -- (2,-0.5);
\draw [->] (0,2) node[above]{$1$}-- (A) to [out=-90,in=180] (B) --(2,0); 
\draw (3,0.5) node {$=$}; 
\end{scope}
\begin{scope}[>=latex,xshift=135pt]
\draw [->] (-0.5,2) node[above left]{$2$}-- (2,-0.5);
\draw [->] (-0.3,2.2) node[above]{$1$}-- (2.2,-0.3);
{\color{green} 
\draw [->] (-0.7,1.8) node[left]{$3$}-- (1.8,-0.7)[thick];
}
\draw (2.3,0.5) node {$\times \,r's'$}; 
\end{scope}
\end{tikzpicture}
\end{align}

\item[(I')] Reversed type: $[\mathcal{M}(r',s',f',q f';q) \mathcal{M}(s',r',f',qf';q)]_{\ast 3} = r' s' \cdot \mathrm{Id}$.
Here the subscript $\ast 3$ means that the multiplication in $\mathcal{W}(-q)$ is reversed. 
\begin{align}\label{M-I'}
\begin{tikzpicture}
\begin{scope}[>=latex,xshift=0pt]
{\color{green} 
\draw [->] (-0.5,1.5) node[left]{$3$}-- (0,1.5) coordinate(A) to [out = 0, in = 90] (1.5,0) coordinate(B) -- (1.5,-0.5)[thick];
}
\draw [<-] (-0.5,2) node[above left]{$1$}-- (2,-0.5);
\draw [<-] (0,2) node[above]{$2$}-- (A) to [out=-90,in=180] (B) --(2,0); 
\draw (3,0.5) node {$=$}; 
\end{scope}
\begin{scope}[>=latex,xshift=135pt]
\draw [<-] (-0.5,2) node[above left]{$1$}-- (2,-0.5);
\draw [<-] (-0.3,2.2) node[above]{$2$}-- (2.2,-0.3);
{\color{green} 
\draw [->] (-0.7,1.8)node[left]{$3$} -- (1.8,-0.7)[thick];
}
\draw (2.3,0.5) node {$\times \,r's'$}; 
\end{scope}
\end{tikzpicture}
\end{align}
\end{itemize}
In other words, the operator $\MM$  satisfies the inversion relation of type (I) (resp. type (I')) if $r~=s'$ and $g'=q^{-1}f'$ (resp. $g'=q f'$).  
\end{lem}

\begin{proof}
(I) An equation $\mathcal{M}(r,s,f,g;q) \mathcal{M}(r',s',f',g';q) = r's' \cdot \mathrm{Id}$ leads to
\begin{align}
\label{inv-1}
&rr' = s s' = r's',
\\
\label{inv-2}
&g g'e^{2\uu'} + rs + fg e^{2\uu'+\ww'} e^{-\ww'} = r's',
\\
&g e^\uu (r's' e^{\ww'} + f'g' e^{2\uu'+\ww'})+ (rs e^{\ww'} + fg e^{2\uu'+\ww'})f'e^{\uu'} = 0,
\\
&g' e^{-\ww'} e^{\uu'} + f e^{\uu'} e^{-\ww'} = 0,
\\
\label{inv-5}
&e^{-\ww'}(r's' e^{\ww'} + f'g' e^{2\uu'+\ww'})+ff'e^{2\uu'} = r's',
\end{align}
and we obtain $r=s', ~s=r', ~f=q g'$ and $g=q^{-1}f'$.
This reduces to \eqref{M-I} by further setting $g' = q^{-1}f'$.
\\
(I') An equation $[\mathcal{M}(r,s,f,g;q) \mathcal{M}(r',s',f',g';q)]_{\ast 3} = r's' \cdot \mathrm{Id}$ leads to \eqref{inv-1}--\eqref{inv-5} but operators in \eqref{inv-2}--\eqref{inv-5} are in reverse order due to $[~~]_{\ast 3}$;
for example \eqref{inv-2} becomes 
$g g'e^{2\uu'} + rs + fg e^{-\ww'} e^{2\uu'+\ww'} = r's'$.
Thus we obtain $r=s', ~s=r', ~f=q^{-1} g'$ and $g=q f'$, and 
get \eqref{M-I'} by further setting $g' = q f'$.
\end{proof}

\begin{remark}
In the two-dimensional view, the operations (I) and (I') are depicted as follows. 
\begin{align}\label{2D-I}
\begin{tikzpicture}
\begin{scope}[>=latex,xshift=0pt]
\draw (0,2.5) node[left] {(I)};
\draw[->] (0,2)--(2,0);
{\color{green}
\draw[->] (0,1.5)--(0.5,1.5) to [out=0, in=90] (1.5,0.5)--(1.5,0)[thick];
}
\draw (2.5,1) node {$=$}; 
\draw[->] (3,2)--(5,0);
{\color{green}
\draw[->] (2.7,1.7)--(4.7,-0.3)[thick];
}
\end{scope}
\begin{scope}[>=latex,xshift=200pt]
\draw (0,2.5) node[left] {(I')};
\draw[<-] (0,2)--(2,0);
{\color{green}
\draw[->] (0,1.5)--(0.5,1.5) to [out=0, in=90] (1.5,0.5)--(1.5,0)[thick];
}
\draw (2.5,1) node {$=$}; 
\draw[<-] (3,2)--(5,0);
{\color{green}
\draw[->] (2.7,1.7)--(4.7,-0.3)[thick];
}
\end{scope}
\end{tikzpicture}
\end{align}
We also have the inverses (iI) of (I) and the inverse (iI') of (I') depicted as follows.
\begin{align}\label{2D-iI}
\begin{tikzpicture}
\begin{scope}[>=latex,xshift=0pt]
\draw (0,2.5) node[left] {(iI)};
\draw[->] (0,2)--(2,0);
{\color{green}
\draw[->] (0.3,2.3)--(2.3,0.3)[thick];
}
\draw (2.5,1) node {$=$}; 
\draw[->] (3,2)--(5,0);
{\color{green}
\draw[->] (3.5,2)--(3.5,1.5) to [out=-90, in=180] (4.5,0.5)--(5,0.5)[thick];
}
\end{scope}
\begin{scope}[>=latex,xshift=200pt]
\draw (0,2.5) node[left] {(iI')};
\draw[<-] (0,2)--(2,0);
{\color{green}
\draw[->] (0.3,2.3)--(2.3,0.3)[thick];
}
\draw (2.5,1) node {$=$}; 
\draw[<-] (3,2)--(5,0);
{\color{green}
\draw[->] (3.5,2)--(3.5,1.5) to [out=-90, in=180] (4.5,0.5)--(5,0.5)[thick];
}
\end{scope}
\end{tikzpicture}
\end{align}

\end{remark}

\subsection{$RLLL$ relation}\label{sec:RLLL}

We briefly recall the $RLLL$ relation \eqref{RLLL} for the q-6v model studied in \cite[\S 7]{IKSTY}. See \cite{IKSTY} for the detail.
For three canonical pairs $\uu_i$ and $\ww_i$, and parameters $\mathscr{P}_i:=(a_i, b_i, c_i, d_i, e_i) \in \C^5$ satisfying $a_i+b_i+c_i+d_i+e_i=0$ for $i=1,2,3$, the $R$-operator $\mathscr{R}=\mathscr{R}(\mathscr{P}_1 ,\mathscr{P}_2 ,\mathscr{P}_3)$ is defined as
\begin{align}
\label{r123}
\mathscr{R} 
&= 
\Psi_q(\e^{-d_1-c_2-b_3+u_1+u_3+w_1-w_2+w_3})^{-1}
\Psi_q(\e^{-d_1-c_2-b_3-e_3+u_1-u_3+w_1-w_2+w_3})^{-1}
\nonumber\\
&\quad \times P_{123} 
\Psi_q(\e^{-b_1-a_2-d_3-e_3+u_1-u_3+w_1-w_2+w_3})
\Psi_q(\e^{-b_1-a_2-d_3+u_1+u_3+w_1-w_2+w_3}),
\\
P_{123} &=  
\e^{\tfrac{1}{\hbar}(u_3-u_2)w_1}
\e^{\tfrac{\lambda_0}{\hbar}(-w_1-w_2+w_3)}
\e^{\tfrac{1}{\hbar}(\lambda_1u_1+\lambda_2u_2+\lambda_3u_3)}\rho_{23}
\label{pijk} 
\end{align}
Here $\Psi_q$ is the quantum dilogarithm, 
\begin{align}\label{Psiq}
\Psi_q(Y) = \frac{1}{(-qY; q^2)_\infty}, \quad (z;q)_\infty = \prod_{n=0}^\infty (1-zq^n),
\end{align}
$\rho_{23}$ is the permutation of the labels of the canonical pairs, and $\lambda_r=\lambda_r(\mathscr{P}_1,\mathscr{P}_2,\mathscr{P}_3)$
for $r=0$, $1$, $2$, $3$ are defined by 
\begin{align}\label{lad}
\lambda_0 = \frac{e_2-e_3}{2},\quad 
\lambda_1 = a_2-a_3+b_2-b_3 + \lambda_0,\quad
\lambda_2 = -a_1-b_2+b_3-\lambda_0,\quad
\lambda_3 = c_1-c_2+c_3.
\end{align}

\begin{thm}\cite[Theorem 7.1]{IKSTY}
\label{th:RLLL}
The $R$-operator \eqref{r123} and the operator $\mathscr{L}$ \eqref{L1} satisfy  the $RLLL$ relation \eqref{RLLL}, where we set $\mathscr{R}_{456} = \mathscr{R}$, $\mathscr{L}_{124} = \mathscr{L}(r_1,s_1,f_1,g_1;q)$,
$\mathscr{L}_{135}=\mathscr{L}(r_2,s_2,f_2,g_2;q)$ and $\mathscr{L}_{236}=\mathscr{L}(r_3,s_3,f_3,g_3;q)$ on the both sides of \eqref{RLLL}, and identify the parameters as
  \begin{equation}
    r_i = \e^{c_i}, \quad
    s_i = \e^{a_i}, \quad
    f_i = \e^{-b_i}, \quad
    g_i = \e^{-d_i},
  \end{equation}
for $i=1,2,3$.
For $\mathscr{L}(r_i,s_i,f_i,g_i;q)$, we use the $q$-Weyl pair defined by a canonical pair $\uu_i$ and $\ww_i$ satisfying $[\uu_i,\ww_i] = \hbar$, which generates the $i$-th component of $\mathcal{W}(q)^{\otimes 3}$. 
\end{thm}

This $RLLL$ relation has the representation with wiring diagrams: 
\begin{align}\label{RLLL-dgm}
\begin{tikzpicture}
\begin{scope}[>=latex,xshift=0pt]
\draw (-0.5,1) node[left]{$\displaystyle{\sum_{\alpha,\beta,\gamma=0,1}} \mathscr{R} ~\circ$};
\draw [-] (-0.3,2.3)--(-0.5,2.3)--(-0.5,-0.3) --(-0.3,-0.3);
\draw [-] (4.3,2.3)--(4.5,2.3)--(4.5,-0.3) --(4.3,-0.3);
\draw [->] (0,0) node[left]{$k$} to [out = 0, in = -135] (2,0.5) coordinate(B) node[below=2pt]{$2$}-- (3,1.5) coordinate(A) node[above=2pt]{$3$}to [out = 45, in = 180] (4,2) node[right]{$c$};
\draw [->] (0,1) node[left]{$j$} to [out = 0, in = -135] (1,1.5) coordinate(C) node[above=2pt]{$1$} to [out = 45, in = 135] (A) to [out = -45, in = 180] (4,1) node[right]{$b$};
\draw [->] (0,2) node[left]{$i$} to [out = 0, in = 135] (C) -- (B) to [out = -45, in = 180] (4,0) node[right]{$a$};
\draw (2.5,1) node[right]{$\gamma$}; 
\draw (2,1.9) node[above]{$\beta$}; 
\draw (1.5,1) node[left]{$\alpha$};
{\color{blue}
\draw [->] (2.4,0.6)--(1.6,0.4);
\draw [->] (1.4,1.6)--(0.6,1.4);
\draw [->] (3.4,1.6)--(2.6,1.4);
}
\draw (5,1) node{$=$};
\end{scope}
\begin{scope}[>=latex,xshift=200pt]
\draw (-0.5,1) node[left]{$\displaystyle{\sum_{\alpha,\beta,\gamma=0,1}}$};
\draw [-] (-0.3,2.3)--(-0.5,2.3)--(-0.5,-0.3) --(-0.3,-0.3);
\draw [-] (4.3,2.3)--(4.5,2.3)--(4.5,-0.3) --(4.3,-0.3);
\draw [->] (0,2) node[left]{$i$} to [out = 0, in = 135] (2,1.5) coordinate(B) node[above=2pt]{$2$} -- (3,0.5) coordinate(C) node[below=2pt]{$1$} to [out = -45, in = 180] (4,0) node[right]{$a$};
\draw [->] (0,1) node[left]{$j$} to [out = 0, in = 135] (1,0.5) coordinate(A) node[below=2pt]{$3$} to [out = -45, in = -135] (C) to [out = 45, in = 180] (4,1) node[right]{$b$};
\draw [->] (0,0) node[left]{$k$} to [out = 0, in = -135] (A) -- (B) to [out = 45, in = 180] (4,2) node[right]{$c$};
\draw (1.5,1) node[left]{$\gamma$}; 
\draw (2,0.2) node[below]{$\beta$}; 
\draw (2.5,1) node[right]{$\alpha$};
{\color{blue}
\draw [->] (2.4,1.6)--(1.6,1.4);
\draw [->] (1.4,0.6)--(0.6,0.4);
\draw [->] (3.4,0.6)--(2.6,0.4);
}
\draw (4.5,1) node[right]{$\circ ~\mathscr{R}$};
\end{scope}
\end{tikzpicture}
\end{align}
for arbitrary $a$, $b$, $c$, $i$, $j$, $k \in \{0,1\}$.
Equivalently, in terms of the components $\mathscr{L}^{ab}_{ij}$ of $\mathscr{L}$ it is written as 
\begin{equation}\label{qybe}
  \mathscr{R} \sum_{\alpha, \beta, \gamma=0,1}(\mathscr{L}^{\alpha \beta}_{ij} 
  \otimes \mathscr{L}^{a \gamma}_{\alpha k} \otimes \mathscr{L}^{bc}_{\beta\gamma})  
  = \sum_{\alpha, \beta, \gamma=0,1}
  (\mathscr{L}^{ab}_{\alpha\beta} \otimes 
  \mathscr{L}^{\alpha c}_{i\gamma} \otimes \mathscr{L}^{\beta\gamma}_{jk}) \, \mathscr{R}.
\end{equation}

\section{Commuting transfer matrices}

\subsection{Ordinary case}

For positive integers $m$ and $n$, consider a square grid $G = G_{m,n}$ 
on a torus given by $m$ horizontal and $n$ vertical directed closed wires as Figure \ref{fig:Gmn}.
We assign the operator $\LL$ \eqref{L1} on each vertex of $G$.

\begin{figure}[ht]
\[
\begin{tikzpicture}
\begin{scope}[>=latex,xshift=0pt]
\draw [->] (-0.5,0)--(4.5,0);
\draw [->] (-0.5,2)--(4.5,2);
\draw [->] (-0.5,3)--(4.5,3);
\draw [->] (0,-0.5)--(0,3.5);
\draw [->] (1,-0.5)--(1,3.5);
\draw [->] (2,-0.5)--(2,3.5);
\draw [->] (4,-0.5)--(4,3.5);
{\color{red}
\draw [-] (-0.5,3.5)--(4.5,3.5);
\draw [-] (-0.5,-0.5)--(4.5,-0.5);
\draw [-] (-0.5,-0.5)--(-0.5,3.5);
\draw [-] (4.5,-0.5)--(4.5,3.5);
}
\draw [-] (-0.7,3.3) to [out=200,in=90] (-0.8,3.0)--(-0.8,1.7) to [out=-90,in=0] (-0.9,1.5) node[left] {$m$} to [out=0,in=90] (-0.8,1.3) -- (-0.8,0) to [out=-90,in=180] (-0.7,-0.2);
\draw [-] (-0.2,3.7) to [out=90,in=0] (0,3.8)--(1.8,3.8) to [out=0,in=-90] (2,3.9) node[above] {$n$} to [out=-90,in=180] (2.2,3.8) -- (4,3.8) to [out=0,in=90] (4.2,3.7); 
\end{scope}
\begin{scope}[>=latex,xshift=190pt]
\draw [->] (-0.5,0)--(4.5,0);
\draw [->] (-0.5,2)--(4.5,2);
\draw [->] (-0.5,3)--(4.5,3);
\draw [->] (0,-0.5)--(0,3.5);
\draw [->] (1,-0.5)--(1,3.5);
\draw [->] (2,-0.5)--(2,3.5);
\draw [->] (4,-0.5)--(4,3.5);
{\color{blue} 
\draw [<-] (-0.4,-0.3) -- (0.4,0.3)[thick];
\draw [<-] (0.6,-0.3) -- (1.4,0.3)[thick];
\draw [<-] (1.6,-0.3) -- (2.4,0.3)[thick];
\draw [<-] (3.6,-0.3) -- (4.4,0.3)[thick];
\draw [<-] (-0.4,1.7) -- (0.4,2.3)[thick];
\draw [<-] (0.6,1.7) -- (1.4,2.3)[thick];
\draw [<-] (1.6,1.7) -- (2.4,2.3)[thick];
\draw [<-] (3.6,1.7) -- (4.4,2.3)[thick];
\draw [<-] (-0.4,2.7) -- (0.4,3.3)[thick];
\draw [<-] (0.6,2.7) -- (1.4,3.3)[thick];
\draw [<-] (1.6,2.7) -- (2.4,3.3)[thick];
\draw [<-] (3.6,2.7) -- (4.4,3.3)[thick];
}
{\color{red}
\draw [-] (-0.5,3.5)--(4.5,3.5);
\draw [-] (-0.5,-0.5)--(4.5,-0.5);
\draw [-] (-0.5,-0.5)--(-0.5,3.5);
\draw [-] (4.5,-0.5)--(4.5,3.5);
}
\end{scope}
\end{tikzpicture}
\]
\caption{The square grid $G_{m,n}$ (left), and the corresponding vertex model (right). 
The red rectangle is a fundamental domain of a torus.}
\label{fig:Gmn}
\end{figure}
%
Let $\mathcal{W}_{mn}(q)$ be a tensor product of $mn$ copies of the $q$-Weyl algebra $\mathcal{W}(q)$. 
Similarly to (\ref{bi}),  we write 
$v_{\mathbf{i}}=v_{i_1} \otimes \cdots \otimes v_{i_m} \in V^{\otimes m}$ for 
$\mathbf{i}=(i_1,i_2,\ldots,i_m) \in \{0,1\}^m$,
and 
$v_{\mathbf{j}} = v_{j_1}  \otimes \cdots \otimes v_{j_n} \in V^{\otimes n}$
 for $\mathbf{j}=(j_1,j_2,\ldots,j_n) \in \{0,1\}^n$, 
 and set $|\mathbf{i}| := \sum_{k=1,\ldots,m} i_k$ and  $|\mathbf{j}| := \sum_{k=1,\ldots,n} j_k$.     
We fix a fundamental domain of a torus, and define a monodromy matrix 
$\mathcal{T}(x,y) \in \mathrm{End}(V^{\otimes m} \otimes V^{\otimes n}) \otimes \mathcal{W}_{mn}(q)$ on $G$ by 
\begin{align}
\mathcal{T}(x,y)(v_{\mathbf{i}} \otimes v_{\mathbf{j}})
= \sum_{\mathbf{a} \in \{0,1\}^m, \mathbf{b} \in \{0,1\}^n}
x^{|\mathbf{a}|} y^{|\mathbf{b}|}
v_{\mathbf{a}} \otimes v_{\mathbf{b}} \otimes 
T^{\mathbf{a},\mathbf{b}}_{\mathbf{i},\mathbf{j}}.
\end{align}
Here $x,y$ are spectral parameters and 
$T^{\mathbf{a},\mathbf{b}}_{\mathbf{i},\mathbf{j}} \in \mathcal{W}_{mn}(q)$ is 
graphically defined as
\begin{align}
\begin{tikzpicture}
\begin{scope}[>=latex,xshift=0pt]
\draw (-3,1.5) node{$\displaystyle{T^{\mathbf{a},\mathbf{b}}_{\mathbf{i},\mathbf{j}} 
= \sum_{\{0,1\}^{\text{inner edges}}}}$};
\draw [->] (-0.5,0) node[left]{$i_m$}--(4.5,0) node[right]{$a_m$};
\draw [->] (-0.5,2) node[left]{$i_2$}--(4.5,2) node[right]{$a_2$};
\draw [->] (-0.5,3) node[left]{$i_1$}--(4.5,3) node[right]{$a_1$};
\draw [->] (0,-0.5) node[below]{$j_1$}--(0,3.5) node[above]{$b_1$};
\draw [->] (1,-0.5) node[below]{$j_2$}--(1,3.5) node[above]{$b_2$};
\draw [->] (2,-0.5) node[below]{$j_3$}--(2,3.5) node[above]{$b_3$};
\draw [->] (4,-0.5) node[below]{$j_n$}--(4,3.5) node[above]{$b_n$};
{\color{blue} 
\draw [<-] (-0.4,-0.3) -- (0.4,0.3)[thick];
\draw [<-] (0.6,-0.3) -- (1.4,0.3)[thick];
\draw [<-] (1.6,-0.3) -- (2.4,0.3)[thick];
\draw [<-] (3.6,-0.3) -- (4.4,0.3)[thick];
\draw [<-] (-0.4,1.7) -- (0.4,2.3)[thick];
\draw [<-] (0.6,1.7) -- (1.4,2.3)[thick];
\draw [<-] (1.6,1.7) -- (2.4,2.3)[thick];
\draw [<-] (3.6,1.7) -- (4.4,2.3)[thick];
\draw [<-] (-0.4,2.7) -- (0.4,3.3)[thick];
\draw [<-] (0.6,2.7) -- (1.4,3.3)[thick];
\draw [<-] (1.6,2.7) -- (2.4,3.3)[thick];
\draw [<-] (3.6,2.7) -- (4.4,3.3)[thick];
}
\end{scope}
\end{tikzpicture}.
\end{align}
The monodromy matrix $\mathcal{T}(x,y)$ itself is depicted as (see (\ref{h}) for the definition of $h$)
\begin{figure}[ht]
\[
\begin{tikzpicture}
\begin{scope}[>=latex,xshift=0pt]
\draw [<-] (3.5,0)--(-0.5,0); 
\draw[fill=blue!20] (3.65,0) circle[radius=0.15] node[right=2pt] {$x^{h}$}; 
\draw [<-] (3.5,2)--(-0.5,2);
\draw[fill=blue!20] (3.65,2) circle[radius=0.15] node[right=2pt] {$x^{h}$}; 
\draw [<-] (0,2.5)--(0,-0.5);
\draw[fill=blue!20] (0,2.65) circle[radius=0.15] node[above=2pt] {$y^{h}$}; 
\draw [<-] (3,2.5)--(3,-0.5);
\draw[fill=blue!20] (3,2.65) circle[radius=0.15] node[above=2pt] {$y^{h}$}; 
\draw (1.6,2.15) node {$\cdots$};
\draw (1.6,0.15) node {$\cdots$};
\draw (0.2,1.15) node {$\vdots$};
\draw (3.2,1.15) node {$\vdots$};
{\color{blue} 
\draw [<-] (-0.4,-0.3) -- (0.4,0.3)[thick];
\draw [<-] (2.6,-0.3) -- (3.4,0.3)[thick];
\draw [<-] (-0.4,1.7) -- (0.4,2.3)[thick];
\draw [<-] (2.6,1.7) -- (3.4,2.3)[thick];
}
\end{scope}
\end{tikzpicture}
.
\]
\end{figure}

We call the part $V^{\otimes m}\otimes V^{\otimes n}$ the {\em auxiliary space} 
of the monodromy matrix $\mathcal{T}(x,y)$.
Now the transfer matrix $T_G(x,y)$, the main object of our study, is defined   
by taking the trace of the monodromy matrix over the auxiliary space:
\begin{align}\label{T-G}
T_G(x,y) &= \mathrm{Tr}_{V^{\otimes m}\otimes V^{\otimes n}}\left(\mathcal{T}(x,y)\right)
= \sum_{\mathbf{i} \in \{0,1\}^m, \mathbf{j} \in \{0,1\}^n} 
T^{\mathbf{i},\mathbf{j}}_{\mathbf{i}, \mathbf{j}} \,
x^{|\mathbf{i}|} y^{|\mathbf{j}|} \in \mathcal{W}_{mn}(q).
\end{align}

One may interpret $T_G(x,y)$ as a $\mathcal{W}_{nm}(q)$-valued partition function 
of the quantized six-vertex model with periodic boundary condition 
twisted by ``boundary magnetic filed" $x^h$ on the horizontal edges and 
$y^h$ on the vertical edges.   
Our first result is the following theorem, which is proved by applying one of the tetrahedron equation (o) \eqref{LM-o}. 
\begin{thm}\label{thm:TT1}
For the square grid $G$ on a torus, the transfer matrices form a 
two parameter commuting family, i.e., the following holds for any $x,y,u,w$:
\begin{align}\label{G-TT}
[ \, T_G(x,y), \,T_G(u,w)\,] = 0. 
\end{align}
\end{thm}

\begin{proof}
Consider the compositions of $2mn$ $\LL$'s \eqref{L1} and $(m+n)$ $\MM$'s  \eqref{eq:M} as shown in the left picture below.
Starting from the top right corner, one can apply the tetrahedron equation of ordinary type \eqref{LM-o}  
successively, $nm$ times in total, to transform the green arrow into L shape, thereby 
changing the diagram into the right one as follows:
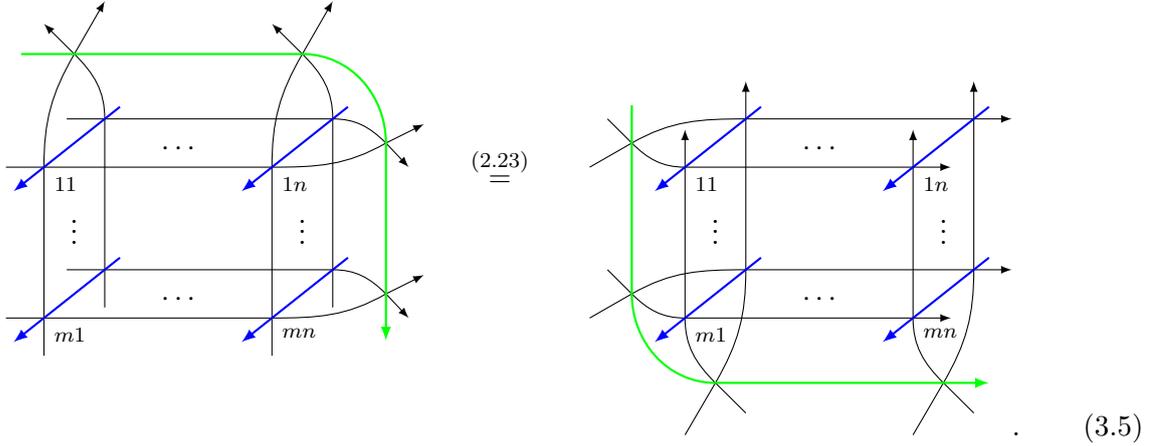
\begin{align}
\begin{tikzpicture}
\begin{scope}[>=latex,xshift=0pt]
\draw [->] (-0.5,0) --(3,0) to [out=0, in=-150] (4.5,0.32) --(5,0.57); 
\draw [->] (-0.5,2) -- (3,2) to [out=0, in=-150] (4.5,2.32) --(5,2.57);
\draw [->] (0,-0.5) --(0,2) to [out=90, in=-120](0.4,3.5)--(0.8,4.2);
\draw [->] (3,-0.5) --(3,2) to [out=90, in=-120](3.4,3.5)--(3.8,4.2);
\draw (1.8,2.25) node {$\cdots$};
\draw (1.8,0.25) node {$\cdots$};
\draw (0.4,1.25) node {$\vdots$};
\draw (3.4,1.25) node {$\vdots$};
\draw (0,2) node[below right] {\scriptsize $1 1$};
\draw (0,0) node[below right] {\scriptsize $m 1$};
\draw (3,2) node[below right] {\scriptsize $1 n$};
\draw (3,0) node[below right] {\scriptsize $m n$};
\draw [->] (0.3,0.64) --(3.8,0.64) to [out=0, in=135] (4.5,0.32) -- (4.8,0);
\draw [->] (0.3,2.64) --(3.8,2.64) to [out=0, in=135] (4.5,2.32) -- (4.8,2);
\draw [->] (0.8,0.14) --(0.8,2.64) to [out=90, in=-45](0.4,3.5) --(0,3.9);
\draw [->] (3.8,0.14) --(3.8,2.64) to [out=90, in=-45] (3.4,3.5) --(3,3.9);
{\color{blue} 
\draw [<-] (-0.4,-0.32) -- (1,0.8)[thick];
\draw [<-] (2.6,-0.32) -- (4,0.8)[thick];
\draw [<-] (-0.4,1.68) -- (1,2.8)[thick];
\draw [<-] (2.6,1.68) -- (4,2.8)[thick];
}
{\color{green}
\draw [->] (-0.3,3.5)--(3.4,3.5) to [out=0,in=90] (4.5,2.32)--(4.5,-0.3)[thick];}
\draw (6,2) node{$\stackrel{\eqref{LM-o}}{=}$};
\end{scope}
\begin{scope}[>=latex,xshift=240pt]
\draw [<-] (3.5,0)--(0,0) to [out=180, in=-45] (-0.7,0.32) --(-1.02,0.64); 
\draw [<-] (3.5,2)--(0,2) to [out=180, in=-45] (-0.7,2.32) --(-1.02,2.64);
\draw [<-] (0,2.5)--(0,0) to [out=-90, in=135] (0.4,-0.86)--(0.8,-1.26);
\draw [<-] (3,2.5)--(3,0) to [out=-90, in=135] (3.4,-0.86)--(3.8,-1.26);
\draw (1.8,2.25) node {$\cdots$};
\draw (1.8,0.25) node {$\cdots$};
\draw (0.4,1.25) node {$\vdots$};
\draw (3.4,1.25) node {$\vdots$};
\draw (0,2) node[below right] {\scriptsize $1 1$};
\draw (0,0) node[below right] {\scriptsize $m 1$};
\draw (3,2) node[below right] {\scriptsize $1 n$};
\draw (3,0) node[below right] {\scriptsize $m n$};
%
\draw [<-] (4.3,0.64)--(0.8,0.64) to [out=180, in=30] (-0.7,0.32) -- (-1.254,0);\draw [<-] (4.3,2.64)--(0.8,2.64) to [out=180, in=30] (-0.7,2.32) -- (-1.254,2);
\draw [<-] (0.8,3.14)--(0.8,0.64) to [out=-90, in=60] (0.4,-0.86) --(0,-1.55);
\draw [<-] (3.8,3.14)--(3.8,0.64) to [out=-90, in=60] (3.4,-0.86) --(3,-1.55);
{\color{blue} 
\draw [<-] (-0.4,-0.32) -- (1,0.8)[thick];
\draw [<-] (2.6,-0.32) -- (4,0.8)[thick];
\draw [<-] (-0.4,1.68) -- (1,2.8)[thick];
\draw [<-] (2.6,1.68) -- (4,2.8)[thick];
}
{\color{green}
\draw [->] (-0.7,2.82)--(-0.7,0.32) to [out=-90,in=180] (0.4,-0.86)--(4,-0.86)[thick];
}
\end{scope}
\end{tikzpicture}
.
\label{grw}
\end{align}
For the spectral parameters $x,y,u,w$, 
apply the operator $x^{h}, y^{h}, u^{h}, w^{h},$ to the top and the right boundaries of the above picture.
On the LHS,  we further apply the equality (\ref{fig:Lh}).
This leads to the relation  shown in the following diagram.
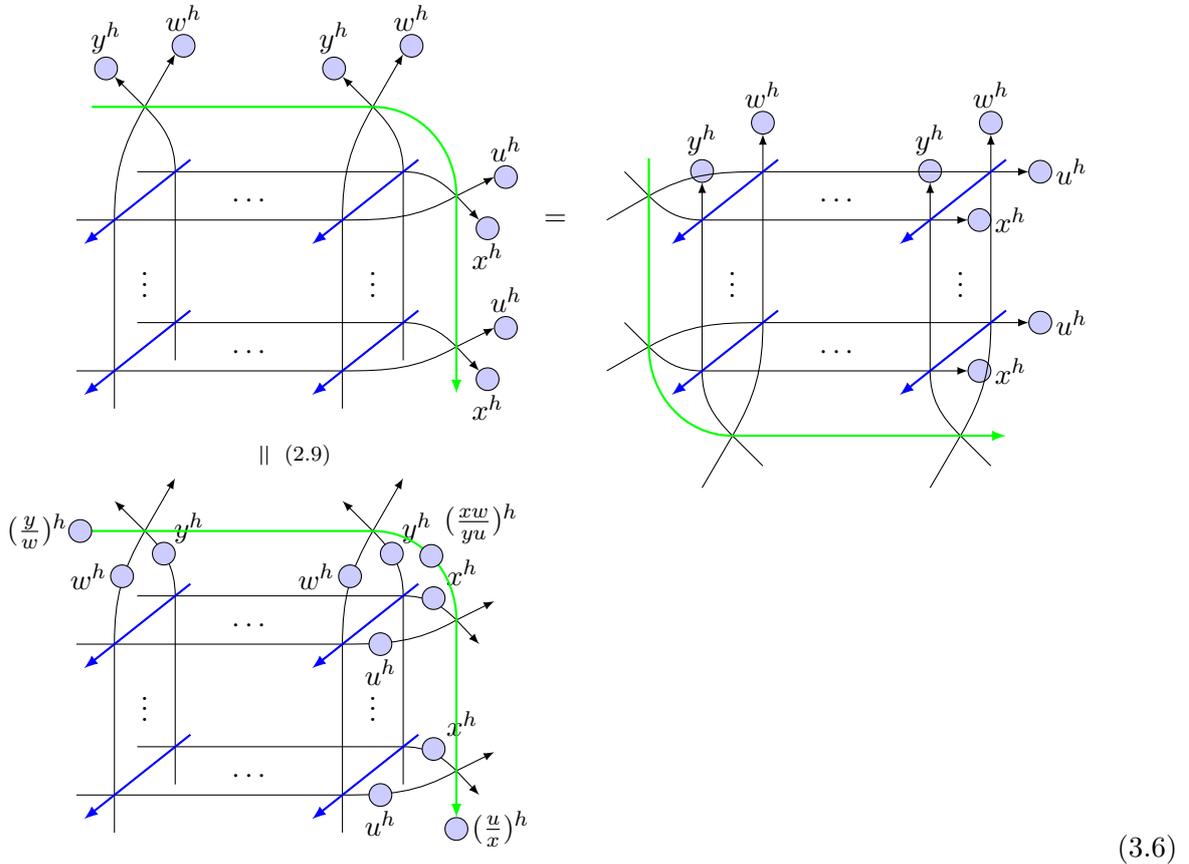
\begin{align}
\begin{tikzpicture}
\begin{scope}[>=latex,xshift=0pt]
\draw [->] (-0.5,0) --(3,0) to [out=0, in=-150] (4.5,0.32) --(5,0.57); 
\draw[fill=blue!20] (5.15,0.57) circle[radius=0.15] node[above=2pt] {$u^{h}$}; 
\draw [->] (-0.5,2) -- (3,2) to [out=0, in=-150] (4.5,2.32) --(5,2.57);
\draw[fill=blue!20] (5.15,2.57) circle[radius=0.15] node[above=2pt] {$u^{h}$}; 
\draw [->] (0,-0.5) --(0,2) to [out=90, in=-120](0.4,3.5)--(0.8,4.2);
\draw[fill=blue!20] (0.91,4.31) circle[radius=0.15] node[above=2pt] {$w^{h}$}; 
\draw [->] (3,-0.5) --(3,2) to [out=90, in=-120](3.4,3.5)--(3.8,4.2);
\draw[fill=blue!20] (3.91,4.31) circle[radius=0.15] node[above=2pt] {$w^{h}$}; 
\draw (1.8,2.25) node {$\cdots$};
\draw (1.8,0.25) node {$\cdots$};
\draw (0.4,1.25) node {$\vdots$};
\draw (3.4,1.25) node {$\vdots$};
\draw [->] (0.3,0.64) --(3.8,0.64) to [out=0, in=135] (4.5,0.32) -- (4.8,0);
\draw[fill=blue!20] (4.91,-0.11) circle[radius=0.15] node[below=2pt] {$x^{h}$}; 
\draw [->] (0.3,2.64) --(3.8,2.64) to [out=0, in=135] (4.5,2.32) -- (4.8,2);
\draw[fill=blue!20] (4.91,1.89) circle[radius=0.15] node[below=2pt] {$x^{h}$}; 
\draw [->] (0.8,0.14) --(0.8,2.64) to [out=90, in=-45](0.4,3.5) --(0,3.9);
\draw[fill=blue!20] (-0.11,4.01) circle[radius=0.15] node[above=2pt] {$y^{h}$}; 
\draw [->] (3.8,0.14) --(3.8,2.64) to [out=90, in=-45] (3.4,3.5) --(3,3.9);
\draw[fill=blue!20] (2.89,4.01) circle[radius=0.15] node[above=2pt] {$y^{h}$}; 
{\color{blue} 
\draw [<-] (-0.4,-0.32) -- (1,0.8)[thick];
\draw [<-] (2.6,-0.32) -- (4,0.8)[thick];
\draw [<-] (-0.4,1.68) -- (1,2.8)[thick];
\draw [<-] (2.6,1.68) -- (4,2.8)[thick];
}
{\color{green}
\draw [->] (-0.3,3.5)--(3.4,3.5) to [out=0,in=90] (4.5,2.32)--(4.5,-0.3)[thick];}
\draw (5.8,2) node{$=$};
\end{scope}
\begin{scope}[>=latex,xshift=220pt]
\draw [<-] (3.5,0)--(0,0) to [out=180, in=-45] (-0.7,0.32) --(-1.02,0.64); 
\draw[fill=blue!20] (3.65,0) circle[radius=0.15] node[right=2pt] {$x^{h}$}; 
\draw [<-] (3.5,2)--(0,2) to [out=180, in=-45] (-0.7,2.32) --(-1.02,2.64);
\draw[fill=blue!20] (3.65,2) circle[radius=0.15] node[right=2pt] {$x^{h}$}; 
\draw [<-] (0,2.5)--(0,0) to [out=-90, in=135] (0.4,-0.86)--(0.8,-1.26);
\draw[fill=blue!20] (0,2.65) circle[radius=0.15] node[above=2pt] {$y^{h}$}; 
\draw [<-] (3,2.5)--(3,0) to [out=-90, in=135] (3.4,-0.86)--(3.8,-1.26);
\draw[fill=blue!20] (3,2.65) circle[radius=0.15] node[above=2pt] {$y^{h}$}; 
\draw (1.8,2.25) node {$\cdots$};
\draw (1.8,0.25) node {$\cdots$};
\draw (0.4,1.25) node {$\vdots$};
\draw (3.4,1.25) node {$\vdots$};
\draw [<-] (4.3,0.64)--(0.8,0.64) to [out=180, in=30] (-0.7,0.32) -- (-1.254,0);\draw[fill=blue!20] (4.45,0.64) circle[radius=0.15] node[right=2pt] {$u^{h}$}; 
\draw [<-] (4.3,2.64)--(0.8,2.64) to [out=180, in=30] (-0.7,2.32) -- (-1.254,2);\draw[fill=blue!20] (4.45,2.64) circle[radius=0.15] node[right=2pt] {$u^{h}$}; 
\draw [<-] (0.8,3.14)--(0.8,0.64) to [out=-90, in=60] (0.4,-0.86) --(0,-1.55);
\draw[fill=blue!20] (0.8,3.29) circle[radius=0.15] node[above=2pt] {$w^{h}$}; 
\draw [<-] (3.8,3.14)--(3.8,0.64) to [out=-90, in=60] (3.4,-0.86) --(3,-1.55);
\draw[fill=blue!20] (3.8,3.29) circle[radius=0.15] node[above=2pt] {$w^{h}$}; 
{\color{blue} 
\draw [<-] (-0.4,-0.32) -- (1,0.8)[thick];
\draw [<-] (2.6,-0.32) -- (4,0.8)[thick];
\draw [<-] (-0.4,1.68) -- (1,2.8)[thick];
\draw [<-] (2.6,1.68) -- (4,2.8)[thick];
}
{\color{green}
\draw [->] (-0.7,2.82)--(-0.7,0.32) to [out=-90,in=180] (0.4,-0.86)--(4,-0.86)[thick];
}
\end{scope}
\begin{scope}[>=latex,yshift=-160pt]
\draw (2,4.5) node{\rotatebox{90}{$=$}};
\draw (2.1,4.5) node[right] {\scriptsize \eqref{fig:Lh}};
\draw [->] (-0.5,0) --(3,0) to [out=0, in=-150] (4.5,0.32) --(5,0.57); 
\draw[fill=blue!20] (3.5,0) circle[radius=0.15] node[below=2pt] {$u^{h}$}; 
\draw [->] (-0.5,2) -- (3,2) to [out=0, in=-150] (4.5,2.32) --(5,2.57);
\draw[fill=blue!20] (3.5,2) circle[radius=0.15] node[below=2pt] {$u^{h}$}; 
\draw [->] (0,-0.5) --(0,2) to [out=90, in=-120](0.4,3.5)--(0.8,4.2);
\draw[fill=blue!20] (0.1,2.9) circle[radius=0.15] node[left=2pt] {$w^{h}$}; 
\draw [->] (3,-0.5) --(3,2) to [out=90, in=-120](3.4,3.5)--(3.8,4.2);
\draw[fill=blue!20] (3.1,2.9) circle[radius=0.15] node[left=2pt] {$w^{h}$}; 
\draw (1.8,2.25) node {$\cdots$};
\draw (1.8,0.25) node {$\cdots$};
\draw (0.4,1.25) node {$\vdots$};
\draw (3.4,1.25) node {$\vdots$};
\draw [->] (0.3,0.64) --(3.8,0.64) to [out=0, in=135] (4.5,0.32) -- (4.8,0);
\draw[fill=blue!20] (4.2,0.61) circle[radius=0.15] node[above right=1pt] {$x^{h}$}; 
\draw [->] (0.3,2.64) --(3.8,2.64) to [out=0, in=135] (4.5,2.32) -- (4.8,2);
\draw[fill=blue!20] (4.2,2.61) circle[radius=0.15] node[above right=1pt] {$x^{h}$}; 
\draw [->] (0.8,0.14) --(0.8,2.64) to [out=90, in=-45](0.4,3.5) --(0,3.9);
\draw[fill=blue!20] (0.65,3.2) circle[radius=0.15] node[above right] {$y^{h}$}; 
\draw [->] (3.8,0.14) --(3.8,2.64) to [out=90, in=-45] (3.4,3.5) --(3,3.9);
\draw[fill=blue!20] (3.65,3.2) circle[radius=0.15] node[above right] {$y^{h}$}; {\color{blue} 
\draw [<-] (-0.4,-0.32) -- (1,0.8)[thick];
\draw [<-] (2.6,-0.32) -- (4,0.8)[thick];
\draw [<-] (-0.4,1.68) -- (1,2.8)[thick];
\draw [<-] (2.6,1.68) -- (4,2.8)[thick];
}
{\color{green}
\draw [->] (-0.3,3.5)--(3.4,3.5) to [out=0,in=90] (4.5,2.32)--(4.5,-0.3)[thick];}
\draw[fill=blue!20] (4.17,3.17) circle[radius=0.15] node[above right=1pt] {$(\frac{xw}{yu})^{h}$};
\draw[fill=blue!20] (-0.45,3.5) circle[radius=0.15] node[left=1pt] {$(\frac{y}{w})^{h}$};
\draw[fill=blue!20] (4.5,-0.45) circle[radius=0.15] node[right=2pt] {$(\frac{u}{x})^{h}$};
\end{scope}
\end{tikzpicture}
\label{T313}
\end{align}
By multiplying $(x/u)^h$ from the left and 
$(w/y)^h$ from the right on the component corresponding to the green arrow, 
we get 
\begin{align}\label{pSTT=TTS}
\begin{tikzpicture}
\begin{scope}[>=latex,xshift=0pt]
\draw [->] (-0.5,0) --(3,0) to [out=0, in=-150] (4.5,0.32) --(5,0.57); 
\draw[fill=blue!20] (3.5,0) circle[radius=0.15] node[below=2pt] {$u^{h}$}; 
\draw [->] (-0.5,2) -- (3,2) to [out=0, in=-150] (4.5,2.32) --(5,2.57);
\draw[fill=blue!20] (3.5,2) circle[radius=0.15] node[below=2pt] {$u^{h}$}; 
\draw [->] (0,-0.5) --(0,2) to [out=90, in=-120](0.4,3.5)--(0.8,4.2);
\draw[fill=blue!20] (0.1,2.9) circle[radius=0.15] node[left=2pt] {$w^{h}$}; 
\draw [->] (3,-0.5) --(3,2) to [out=90, in=-120](3.4,3.5)--(3.8,4.2);
\draw[fill=blue!20] (3.1,2.9) circle[radius=0.15] node[left=2pt] {$w^{h}$}; 
\draw (1.8,2.25) node {$\cdots$};
\draw (1.8,0.25) node {$\cdots$};
\draw (0.4,1.25) node {$\vdots$};
\draw (3.4,1.25) node {$\vdots$};
\draw [->] (0.3,0.64) --(3.8,0.64) to [out=0, in=135] (4.5,0.32) -- (4.8,0);
\draw[fill=blue!20] (4.2,0.61) circle[radius=0.15] node[above right=1pt] {$x^{h}$}; 
\draw [->] (0.3,2.64) --(3.8,2.64) to [out=0, in=135] (4.5,2.32) -- (4.8,2);
\draw[fill=blue!20] (4.2,2.61) circle[radius=0.15] node[above right=1pt] {$x^{h}$}; 
\draw [->] (0.8,0.14) --(0.8,2.64) to [out=90, in=-45](0.4,3.5) --(0,3.9);
\draw[fill=blue!20] (0.65,3.2) circle[radius=0.15] node[above right] {$y^{h}$}; 
\draw [->] (3.8,0.14) --(3.8,2.64) to [out=90, in=-45] (3.4,3.5) --(3,3.9);
\draw[fill=blue!20] (3.65,3.2) circle[radius=0.15] node[above right] {$y^{h}$}; {\color{blue} 
\draw [<-] (-0.4,-0.32) -- (1,0.8)[thick];
\draw [<-] (2.6,-0.32) -- (4,0.8)[thick];
\draw [<-] (-0.4,1.68) -- (1,2.8)[thick];
\draw [<-] (2.6,1.68) -- (4,2.8)[thick];
}
{\color{green}
\draw [->] (-0.3,3.5)--(3.4,3.5) to [out=0,in=90] (4.5,2.32)--(4.5,-0.3)[thick];}
\draw[fill=blue!20] (4.17,3.17) circle[radius=0.15] node[above right=1pt] {$(\frac{xw}{yu})^{h}$};
%
\draw (5.8,2) node{$=$};
\end{scope}
\begin{scope}[>=latex,xshift=220pt]
\draw [<-] (3.5,0)--(0,0) to [out=180, in=-45] (-0.7,0.32) --(-1.02,0.64); 
\draw[fill=blue!20] (3.65,0) circle[radius=0.15] node[right=2pt] {$x^{h}$}; 
\draw [<-] (3.5,2)--(0,2) to [out=180, in=-45] (-0.7,2.32) --(-1.02,2.64);
\draw[fill=blue!20] (3.65,2) circle[radius=0.15] node[right=2pt] {$x^{h}$}; 
\draw [<-] (0,2.5)--(0,0) to [out=-90, in=135] (0.4,-0.86)--(0.8,-1.26);
\draw[fill=blue!20] (0,2.65) circle[radius=0.15] node[above=2pt] {$y^{h}$}; 
\draw [<-] (3,2.5)--(3,0) to [out=-90, in=135] (3.4,-0.86)--(3.8,-1.26);
\draw[fill=blue!20] (3,2.65) circle[radius=0.15] node[above=2pt] {$y^{h}$}; 
\draw (1.8,2.25) node {$\cdots$};
\draw (1.8,0.25) node {$\cdots$};
\draw (0.4,1.25) node {$\vdots$};
\draw (3.4,1.25) node {$\vdots$};
\draw [<-] (4.3,0.64)--(0.8,0.64) to [out=180, in=30] (-0.7,0.32) -- (-1.254,0);\draw[fill=blue!20] (4.45,0.64) circle[radius=0.15] node[right=2pt] {$u^{h}$}; 
\draw [<-] (4.3,2.64)--(0.8,2.64) to [out=180, in=30] (-0.7,2.32) -- (-1.254,2);\draw[fill=blue!20] (4.45,2.64) circle[radius=0.15] node[right=2pt] {$u^{h}$}; 
\draw [<-] (0.8,3.14)--(0.8,0.64) to [out=-90, in=60] (0.4,-0.86) --(0,-1.55);
\draw[fill=blue!20] (0.8,3.29) circle[radius=0.15] node[above=2pt] {$w^{h}$}; 
\draw [<-] (3.8,3.14)--(3.8,0.64) to [out=-90, in=60] (3.4,-0.86) --(3,-1.55);
\draw[fill=blue!20] (3.8,3.29) circle[radius=0.15] node[above=2pt] {$w^{h}$}; 
{\color{blue} 
\draw [<-] (-0.4,-0.32) -- (1,0.8)[thick];
\draw [<-] (2.6,-0.32) -- (4,0.8)[thick];
\draw [<-] (-0.4,1.68) -- (1,2.8)[thick];
\draw [<-] (2.6,1.68) -- (4,2.8)[thick];
}
{\color{green}
\draw [->] (-0.7,2.82)--(-0.7,0.32) to [out=-90,in=180] (0.4,-0.86)--(4,-0.86)[thick];
}
\draw[fill=blue!20] (-0.7,2.97) circle[radius=0.15] node[left=1pt] {$(\frac{w}{y})^{h}$};
\draw[fill=blue!20] (4.15,-0.86) circle[radius=0.15] node[right=2pt] {$(\frac{x}{u})^{h}$};
\end{scope}
\end{tikzpicture}
.
\end{align}
Green arrows here are compositions of the operator
$\MM \in \mathrm{End}(V \otimes V) \otimes \mathcal{W}(-q)$ in (\ref{eq:M}) and Figure \ref{fig:LM}.
Employ its representation explained in (\ref{uwp})--(\ref{posc}), which amounts to replacing 
$\MM$ by $M \in \mathrm{End}(V \otimes V \otimes \mathcal{V}_+)$ in (\ref{M}).
Furthermore, take the trace over $\mathcal{V}_+$.
Then the green arrows are closed to form a circle as in Figure \ref{fig:trm} with $N=n+m$, 
and yield $R\bigl(\frac{xw}{yu}\bigr)$ (\ref{Rz}) 
on both sides thanks to the cyclicity of the trace.\footnote{Precisely speaking,  to achieve 
the identification with (\ref{pSTT=TTS}), one needs to reverse the order and shift the indices of 
$(\ref{tred})|_{N=n+m}$ as 
 $i_1,\ldots, i_{m+n} \rightarrow i_n,\ldots, i_2, i_1, i_{m+n},\ldots, i_{n+2},i_{n+1}$ 
 and similarly for $a_k, b_k$ and $j_k$.
 This rearrangement does not influence the assertion of Lemma \ref{le:rinv}.
For simplicity we write the $R$ matrix with this conventional adjustment also as $R(z)$ in the main text.}
In this way, we obtain 
\begin{align}\label{STT=TTS}
R\Bigl(\frac{xw}{yu}\Bigr)\mathcal{T}(u,w) \mathcal{T}(x,y)
= \mathcal{T}(x,y) \mathcal{T}(u,w) R\Bigl(\frac{xw}{yu}\Bigr),
\end{align}
which is depicted as follows:
\begin{align}\label{RTT=TTR}
\begin{tikzpicture}
\begin{scope}[>=latex,xshift=0pt]
\draw [->] (-0.5,0) --(3,0) to [out=0, in=-150] (4.5,0.32) --(5,0.57); 
\draw[fill=blue!20] (3.5,0) circle[radius=0.15] node[below=2pt] {$u^{h}$}; 
\draw [->] (-0.5,2) -- (3,2) to [out=0, in=-150] (4.5,2.32) --(5,2.57);
\draw[fill=blue!20] (3.5,2) circle[radius=0.15] node[below=2pt] {$u^{h}$}; 
\draw [->] (0,-0.5) --(0,2) to [out=90, in=-120](0.4,3.5)--(0.8,4.2);
\draw[fill=blue!20] (0.1,2.9) circle[radius=0.15] node[left=2pt] {$w^{h}$}; 
\draw [->] (3,-0.5) --(3,2) to [out=90, in=-120](3.4,3.5)--(3.8,4.2);
\draw[fill=blue!20] (3.1,2.9) circle[radius=0.15] node[left=2pt] {$w^{h}$}; 
\draw (1.8,2.25) node {$\cdots$};
\draw (1.8,0.25) node {$\cdots$};
\draw (0.4,1.25) node {$\vdots$};
\draw (3.4,1.25) node {$\vdots$};
\draw [->] (0.3,0.64) --(3.8,0.64) to [out=0, in=135] (4.5,0.32) -- (4.8,0);
\draw[fill=blue!20] (4.2,0.61) circle[radius=0.15] node[above right=1pt] {$x^{h}$}; 
\draw [->] (0.3,2.64) --(3.8,2.64) to [out=0, in=135] (4.5,2.32) -- (4.8,2);
\draw[fill=blue!20] (4.2,2.61) circle[radius=0.15] node[above right=1pt] {$x^{h}$}; 
\draw [->] (0.8,0.14) --(0.8,2.64) to [out=90, in=-45](0.4,3.5) --(0,3.9);
\draw[fill=blue!20] (0.65,3.2) circle[radius=0.15] node[above right] {$y^{h}$}; 
\draw [->] (3.8,0.14) --(3.8,2.64) to [out=90, in=-45] (3.4,3.5) --(3,3.9);
\draw[fill=blue!20] (3.65,3.2) circle[radius=0.15] node[above right] {$y^{h}$}; {\color{blue} 
\draw [<-] (-0.4,-0.32) -- (1,0.8)[thick];
\draw [<-] (2.6,-0.32) -- (4,0.8)[thick];
\draw [<-] (-0.4,1.68) -- (1,2.8)[thick];
\draw [<-] (2.6,1.68) -- (4,2.8)[thick];
}
{\color{green}
\draw [-] (-0.2,3.5)--(3.4,3.5) to [out=0,in=90] (4.5,2.32)--(4.5,-0.2) to [out=-90,in=180] (4.8,-0.5) to [out=0,in=-90] (5.1,-0.2)--(5.1,2.5) to [out=90,in=0] (3.5,3.98)--(-0.2,3.98) to [out=180,in=90] (-0.44,3.74) to [out=-90,in=180] (-0.2,3.5)
[thick];
\draw [->] (-0.2,3.5)--(2.0,3.5)[thick];
}
\draw[fill=blue!20] (4.17,3.17) circle[radius=0.15] node[above right=1pt] {$(\frac{xw}{yu})^{h}$};
%
\draw (5.8,2) node{$=$};
\end{scope}
\begin{scope}[>=latex,xshift=220pt]
\draw [<-] (3.5,0)--(0,0) to [out=180, in=-45] (-0.7,0.32) --(-1.02,0.64); 
\draw[fill=blue!20] (3.65,0) circle[radius=0.15] node[right=2pt] {$x^{h}$}; 
\draw [<-] (3.5,2)--(0,2) to [out=180, in=-45] (-0.7,2.32) --(-1.02,2.64);
\draw[fill=blue!20] (3.65,2) circle[radius=0.15] node[right=2pt] {$x^{h}$}; 
\draw [<-] (0,2.5)--(0,0) to [out=-90, in=135] (0.4,-0.86)--(0.8,-1.26);
\draw[fill=blue!20] (0,2.65) circle[radius=0.15] node[above=2pt] {$y^{h}$}; 
\draw [<-] (3,2.5)--(3,0) to [out=-90, in=135] (3.4,-0.86)--(3.8,-1.26);
\draw[fill=blue!20] (3,2.65) circle[radius=0.15] node[above=2pt] {$y^{h}$}; 
\draw (1.8,2.25) node {$\cdots$};
\draw (1.8,0.25) node {$\cdots$};
\draw (0.4,1.25) node {$\vdots$};
\draw (3.4,1.25) node {$\vdots$};
\draw [<-] (4.3,0.64)--(0.8,0.64) to [out=180, in=30] (-0.7,0.32) -- (-1.254,0);\draw[fill=blue!20] (4.45,0.64) circle[radius=0.15] node[right=2pt] {$u^{h}$}; 
\draw [<-] (4.3,2.64)--(0.8,2.64) to [out=180, in=30] (-0.7,2.32) -- (-1.254,2);\draw[fill=blue!20] (4.45,2.64) circle[radius=0.15] node[right=2pt] {$u^{h}$}; 
\draw [<-] (0.8,3.14)--(0.8,0.64) to [out=-90, in=60] (0.4,-0.86) --(0,-1.55);
\draw[fill=blue!20] (0.8,3.29) circle[radius=0.15] node[above=2pt] {$w^{h}$}; 
\draw [<-] (3.8,3.14)--(3.8,0.64) to [out=-90, in=60] (3.4,-0.86) --(3,-1.55);
\draw[fill=blue!20] (3.8,3.29) circle[radius=0.15] node[above=2pt] {$w^{h}$}; 
{\color{blue} 
\draw [<-] (-0.4,-0.32) -- (1,0.8)[thick];
\draw [<-] (2.6,-0.32) -- (4,0.8)[thick];
\draw [<-] (-0.4,1.68) -- (1,2.8)[thick];
\draw [<-] (2.6,1.68) -- (4,2.8)[thick];
}
{\color{green}
\draw [-] (-0.7,2.8)--(-0.7,0.32) to [out=-90,in=180] (0.4,-0.86)--(4,-0.86) to [out=0,in=90] (4.24,-1.1) to [out=-90,in=0] (4,-1.34)--(0.4,-1.34) to [out=180,in=-90] (-1.3,0.32)--(-1.3,2.8) to [out=90,in=180] (-1,3.1) to [out=0,in=90] (-0.7,2.8)
[thick];
\draw [->] (0.4,-0.86)--(2,-0.86)[thick];
}
\draw[fill=blue!20] (4.24,-1.1) circle[radius=0.15] node[right=2pt] {$(\frac{xw}{yu})^{h}$};
\end{scope}
\end{tikzpicture}
.
\end{align}
This is an equality in 
$\mathrm{End}\left((V^{\otimes m} \otimes V^{\otimes n})^{\otimes 2}\right) \otimes \mathcal{W}_{mn}(q)$.
Here $(V^{\otimes m} \otimes V^{\otimes n})^{\otimes 2}$
is the two copies of the auxiliary space for the two monodromy matrices, which correspond to the front 
and the back layers of diagram (\ref{RTT=TTR}).
On the other hand, the $R$ matrix 
$R\bigl(\frac{xw}{yu}\bigr)$ only acts on the auxiliary space 
$(V^{\otimes m} \otimes V^{\otimes n})^{\otimes 2}= V^{\otimes m+n} \otimes V^{\otimes m+n}$.
Now we invoke  Lemma \ref{le:rinv} with $N=m+n$, 
which guarantees that the $R$ matrix is invertible for generic $\frac{xw}{yu}$.
Therefore one may multiply $R\bigl(\frac{xw}{yu}\bigr)^{-1}$ from the right to both sides of (\ref{STT=TTS}) 
and take the trace over the auxiliary space. 
In view of the definition (\ref{T-G}), the result is nothing but $T_G(u,w)T_G(x,y) = T_G(x,y)T_G(u,w)$.
\end{proof}

\begin{remark}\label{re:inho}
The monodromy matrix $\mathcal{T}(x,y)$ admits two types of inhomogeneous generalizations which preserves the commutativity of the transfer matrices as stated in Theorem \ref{thm:TT1}.
\\
(i) Inhomogeneous parameters for $\mathcal{L}$.
In a wiring diagram $G$, we place $\mathcal{L}(r_i,s_i,f_i,g_i;q)$ at vertex $i$ as we did in \S \ref{sec:RLLL}, and define the monodromy matrix $\mathcal{T}_G(x,y)$ with these inhomogeneous parameters. This inhomogeneity does not affect $\mathcal{M}$, nor the (four) tetrahedron equations or the inversion relations.
Hence the commutativity of transfer matrices \eqref{G-TT} still holds.    
\\
(ii) Inhomogeneous spectral parameters. In the top-left diagram in (\ref{T313}), replace the $m$ pairs of parameters $(x,u)$  with
$(x\mu_1, u\mu_1), \ldots, (x\mu_m, u\mu_m)$ from top to bottom, and similarly, replace
the $n$ pairs of $(y,w)$ with $(y\nu_1,w\nu_1), \ldots, (y\nu_n,w\nu_n)$ from left to right,
where $\mu_i, \nu_i$ are arbitrary nonzero parameters.
Then the subsequent argument remains valid, as it only concerns the pairwise ratios, which remain unchanged;
$x\mu_i/u\mu_i = x/u$ and $y\nu_i/w\nu_i= y/w$.
As a result, one obtains the {\em inhomogeneous} generalization 
$T_G(x,y| \mu_1,\ldots, \mu_m, \nu_1,\ldots, \nu_n)$ of $T_G(x,y)$, which still satisfies the 
commutativity:
\begin{align}\label{Tmunu}
[T_G(x,y|\mu_1,\ldots, \mu_m, \nu_1,\ldots, \nu_n),T_G(u,w| \mu_1,\ldots, \mu_m, \nu_1,\ldots, \nu_n)]=0.
\end{align}
This extension naturally serves as a 3D analogue of the well-known ``inhomogeneity in the bulk" 
 and the ``boundary magnetic fields" that can be introduced in 2D commuting transfer matrices.
We stress that (\ref{Tmunu}) implies the commutativity of 
the transfer matrix $T_G(x,y)$ corresponding to an arbitrary {\em fixed} boundary condition.
In fact, the local states at the boundary edges can be frozen to $0$ or $1$  
by setting the corresponding magnetic fields to zero or infinity. 
This fact will be utilized in \S \ref{ss:FP} to reproduce the celebrated free parafermion model.
\end{remark}

\subsection{Admissible case}

Now we consider more general wiring diagrams on a torus, consisting of a finite numbers of directed closed  wires. We assume that in a diagram only two wires intersect at each crossing. See Figure \ref{fig:Gs}.
In a diagram on a torus, an intersection point is called a vertex, 
and an area surrounded by wires is called a face. 
For such a diagram $G$, we define a (quantized) six-vertex model in the same manner as the square grid case, by assigning the operator $\mathcal{L}$ \eqref{L1} to each vertex. 

\begin{figure}[ht]
\[
\begin{tikzpicture}
\begin{scope}[>=latex,xshift=0pt]
\draw [->] (2.5,0)--(2.5,4);
\draw [->] (1.5,0)--(1.5,1) to [out=90,in=0](0,2);
\draw [->] (5,2)--(3.5,2) to [out=180,in=-40](2.5,2.5) to [out=140,in=-90](1.5,4);
\draw [->] (3.5,0) to [out=90,in=180](5,1);
\draw [->] (0,1)--(1.5,1) to [out=0,in=-90](3.5,2)--(3.5,4);
{\color{red}
\draw [-] (0,4)--(5,4);
\draw [-] (0,0)--(5,0);
\draw [-] (0,0)--(0,4);
\draw [-] (5,0)--(5,4);
}
\end{scope}
\begin{scope}[>=latex,xshift=190pt]
\draw [<-] (2.5,0)--(2.5,4);
\draw [->] (1.5,0)--(1.5,1) to [out=90,in=0](0,2);
\draw [->] (5,2)--(3.5,2) to [out=180,in=-40](2.5,2.5) to [out=140,in=-90](1.5,4);
\draw [->] (3.5,0) to [out=90,in=180](5,1);
\draw [->] (0,1)--(1.5,1) to [out=0,in=-90](3.5,2)--(3.5,4);
{\color{red}
\draw [-] (0,4)--(5,4);
\draw [-] (0,0)--(5,0);
\draw [-] (0,0)--(0,4);
\draw [-] (5,0)--(5,4);
}
\end{scope}
\end{tikzpicture}
\]
\caption{Wiring diagrams on a torus. The left one is admissible, and the right one is not. 
The only difference is the orientation of the middle vertical wire.}
\label{fig:Gs}
\end{figure}
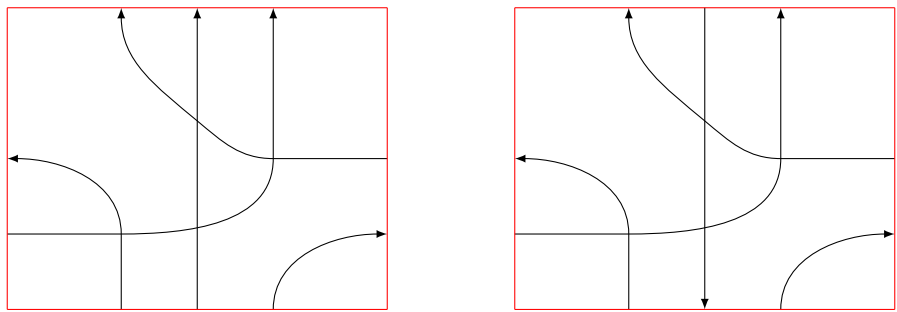

Fix a fundamental domain of the torus, so that its boundary (a red rectangle) 
does not overlap the vertices of $G$.
As in (\ref{grw}), we introduce a green arrow going from the NW to the SE via the NE corner,
and attempt to move it to the one passing through the SW corner. 
The initial and final arrows will be referred to as the NE arrow and the SW arrow, respectively.
In a similar manner to the ordinary case, this is achieved by applying the tetrahedron equation \eqref{LM-o}--\eqref{LM-t} and the inversion relations \eqref{M-I}--\eqref{M-I'}.  
We classify the operations (o), (h), (v), (t), (I), (iI), (I') and (iI') into three types as
\begin{align}\label{types}
\text{type A: (v), (I), (iI)}, \qquad
\text{type B: (h), (I'),(iI')}, \qquad
\text{type C: (o), (t)}.
\end{align}
These groupings are introduced according to the types of constrains among the parameters
in Proposition \ref{pr:te4} and Lemma \ref{le:inv}.
Types A and B correspond to the distinct  constraints $(r', g') =(s',q^{-1}f')$ and 
$(r', g') =(s',qf')$, respectively, whereas type C imposes no condition.
Evidently,  the parameters $(r',s',f',g')$ for $\mathcal{M}$ must be chosen 
in a specific manner when sending the NE green arrow to the NE 
by applying the tetrahedron equations and inversion relations.
This observation motivates the following definition.

\begin{definition}\label{def:adm}
A wiring diagram is said to be {\it admissible} if the initial NE arrow 
can be transformed into the final SW arrow by a sequence of the eight operations (\ref{types}), 
in such a way that operations of types A and B do not coexist.
\end{definition}

For example, in Figure \ref{fig:Gs}, the left diagram is admissible, but the right one is not. 
In Figure \ref{fig:ad-G}, we show a sequence of operations for the left diagram: (h)-(h)-(o)(I')-(iI')-(o). For the right diagram, the sequence is modified to (h)-(t)-(v)(I')-(iI')-(o), where operations of types A and B coexist.

\begin{figure}[H]
\[
\begin{tikzpicture}[scale=0.8, font=\small]
\begin{scope}[>=latex,xshift=0pt]
\draw [->] (2.5,0)--(2.5,4);
\draw [->] (1.5,0)--(1.5,1) to [out=90,in=0](0,2);
\draw [->] (5,2)--(3.5,2) to [out=180,in=-40](2.5,2.5) to [out=140,in=-90](1.5,4);
\draw [->] (3.5,0) to [out=90,in=180](5,1);
\draw [->] (0,1)--(1.5,1) to [out=0,in=-90](3.5,2)--(3.5,4);
{\color{red}
\draw [-] (0,4)--(5,4);
\draw [-] (0,0)--(5,0);
\draw [-] (0,0)--(0,4);
\draw [-] (5,0)--(5,4);
}
{\color{green}
\draw [->] (0,4.2)--(5,4.2) to [out=0,in=90] (5.2,4)--(5.2,0)[thick];
}
\draw (2,4.5) node[right=-24pt] {\small the initial NE arrow};
\end{scope}
\begin{scope}[>=latex,xshift=190pt]
\draw [->] (2.5,0)--(2.5,4);
\draw [->] (1.5,0)--(1.5,1) to [out=90,in=0](0,2);
\draw [->] (5,2)--(3.5,2) to [out=180,in=-40](2.5,2.5) to [out=140,in=-90](1.5,4);
\draw [->] (3.5,0) to [out=90,in=180](5,1);
\draw [->] (0,1)--(1.5,1) to [out=0,in=-90](3.5,2)--(3.5,4);
{\color{red}
\draw [-] (0,4)--(5,4);
\draw [-] (0,0)--(5,0);
\draw [-] (0,0)--(0,4);
\draw [-] (5,0)--(5,4);
}
{\color{green}
\draw [->] (0,4.2)--(1,4.2) to [out=0,in=180] (1.5,3)--(3.5,3) to [out=0,in=180] (5,-0.2)[thick];
\draw[fill] (3.5,2) circle[radius=0.1];
}
\draw (-0.5,2) node[left] {$=$};
\draw (2.3,-0.6) node[right] {$\downarrow$ ~(h)};
\end{scope}
\begin{scope}[>=latex,xshift=190pt, yshift=-150]
\draw [->] (2.5,0)--(2.5,4);
\draw [->] (1.5,0)--(1.5,1) to [out=90,in=0](0,2);
\draw [->] (5,2)--(3.5,2) to [out=180,in=-40](2.5,2.5) to [out=140,in=-90](1.5,4);
\draw [->] (3.5,0) to [out=90,in=180](5,1);
\draw [->] (0,1)--(1.5,1) to [out=0,in=-90](3.5,2)--(3.5,4);
{\color{red}
\draw [-] (0,4)--(5,4);
\draw [-] (0,0)--(5,0);
\draw [-] (0,0)--(0,4);
\draw [-] (5,0)--(5,4);
}
{\color{green}
\draw [->] (0,4.2)--(1,4.2) to [out=0,in=180] (1.5,3)--(2.5,3) to [out=0,in=90] (3,1) to [out=-90,in=180] (3.5,0.5) to [out=0,in=180] (5,-0.2)[thick];
\draw[fill] (2.5,2.5) circle[radius=0.1];
}
\draw (-0.5,2) node[left] {$\leftarrow$};
\draw (-0.4,2.4) node[left] {(h)};

\end{scope}
\begin{scope}[>=latex,yshift=-150]
\draw [->] (2.5,0)--(2.5,4);
\draw [->] (1.5,0)--(1.5,1) to [out=90,in=0](0,2);
\draw [->] (5,2)--(3.5,2) to [out=180,in=-40](2.5,2.5) to [out=140,in=-90](1.5,4);
\draw [->] (3.5,0) to [out=90,in=180](5,1);
\draw [->] (0,1)--(1.5,1) to [out=0,in=-90](3.5,2)--(3.5,4);
{\color{red}
\draw [-] (0,4)--(5,4);
\draw [-] (0,0)--(5,0);
\draw [-] (0,0)--(0,4);
\draw [-] (5,0)--(5,4);
}
{\color{green}
\draw [->] (0,4.2)--(1,4.2) to [out=0,in=100] (1.4,3.5) to [out=-80,in=180] (2,3.2) to [out=0,in=100] (2.2,2) to [out=-80,in=100] (4,1) to [out=-80,in=180] (5,-0.2)[thick];
\draw[fill] (2.5,1.1) circle[radius=0.1];
\draw (2.1,3) node {$\ast$};
}
\draw (2.3,-0.6) node[right] {$\downarrow$ ~(o),~(I')};
\end{scope}
\begin{scope}[>=latex,yshift=-300]
\draw [->] (2.5,0)--(2.5,4);
\draw [->] (1.5,0)--(1.5,1) to [out=90,in=0](0,2);
\draw [->] (5,2)--(3.5,2) to [out=180,in=-40](2.5,2.5) to [out=140,in=-90](1.5,4);
\draw [->] (3.5,0) to [out=90,in=180](5,1);
\draw [->] (0,1)--(1.5,1) to [out=0,in=-90](3.5,2)--(3.5,4);
{\color{red}
\draw [-] (0,4)--(5,4);
\draw [-] (0,0)--(5,0);
\draw [-] (0,0)--(0,4);
\draw [-] (5,0)--(5,4);
}
{\color{green}
\draw [->] (0,4.2)--(0.5,4.2) to [out=0,in=100] (1,2.5) to [out=-80,in=100] (2,1) to [out=-80,in=180] (4,0.5) to [out=0,in=180] (5,-0.2)[thick];
}
\draw (5.5,2) node[right] {$\rightarrow$};
\draw (5.4,2.4) node[right] {(iI')};
\end{scope}
\begin{scope}[>=latex,xshift=190,yshift=-300]
\draw [->] (2.5,0)--(2.5,4);
\draw [->] (1.5,0)--(1.5,1) to [out=90,in=0](0,2);
\draw [->] (5,2)--(3.5,2) to [out=180,in=-40](2.5,2.5) to [out=140,in=-90](1.5,4);
\draw [->] (3.5,0) to [out=90,in=180](5,1);
\draw [->] (0,1)--(1.5,1) to [out=0,in=-90](3.5,2)--(3.5,4);
{\color{red}
\draw [-] (0,4)--(5,4);
\draw [-] (0,0)--(5,0);
\draw [-] (0,0)--(0,4);
\draw [-] (5,0)--(5,4);
}
{\color{green}
\draw [->] (0,4.2)--(0.3,4.2) to [out=0,in=100] (1,1.5) to [out=-80,in=100] (2,1.1) to [out=-80,in=180] (4,0.5) to [out=0,in=180] (5,-0.2)[thick];
\draw[fill] (1.5,1) circle[radius=0.1];
\draw (1.15,1.5) node {$\ast$};
}
\draw (2.3,-0.6) node[right] {$\downarrow$ ~(o)};
\end{scope}
\begin{scope}[>=latex,xshift=190,yshift=-450]
\draw [->] (2.5,0)--(2.5,4);
\draw [->] (1.5,0)--(1.5,1) to [out=90,in=0](0,2);
\draw [->] (5,2)--(3.5,2) to [out=180,in=-40](2.5,2.5) to [out=140,in=-90](1.5,4);
\draw [->] (3.5,0) to [out=90,in=180](5,1);
\draw [->] (0,1)--(1.5,1) to [out=0,in=-90](3.5,2)--(3.5,4);
{\color{red}
\draw [-] (0,4)--(5,4);
\draw [-] (0,0)--(5,0);
\draw [-] (0,0)--(0,4);
\draw [-] (5,0)--(5,4);
}
{\color{green}
\draw [->] (0,4.2)--(0.2,4.2) to [out=0,in=100] (1.2,0.7) to [out=-80,in=180] (1.5,0.4) -- (4,0.4) to [out=0,in=180] (5,-0.2)[thick];
}
\draw (-0.5,2) node[left] {$=$};
\end{scope}
\begin{scope}[>=latex,yshift=-450pt]
\draw [->] (2.5,0)--(2.5,4);
\draw [->] (1.5,0)--(1.5,1) to [out=90,in=0](0,2);
\draw [->] (5,2)--(3.5,2) to [out=180,in=-40](2.5,2.5) to [out=140,in=-90](1.5,4);
\draw [->] (3.5,0) to [out=90,in=180](5,1);
\draw [->] (0,1)--(1.5,1) to [out=0,in=-90](3.5,2)--(3.5,4);
{\color{red}
\draw [-] (0,4)--(5,4);
\draw [-] (0,0)--(5,0);
\draw [-] (0,0)--(0,4);
\draw [-] (5,0)--(5,4);
}
{\color{green}
\draw [->] (-0.2,4)--(-0.2,0) to [out=-90,in=180] (0,-0.2)--(5,-0.2)[thick];
}
\draw (-0.1,-0.5) node[right] {\small the final SW arrow};
\end{scope}
\end{tikzpicture}
\]
\caption{An admissible transformation. The vertices affected in the next step are shown in green. The region marked with $*$ denotes where the inversion relation is applied.}
\label{fig:ad-G}
\end{figure}
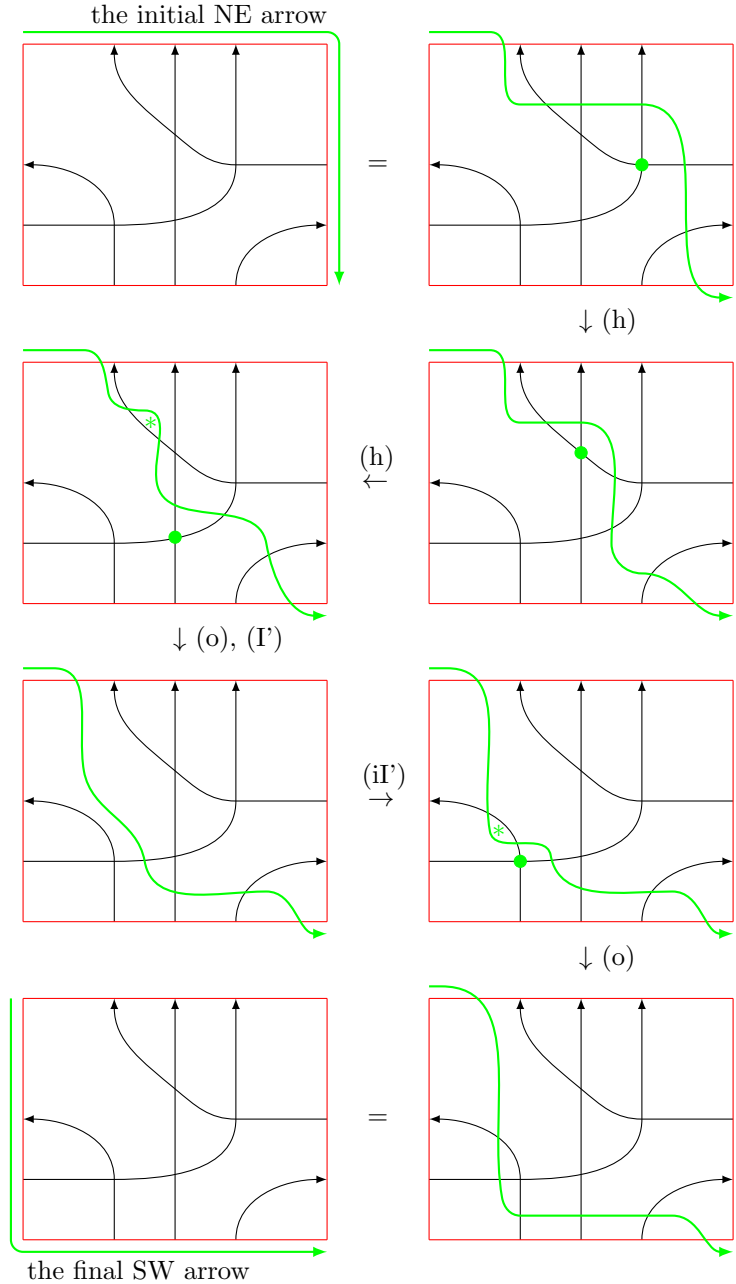

Now we describe the main results. Let $G$ be a diagram on a torus as described above, with $K$ vertices. Assign the operator $\mathcal{L}$ \eqref{L1} to each vertex in $G$, and fix a fundamental domain of the torus. The monodromy matrix $\mathcal{T}_G(x,y)$ is defined schematically by the following picture:
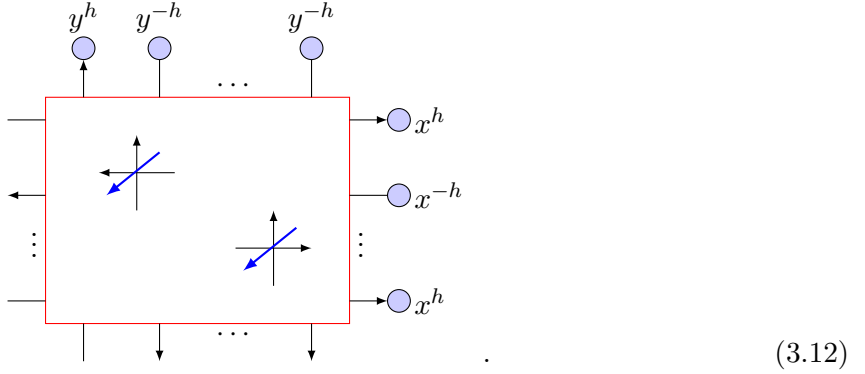
\begin{align}\label{pT_G}
\begin{tikzpicture}
\begin{scope}[>=latex,xshift=0pt]
{\color{red}
\draw[-] (-0,-0)--(-0,3)--(4,3)--(4,0)--(0,0);
}
\draw [-] (-0.5,0.3)--(0,0.3);  \draw [->] (4,0.3)--(4.5,0.3);
\draw[fill=blue!20] (4.65,0.3) circle[radius=0.15] node[right=2pt] {$x^{h}$}; 
\draw [<-] (-0.5,1.7)--(0,1.7); \draw [-] (4,1.7)--(4.5,1.7); 
\draw[fill=blue!20] (4.65,1.7) circle[radius=0.15] node[right=2pt] {$x^{-h}$}; 
\draw [-] (-0.5,2.7)--(0,2.7);  \draw [->] (4,2.7)--(4.5,2.7);
\draw[fill=blue!20] (4.65,2.7) circle[radius=0.15] node[right=2pt] {$x^{h}$}; 
\draw [-] (0.5,-0.5)--(0.5,0);  \draw [->] (0.5,3)--(0.5,3.5); 
\draw[fill=blue!20] (0.5,3.65) circle[radius=0.15] node[above=2pt] {$y^{h}$}; 
\draw [<-] (1.5,-0.5)--(1.5,0);  \draw [-] (1.5,3)--(1.5,3.5); 
\draw[fill=blue!20] (1.5,3.65) circle[radius=0.15] node[above=2pt] {$y^{-h}$}; 
\draw [<-] (3.5,-0.5)--(3.5,0);  \draw [-] (3.5,3)--(3.5,3.5); 
\draw[fill=blue!20] (3.5,3.65) circle[radius=0.15] node[above=2pt] {$y^{-h}$}; 
\draw (2.5,3.15) node {$\cdots$};
\draw (2.5,-0.15) node {$\cdots$};
\draw (-0.15,1.15) node {$\vdots$};
\draw (4.15,1.15) node {$\vdots$};
\draw [->] (2.5,1) -- (3.5,1); \draw [->] (3,0.5) -- (3,1.5);
{\color{blue} 
\draw [<-] (2.6,0.7) -- (3.3,1.27)[thick];
}
\draw [<-] (0.7,2) -- (1.7,2); \draw [->] (1.2,1.5) -- (1.2,2.5);
{\color{blue} 
\draw [<-] (0.8,1.7) -- (1.5,2.27)[thick];
}
\end{scope}
\end{tikzpicture}
\ .
\end{align}
Here we attach each outgoing (resp. incoming) arrow at the east boundary with $x^h$ (resp. $x^{-h}$), and attach each outgoing (resp. incoming) arrow at the north boundary with $y^h$ (resp. $y^{-h}$).
We assume that the total number of arrows on the north boundary (resp. east boundary) is $n$ (resp. $m$), and write $N$ for $m+n$.
For each arrow on the east boundary we define their sign by 
$$
  s_x(k) = 
\begin{cases}
  1 & \text{if the $k$-th arrow is outgoing},
  \\
  -1 & \text{if the $k$-th arrow is incoming}.
\end{cases} 
$$
for $k=1,\ldots,m$, 
and for arrows on the north boundary define $s_y(k)$ for $k=1,\ldots,n$ in the same manner.
Under this setting, we have 
$\mathcal{T}_G(x,y) \in \mathrm{End}(V^{\otimes N}) \otimes \mathcal{W}_K(q)$.

As in the ordinary case, 
we define the transfer matrix $T_G(x,y)$ as $\Tr_{V^{\otimes N}} \mathcal{T}_G(x,y)$:
\begin{align}\label{T-Gg}
T_G(x,y) &= \mathrm{Tr}_{V^{\otimes N}}\left(\mathcal{T}_G(x,y)\right)
= \sum_{\mathbf{i} \in \{0,1\}^m, \,\mathbf{j} \in \{0,1\}^n} 
T^{\mathbf{i},\mathbf{j}}_{\mathbf{i}, \mathbf{j}} \,
x^{||\mathbf{i}||} y^{||\mathbf{j}||} \in \mathcal{W}_{K}(q),
\end{align}
where we set $||\mathbf{i}|| := \sum_{k=1,\ldots,m} s_x(k) i_k$ 
and $||\mathbf{j}|| := \sum_{k=1,\ldots,n} s_y(k) j_k$.
The following theorem is proved in the same manner as Theorem \ref{thm:TT1}.
 
\begin{thm}\label{thm:TT2}
If a wiring diagram $G$ on a torus with a fixed fundamental domain is admissible, then for any $x,y,u,w \in \C$, it holds that 
\begin{align}
[T_G(x,y), ~T_G(u,w)] = 0.
\end{align}
\end{thm} 

\begin{proof}
We just present some points different from the proof of Theorem \ref{thm:TT1}.
For an admissible $G$, an equality of diagrams similar to \eqref{pSTT=TTS} holds, 
where we set $r'=s'$ and $g' = qf'$ (resp. $r'=s'$ and $g' = q^{-1}f'$) in 
$\mathcal{M}$ \eqref{eq:M}, when we apply only the type B and C moves (resp. type A and C moves) to $G$. 
When the operators $x^{-h}$ or $y^{-h}$ are attached to arrows, Lemma \ref{lem:L-h} is applied in a form as 
$$
  x^{-h_1} y^{-h_2} \cdot \mathcal{L}\cdot (y/x)^{h_3} = (y/x)^{h_3} \cdot \mathcal{L} \cdot x^{-h_1} y^{-h_2}.
$$

For a vertex on the NE arrow, which has incoming arrow at the east or north boundary,  
we use the transpose $\overline{M}$ of $M$ \eqref{M} with respect to $V \otimes V$ instead of $M$. 
The operator $\overline{M}$ is written as $\sigma \cdot (1 \otimes A) \cdot M \cdot (A^{-1} \otimes 1) \cdot \sigma$, where $\sigma$ acts on $V\otimes V$ as $\sigma (v_i \otimes v_j) = v_{1-i} \otimes v_{1-j}$, and $A$ acts on $V$ as $A(v_i) = (q \alpha^2)^i \,v_i$. As both $\sigma$ and $A$ are invertible, the trace of the NE arrow gives an invertible operator: if the $k$-th vertex has $\overline{M}$, for example, we get
$$
\mathrm{Tr}_{\mathcal{V}_+}
\left(z^h M_1 \cdots \overline{M}_k \cdots M_N \right) = (\sigma \cdot 1 \otimes A)_k \cdot R(z) \cdot (A^{-1} \otimes 1 \cdot \sigma)_k,
$$
which is invertible. 
If we used only the type B and C moves (resp. type A and C moves) to $G$,
we set $\alpha=\pm q^{-1}$ (resp. $\alpha= \pm 1)$ for $M$ and $\overline{M}$. 
\end{proof} 

\begin{remark}\label{rmk:inho2}
The transfer matrix $T_G(x,y)$ for an admissible $G$ also allows the two types of inhomogeneous generalization as described in Remark \ref{re:inho},
which preserves the commutativity.
\end{remark}

The transfer matrix $T_G(x,y)$ has the following symmetries.

\begin{prop}\label{prop:inv}
\begin{itemize}
\item[(1)] 
The transfer matrix $T_G(x,y)$ is invariant under translation of the fundamental domain.
\item[(2)]
The $SL(2,\Z)$-action on the torus 
induces the following transformation of $T_G(x,y)$,
$$
  T_G(x,y) = \sum_{\mathbf{i},\,\mathbf{j}} 
T^{\mathbf{i},\mathbf{j}}_{\mathbf{i}, \mathbf{j}} x^{||\mathbf{i}||} y^{||\mathbf{j}||} \stackrel{M}{\longmapsto}
  T_G(x',y') = \sum_{\mathbf{i},\,\mathbf{j}} T^{\mathbf{i},\mathbf{j}}_{\mathbf{i}, \mathbf{j}} {x'}^{n_x} {y'}^{n_y},  
$$
for $S =\begin{pmatrix} a & b \\ c & d \end{pmatrix} \in SL(2,\Z)$. Here $(x',y')=(x^ay^b,x^cy^d)$ is the spectral parameters associated to the resulting fundamental domain, and $(n_x,n_y) := (||\mathbf{i}||,||\mathbf{j}||) S^{-1}.$
\end{itemize}
\end{prop}

\begin{proof}
(1) For $\mathbf{i}$ and $\mathbf{j}$ such that $\mathbf{i} \neq 0$ or $\mathbf{j} \neq 0$, take a configuration $W$ on $G$, which is contributing to $T^{\mathbf{i},\mathbf{j}}_{\mathbf{i}, \mathbf{j}}$. By Figure \ref{fig:6v}, $W$ consists of some oriented closed paths of weight $1$. Take one of them and write $p$ for this. Assume it respectively crosses the north and east boundary $k_y$ and $k_x$ times. Here we count the number of crossings with sign: $+1$ and $-1$ when the path is outgoing and incoming respectively. These numbers $k_x$ and $k_y$ are invariant to translation of the fundamental domain. When $\mathbf{i} = 0$ and $\mathbf{j} = 0$, it is trivial that $||\mathbf{i}||=||\mathbf{j}||=n_x=n_y=0$.
Consequently, the claim follows.
\\
(2) By $S = \begin{pmatrix} a & b \\ c & d \end{pmatrix} \in SL(2,\Z)$, the original fundamental domain spanned by $\vec{A} := (1,0)$ and $\vec{B} := (0,1)$ is transformed to that spanned by $\vec{A'}:=(1,0) S$ and $\vec{B'}:=(0,1) S$. Hence the spectral parameters are related by $(x',y') = (x^a y^b, x^c y^d)$. We compute how many times the oriented closed path $p$ (taken in (1)) crosses the new boundaries using the following pictures.

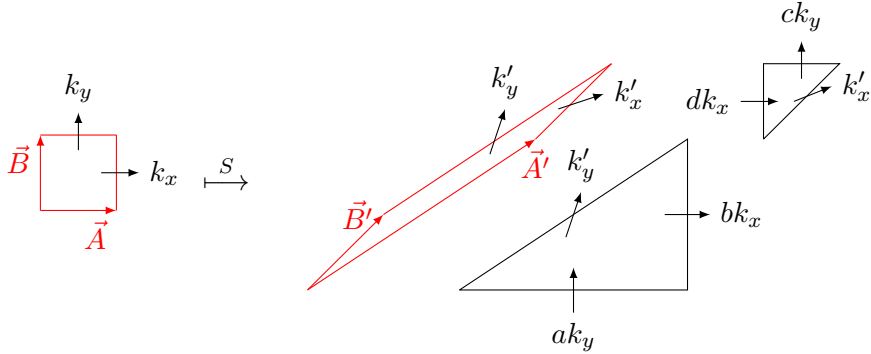
\begin{align*}
\begin{tikzpicture}
\begin{scope}[>=latex,xshift=0pt,yshift=30pt]
{\color{red}
\draw [->] (0,0)--(1,0) node[below left]{$\vec{A}$};
\draw [-] (0,1)--(1,1);
\draw [->] (0,0)--(0,1) node[below left]{$\vec{B}$};
\draw [-] (1,0)--(1,1);
}
\draw [->] (0.5,0.8)  --(0.5,1.3) node [above]{$k_y$};
\draw [->] (0.8,0.5)  --(1.3,0.5) node [right]{$k_x$};
\draw (2,0.5) node[right]{$\stackrel{S}{\longmapsto}$};
\end{scope}
\begin{scope}[>=latex,xshift=100pt]
{\color{red}
\draw [->] (0,0)--(3,2) node[below ]{$\vec{A'}$};
\draw [-] (1,1)--(4,3);
\draw [->] (0,0)--(1,1) node[left]{$\vec{B'}$};
\draw [-] (3,2)--(4,3);
}
\draw [->] (2.4,1.8)  --(2.6,2.4) node [above]{$k_y'$};
\draw [->] (3.3,2.4)  --(3.9,2.6) node [right]{$k_x'$};
\draw [-] (2,0)--(5,2)-- (5,0)--(2,0);
\draw [->] (3.4,0.7)  --(3.6,1.3) node [above]{$k_y'$};
\draw [->] (3.5,-0.3)node [below]{$a k_y$} --(3.5,0.3);
\draw [->] (4.7,1)  --(5.3,1) node [right]{$b k_x$};
\draw [-] (6,2)--(6,3)-- (7,3)--(6,2); 
\draw [->] (6.4,2.5)  --(6.9,2.7) node [right]{$k_x'$};
\draw [->] (6.5,2.8)  --(6.5,3.3) node [above]{$c k_y$};
\draw [->] (5.7,2.5)node [left]{$d k_x$} --(6.2,2.5);
\end{scope}
\end{tikzpicture}
\end{align*}
Due to the conservation of weights we have $a k_y = k'_y + b k_x$ and $d k_x = c k_y + k_x'$, thus we get $(k_x',k_y') = (k_x,k_y) S^{-1}$. By summing up all configurations contributing to $T^{\mathbf{i},\mathbf{j}}_{\mathbf{i}, \mathbf{j}}$, we obtain  $(n_x,n_y) := (||\mathbf{i}||,||\mathbf{j}||) S^{-1}$. 
Note that $x^{||\mathbf{i}||} y^{||\mathbf{j}||} = {x'}^{n_x} {y'}^{n_y}$. 
\end{proof}

Utilizing the $RLLL$ relation \eqref{RLLL-dgm}, we define the {\it Yang-Baxter move} (YB move) of the q-6V model.  
Let $R_{ijk}$ be the local transformation of wiring diagrams given by
\begin{equation}\label{YBm}
\begin{tikzpicture}
\begin{scope}[>=latex,xshift=0pt]
\draw [->] (0,2) to [out = 0, in = 135] (2,1.5) coordinate(B) node[above=2pt]{$j$} -- (3,0.5) coordinate(C) node[below=2pt]{$i$} to [out = -45, in = 180] (4,0);
\draw [->] (0,1) to [out = 0, in = 135] (1,0.5) coordinate(A) node[below=2pt]{$k$} to [out = -45, in = -135] (C) to [out = 45, in = 180] (4,1);
\draw [->] (0,0) to [out = 0, in = -135] (A) -- (B) to [out = 45, in = 180] (4,2);
\draw[->] (4.5,1) -- (5.5,1);
\draw (5,1.3) node{$R_{ijk}$};
\end{scope}
\begin{scope}[>=latex,xshift=170pt]
\draw [->] (0,0) to [out = 0, in = -135] (2,0.5) coordinate(B) node[below=2pt]{$j$}-- (3,1.5) coordinate(A) node[above=2pt]{$k$}to [out = 45, in = 180] (4,2);
\draw [->] (0,1) to [out = 0, in = -135] (1,1.5) coordinate(C) node[above=2pt]{$i$} to [out = 45, in = 135] (A) to [out = -45, in = 180] (4,1);
\draw [->] (0,2) to [out = 0, in = 135] (C) -- (B) to [out = -45, in = 180] (4,0);
\end{scope}
\end{tikzpicture}
\ .
\end{equation}
For a wiring diagram $G$ on a torus, which includes a sub diagram as the left side of \eqref{YBm}, write $G'$ for $R_{ijk}(G)$. 
Let $\mathcal{T}_G(x,y)$ and $\mathcal{T}_{G'}(x,y)$ be the monodromy matrices for the q-6v model on $G$ and $G'$ respectively. Write $\mathscr{R}_{ijk}$ for the $R$-operator \eqref{r123} with the subscripts of all parameters $a, b, c, d, e$ and operators $\uu, \ww$ changed from $1$, $2$, $3$ to $i$, $j$, $k$.
As the corollary of Theorem \ref{th:RLLL}, we obtain the following.

\begin{cor}\label{cor:YB}
For the monodromy matrices $\mathcal{T}_{G}(x,y)$ and $\mathcal{T}_{G'}(x,y)$, it holds that
\begin{align}\label{YBM}
\mathcal{T}_{G}(x,y) = \mathscr{R}_{ijk} \mathcal{T}_{G'}(x,y)\mathscr{R}_{ijk}^{-1}.
\end{align}
In particular, when \( G \) is admissible, the transfer matrices associated with \( G' \) commute:
\[
[T_{G'}(x,y), \, T_{G'}(u,w)] = 0
\]
\end{cor}  

When $G$ is admissible, it is easy to see that the wiring diagram $G'$ obtained from $G$ by the YB move is also admissible. Thus this corollary is consistent with Theorem \ref{thm:TT2}.

It is of interest to find a sufficient condition for $G$ to be admissible.
We call a face in a wiring diagram {\it oriented} if its boundary is oriented.
We propose the following conjecture.

\begin{conjecture}\label{thm:adm}
If a wiring diagram $G$ on a torus includes no oriented faces, then it is admissible. 
\end{conjecture}

In Figure \ref{fig:Gs}, the admissible diagram on the left contains no oriented faces, whereas the non-admissible diagram on the right contains two oriented faces marked with $\sharp$ and \# respectively.
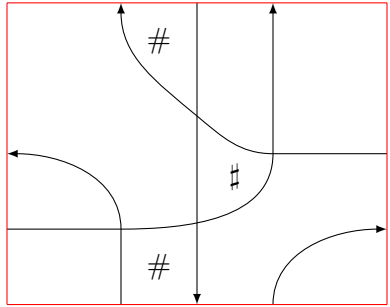
\begin{align*}
\begin{tikzpicture}
\begin{scope}[>=latex,xshift=0pt]
\draw [<-] (2.5,0)--(2.5,4);
\draw [->] (1.5,0)--(1.5,1) to [out=90,in=0](0,2);
\draw [->] (5,2)--(3.5,2) to [out=180,in=-40](2.5,2.5) to [out=140,in=-90](1.5,4);
\draw [->] (3.5,0) to [out=90,in=180](5,1);
\draw [->] (0,1)--(1.5,1) to [out=0,in=-90](3.5,2)--(3.5,4);
{\color{red}
\draw [-] (0,4)--(5,4);
\draw [-] (0,0)--(5,0);
\draw [-] (0,0)--(0,4);
\draw [-] (5,0)--(5,4);
}
\draw (3,1.7) node {$\sharp$};
\draw (2,0.5) node {\#};
\draw (2,3.5) node {\#};
%
\end{scope}
\end{tikzpicture}
\end{align*}

It is easier to show the inverse of the above conjecture,
by using the tangle diagrams introduced in the next subsection.

\begin{prop}\label{prop:I}
If a wiring diagram $G$ on a torus includes at least one oriented face,  then it is non-admissible. 
\end{prop}

\subsection{Tangle diagrams for the eight transformations}

We introduce the following tangle diagrams corresponding to the eight transformations \eqref{2D-LM}, \eqref{2D-I} and \eqref{2D-iI} divided into three types as \eqref{types}.

\begin{align}\label{tan}
\begin{tikzpicture}
\begin{scope}[>=latex,xshift=0pt]
\draw (-2,0.75) node[right] {type A:};
{\color{green}
\draw [->,dashed] (0,1.5)--(2,1.5);
\draw [->,dashed] (0,0)--(2,0);
}
\draw [->] (0.5,0)--(1.5,1.5);
\draw [->] (0.5,1.5)--(1.5,0);
\draw (1,-0.5) node {(v)};
{\color{green}
\draw [->,dashed] (3,1.5)--(5,1.5);
\draw [->,dashed] (3,0)--(5,0);
}
\draw [->] (3.5,0) to [out=90,in=180] (4,0.5) to [out=0,in=90] (4.5,0);
\draw (4,-0.5) node {(iI)};
{\color{green}
\draw [->,dashed] (6,1.5)--(8,1.5);
\draw [->,dashed] (6,0)--(8,0);
}
\draw [->] (6.5,1.5) to [out=-90,in=180] (7,1) to [out=0,in=-90] (7.5,1.5);
\draw (7,-0.5) node {(I)};
\end{scope}
\begin{scope}[>=latex,yshift=-80pt]
\draw (-2,0.75) node[right] {type B:};
{\color{green}
\draw [->,dashed] (0,1.5)--(2,1.5);
\draw [->,dashed] (0,0)--(2,0);
}
\draw [<-] (0.5,0)--(1.5,1.5);
\draw [<-] (0.5,1.5)--(1.5,0);
\draw (1,-0.5) node {(h)};
{\color{green}
\draw [->,dashed] (3,1.5)--(5,1.5);
\draw [->,dashed] (3,0)--(5,0);
}
\draw [<-] (3.5,0) to [out=90,in=180] (4,0.5) to [out=0,in=90] (4.5,0);
\draw (4,-0.5) node {(iI')};
{\color{green}
\draw [->,dashed] (6,1.5)--(8,1.5);
\draw [->,dashed] (6,0)--(8,0);
}
\draw [<-] (6.5,1.5) to [out=-90,in=180] (7,1) to [out=0,in=-90] (7.5,1.5);
\draw (7,-0.5) node {(I')};
\end{scope}
\begin{scope}[>=latex,yshift=-160pt]
\draw (-2,0.75) node[right] {type C:};
{\color{green}
\draw [->,dashed] (0,1.5)--(2,1.5);
\draw [->,dashed] (0,0)--(2,0);
}
\draw [->] (0.5,0)--(1.5,1.5);
\draw [<-] (0.5,1.5)--(1.5,0);
\draw (1,-0.5) node {(o)};
{\color{green}
\draw [->,dashed] (3,1.5)--(5,1.5);
\draw [->,dashed] (3,0)--(5,0);
}
\draw [<-] (3.5,0)--(4.5,1.5);
\draw [->] (3.5,1.5)--(4.5,0);
\draw (4,-0.5) node {(t)};
\end{scope}
\end{tikzpicture}
\end{align}
Here the green dashed arrows correspond to green thick arrows in \eqref{2D-LM}, \eqref{2D-I} and \eqref{2D-iI}, and the moves of thick arrows from NE to SW
are denoted by dashed arrows from top to bottom in \eqref{tan}.

We have the following equivalences of diagrams.

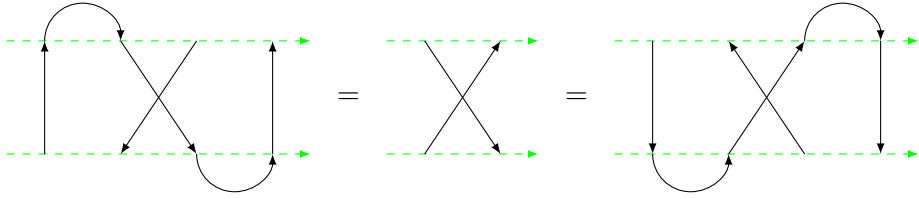
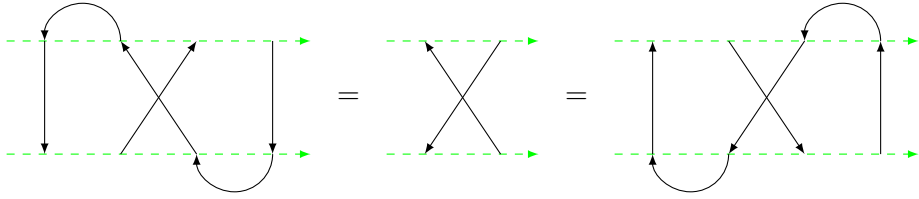
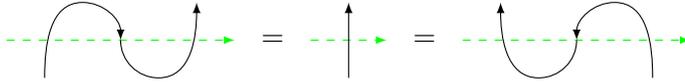
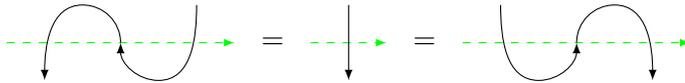
\begin{align}\label{tan-equiv}
\begin{tikzpicture}
\begin{scope}[>=latex,xshift=0pt]
{\color{green}
\draw [->,dashed] (0,1.5)--(4,1.5);
\draw [->,dashed] (0,0)--(4,0);
}
\draw [->] (0.5,0)--(0.5,1.5);
\draw [->] (0.5,1.5) to [out=90,in=180] (1,2) to [out=0,in=90] (1.5,1.5);
\draw [->] (1.5,1.5)--(2.5,0);
\draw [->] (2.5,1.5)--(1.5,0);
\draw [->] (2.5,0) to [out=-90,in=180] (3,-0.5) to [out=0,in=-90] (3.5,0);
\draw [->] (3.5,0)--(3.5,1.5);
\draw (4.5,0.75) node {=};
{\color{green}
\draw [->,dashed] (5,1.5)--(7,1.5);
\draw [->,dashed] (5,0)--(7,0);
}
\draw [->] (5.5,0)--(6.5,1.5);
\draw [->] (5.5,1.5)--(6.5,0);
\draw (7.5,0.75) node {=};
{\color{green}
\draw [->,dashed] (8,1.5)--(12,1.5);
\draw [->,dashed] (8,0)--(12,0);
}
\draw [->] (8.5,0) to [out=-90,in=180] (9,-0.5) to [out=0,in=-90] (9.5,0);
\draw [->] (8.5,1.5)--(8.5,0);
\draw [->] (9.5,0)--(10.5,1.5);
\draw [->] (10.5,0)--(9.5,1.5);
\draw [->] (11.5,1.5)--(11.5,0);
\draw [->] (10.5,1.5) to [out=90,in=180] (11,2) to [out=0,in=90] (11.5,1.5);
\draw (-0.5,2.5) node[right] {(iI)-(t)-(I) = (v) = (I)-(o)-(iI):};
\end{scope}
\begin{scope}[>=latex,yshift=-100pt]
{\color{green}
\draw [->,dashed] (0,1.5)--(4,1.5);
\draw [->,dashed] (0,0)--(4,0);
}
\draw [<-] (0.5,0)--(0.5,1.5);
\draw [<-] (0.5,1.5) to [out=90,in=180] (1,2) to [out=0,in=90] (1.5,1.5);
\draw [<-] (1.5,1.5)--(2.5,0);
\draw [<-] (2.5,1.5)--(1.5,0);
\draw [<-] (2.5,0) to [out=-90,in=180] (3,-0.5) to [out=0,in=-90] (3.5,0);
\draw [<-] (3.5,0)--(3.5,1.5);
\draw (4.5,0.75) node {=};
{\color{green}
\draw [->,dashed] (5,1.5)--(7,1.5);
\draw [->,dashed] (5,0)--(7,0);
}
\draw [<-] (5.5,0)--(6.5,1.5);
\draw [<-] (5.5,1.5)--(6.5,0);
\draw (7.5,0.75) node {=};
{\color{green}
\draw [->,dashed] (8,1.5)--(12,1.5);
\draw [->,dashed] (8,0)--(12,0);
}
\draw [<-] (8.5,0) to [out=-90,in=180] (9,-0.5) to [out=0,in=-90] (9.5,0);
\draw [<-] (8.5,1.5)--(8.5,0);
\draw [<-] (9.5,0)--(10.5,1.5);
\draw [<-] (10.5,0)--(9.5,1.5);
\draw [<-] (11.5,1.5)--(11.5,0);
\draw [<-] (10.5,1.5) to [out=90,in=180] (11,2) to [out=0,in=90] (11.5,1.5);
\draw (-0.5,2.5) node[right] {(iI')-(o)-(I') = (h) = (I')-(t)-(iI'):};
\end{scope}
\begin{scope}[>=latex,yshift=-160pt]
{\color{green}
\draw [->,dashed] (0,0)--(3,0);
}
\draw [->] (0.5,-0.5) to [out=90,in=180] (1,0.5) to [out=0,in=90] (1.5,0);
\draw [->] (1.5,0) to [out=-90,in=180] (2,-0.5) to [out=0,in=-90] (2.5,0.5);
\draw (3.5,0) node {=};
{\color{green}
\draw [->,dashed] (4,0)--(5,0);
}
\draw [<-] (4.5,0.5)--(4.5,-0.5);
\draw (5.5,0) node {=};
{\color{green}
\draw [->,dashed] (6,0)--(9,0);
}
\draw [<-] (6.5,0.5) to [out=-90,in=180] (7,-0.5) to [out=0,in=-90] (7.5,0);
\draw [<-] (7.5,0) to [out=90,in=180] (8,0.5) to [out=0,in=90] (8.5,-0.5);
\draw (-0.5,1) node[right] {(iI)-(I) = Id = (I')-(iI'):};
\end{scope}
\begin{scope}[>=latex,yshift=-230pt]
{\color{green}
\draw [->,dashed] (0,0)--(3,0);
}
\draw [<-] (0.5,-0.5) to [out=90,in=180] (1,0.5) to [out=0,in=90] (1.5,0);
\draw [<-] (1.5,0) to [out=-90,in=180] (2,-0.5) to [out=0,in=-90] (2.5,0.5);
\draw (3.5,0) node {=};
{\color{green}
\draw [->,dashed] (4,0)--(5,0);
}
\draw [->] (4.5,0.5)--(4.5,-0.5);
\draw (5.5,0) node {=};
{\color{green}
\draw [->,dashed] (6,0)--(9,0);
}
\draw [->] (6.5,0.5) to [out=-90,in=180] (7,-0.5) to [out=0,in=-90] (7.5,0);
\draw [->] (7.5,0) to [out=90,in=180] (8,0.5) to [out=0,in=90] (8.5,-0.5);
\draw (-0.5,1) node[right] {(iI')-(I') = Id = (I)-(iI):};
\end{scope}
\end{tikzpicture}
\end{align}
Using these equivalences, any tangle diagram is transformed into that containing neither (v) nor (h). 

For a wiring diagram $G$ on a torus with a fixed fundamental domain,  
let $D_G$ denote the tangle diagram associated with $G$,  
obtained by gluing the parts in \eqref{tan} as simply as possible.
Here, “simply” means that local configurations such as (I)-(iI), (I')-(iI'), etc.\  
are removed as much as possible by applying the third and fourth relations in \eqref{tan-equiv}.
Note that the diagram $D_G$ is not unique in general.

See Figure \ref{tan-ex} for the tangle diagram corresponding to the admissible transformation of Figure \ref{fig:ad-G}.

\begin{figure}[H]
\[
\begin{tikzpicture}
\begin{scope}[>=latex,xshift=0pt]
{\color{green}
\draw [->,dashed] (0,0)--(7,0);
\draw [->,dashed] (0,1)--(7,1);
\draw [->,dashed] (0,2)--(7,2);
\draw [->,dashed] (0,3)--(7,3);
\draw [->,dashed] (0,4)--(7,4);
\draw [->,dashed] (0,5)--(7,5);
}
\draw [->] (1,4)--(1,5);
\draw [->] (2,4)--(2,5);
\draw [<-] (3,4)--(5,5);
\draw [->] (5,4)--(3,5);
\draw [->] (6,4)--(6,5);
\draw [->] (1,3)--(1,4);
\draw [<-] (2,3)--(3,4);
\draw [->] (3,3)--(2,4);
\draw [->] (5,3)--(5,4);
\draw [->] (6,3)--(6,4);
\draw [<-] (1,3) to [out=-95,in=-85] (2,3);
\draw [->] (3,2)--(5,3);
\draw [->] (5,2)--(3,3);
\draw [->] (6,2)--(6,3);
\draw [<-] (1,1) to [out=95,in=85] (2,1);
\draw [->] (3,1)--(3,2);
\draw [->] (5,1)--(5,2);
\draw [->] (6,1)--(6,2);
\draw [<-] (1,0)--(1,1);
\draw [->] (2,0)--(3,1);
\draw [->] (3,0)--(2,1);
\draw [->] (5,0)--(5,1);
\draw [->] (6,0)--(6,1);
\draw (5,4.5) node {(h)};
\draw (3,3.5) node {(h)};
\draw (0.8,2.5) node {(I')};
\draw (5,2.5) node {(o)};
\draw (0.8,1.5) node {(iI')};
\draw (3,0.5) node {(o)};
\end{scope}
\end{tikzpicture}
\]
\caption{The tangle description of the admissible transformation of Figure \ref{fig:ad-G}.}
\label{tan-ex}
\end{figure}
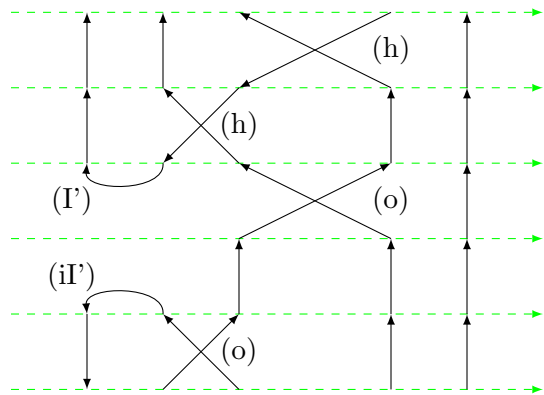

\begin{proof}(Proposition \ref{prop:I})
Assume that a wiring diagram $G$ on a torus has an oriented face. Due to Proposition \ref{prop:inv} (1), we can assume that the boundary of this oriented face (the cycle, for short) does not cross the red rectangle corresponding to the fixed fundamental domain of $G$. 
Let $D_G$ be the tangle diagram of $G$. If it has (v) or (h), apply the equivalence relations \eqref{tan-equiv} so that neither (v) nor (h) is included. Make the resulting diagram simple.
Since the cycle contains at least one vertex, it has at least one (t) or (o) at its left or right side.
Assume that it has (t) at the left side, so it is counterclockwise as the following schematic diagram, where dashed black arrows denote the cycle with (t). 
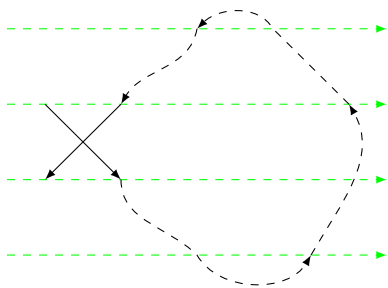
\begin{align*}
\begin{tikzpicture}
\begin{scope}[>=latex,xshift=0pt]
{\color{green}
\draw [->,dashed] (0,2)--(5,2);
\draw [->,dashed] (0,1)--(5,1);
\draw [->,dashed] (0,0)--(5,0);
\draw [->,dashed] (0,-1)--(5,-1);
}
\draw [->,dashed] (1.5,0) to [out=-90,in=135] (2.5,-1) to [out=-60,in=-120] (4,-1);
\draw [->,dashed] (4,-1) to [out=60,in=-60] (4.5,1);
\draw [->,dashed] (4.5,1) to [out=135,in=-45] (3.5,2) to [out=120,in=45] (2.5,2); 
\draw [->,dashed] (2.5,2) to [out=-100,in=60] (1.5,1);
%
\draw [->] (0.5,1)--(1.5,0);
\draw [->] (1.5,1)--(0.5,0);
\end{scope}
\end{tikzpicture}
\end{align*}
Thus, at the top side it must have (I') or (iI'), and at the bottom side it must have (I) or (iI). Hence the diagram is non-admissible.
The other cases of (t) at right side, and (o) at left or right side, are obtained by right-left or top-bottom reflection of this case.   
\end{proof}

\section{Examples}

We present two examples of the commuting transfer matrices $T_G(x,y)$ defined in \eqref{T-Gg}.

\subsection{Case 1}

The first example is the left wiring diagram of Figure \ref{fig:Gs}, our driving example. 
We assign $i_1, i_2, j_1, j_2 \in \{0,1\}$ to edges intersect with the red rectangle, and number vertices $1,2,3,4$ as follows. 
To each vertex $k$, we assign $\mathscr{L}(r_k, s_k,f_k,g_k;q)$ with a canonical pair $(\uu_k, \ww_k)$, where we consider an inhomogeneous transfer matrix (See  Remark \ref{rmk:inho2}).

\begin{align}\label{fig:G1}
\begin{tikzpicture}
\begin{scope}[>=latex,xshift=0pt]
\draw [->] (2.5,0)--(2.5,4) node[above]{$j_2$};
\draw [->] (1.5,0)--(1.5,1) to [out=90,in=0](0,2); 
\draw [->] (5,2) node[right]{$i_1$}--(3.5,2) to [out=180,in=-40](2.5,2.5) to [out=140,in=-90](1.5,4) node[above]{$j_1$};
\draw [->] (3.5,0) to [out=90,in=180](5,1)node[right]{$i_2$};
\draw [->] (0,1)--(1.5,1) to [out=0,in=-90](3.5,2)--(3.5,4)node[above]{$i_2$};
{\color{red}
\draw [-] (0,4)--(5,4);
\draw [-] (0,0)--(5,0);
\draw [-] (0,0)--(0,4);
\draw [-] (5,0)--(5,4);
}
\draw (1.5,1)node[below left]{$1$};
\draw (2.5,1.1)node[below left]{$2$};
\draw (2,2.4)node[right]{$3$};
\draw (3.5,2)node[below right]{$4$};
\end{scope}
\end{tikzpicture}
\end{align}

The transfer matrix $T_G(x,y)$ takes the form
\begin{align*}
T_G(x,y) &= \sum_{i_1,i_2,j_1,j_2}  T_{\mathrm{i},\mathrm{j}}^{\mathrm{i},\mathrm{j}}\, x^{-i_1+i_2}y^{j_1+j_2+i_2}
\\
&=r_1r_2r_3r_4 + s_1s_2s_3s_4 y^3 + 
e^{\uu_1 + \uu_3 + \uu_4} \left( \frac{g_1g_3g_4 r_2 y}{x} + f_1f_3f_4 s_2 x y^2 \right) 
\\
&\quad +e^{\uu_1 + \uu_2 + \uu_4} \left( \frac{g_1g_2g_4 s_3 y^2}{x} 
+ f_1f_2f_4 r_3 x y \right) 
 + A_{0,1}\, y + A_{0,2} \,y^2,
\displaybreak[0]
\\[1mm]
A_{0,1}&=e^{\uu_2 + \uu_3 + 2\uu_4 - \ww_1 + \ww_4} f_2 f_4 g_3 g_4 
+ e^{2\uu_4 - \ww_2 - \ww_3 + \ww_4} f_4 g_4 r_1 
+ e^{2\uu_1 + \ww_1 - \ww_4} f_1 g_1 r_2 r_3 \\[2mm]
&+ e^{2\uu_2 + 2\uu_3 - \ww_1 + \ww_2 + \ww_3} f_2 f_3 g_2 g_3 r_4 
+ e^{\uu_2 + \uu_3} f_3 g_2 r_1 r_4 
+ e^{\ww_1 - \ww_4} r_1 r_2 r_3 s_1 \\[2mm]
&+ e^{2\uu_3 - \ww_1 + \ww_2 + \ww_3} f_3 g_3 r_2 r_4 s_2 
+ e^{2\uu_2 - \ww_1 + \ww_2 + \ww_3} f_2 g_2 r_3 r_4 s_3 
+ e^{-\ww_1 + \ww_2 + \ww_3} r_2 r_3 r_4 s_2 s_3 \\[2mm]
&+ e^{\uu_2 + \uu_3 - \ww_1 + \ww_4} f_2 g_3 r_4 s_4 
+ e^{-\ww_2 - \ww_3 + \ww_4} r_1 r_4 s_4.
\displaybreak[0]
\\
A_{0,2} \;=\;&
e^{2\uu_1 + \uu_2 + \uu_3 + \ww_1 - \ww_4} \,f_1 f_3 g_1 g_2 
\;+\; e^{2\uu_2 + 2\uu_3 + \ww_2 + \ww_3 - \ww_4} \,f_2 f_3 g_2 g_3 s_1
\;+\; e^{\uu_2 + \uu_3 + \ww_1 - \ww_4} \,f_3 g_2 r_1 s_1 
\\[6pt]
&+\; e^{2\uu_3 + \ww_2 + \ww_3 - \ww_4} \,f_3 g_3 r_2 s_1 s_2 
\;+\; e^{2\uu_2 + \ww_2 + \ww_3 - \ww_4} \,f_2 g_2 r_3 s_1 s_3 
\;+\; e^{2\uu_4 - \ww_1 + \ww_4} \,f_4 g_4 s_2 s_3 
\\[6pt]
&+\; e^{\ww_2 + \ww_3 - \ww_4} \,r_2 r_3 s_1 s_2 s_3 
\;+\; e^{2\uu_1 + \ww_1 - \ww_2 - \ww_3} \,f_1 g_1 s_4 
\;+\; e^{\uu_2 + \uu_3} \,f_2 g_3 s_1 s_4 
\\[6pt]
&+\; e^{\ww_1 - \ww_2 - \ww_3} \,r_1 s_1 s_4 
\;+\; e^{-\ww_1 + \ww_4} \,r_4 s_2 s_3 s_4 \,.
%
\end{align*}
Since this $G$ is admissible, the coefficients of $T_G(x,y)$ commute with each other.

\subsection{Case 2}

The second example is as follows.  
We assign $i_1, j_1, j_2 \in \{0,1\}$ to the edges intersecting the red rectangle,  
and a canonical pair $(\uu_k, \ww_k)$ to each vertex $k$ (for $k = 1,2,3,4$).
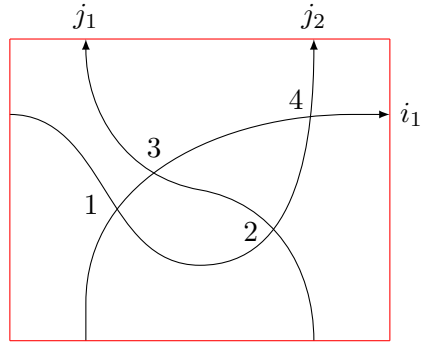
\begin{align}\label{fig:G2}
\begin{tikzpicture}
\begin{scope}[>=latex,xshift=0pt]
\draw [->] (1,0)--(1,0.5) to [out=90,in=180] (4.5,3)--(5,3) node[right]{$i_1$};
\draw [->] (0,3) to [out=0,in=180](2.5,1) to [out=0,in=-90] (4,4)node[above]{$j_2$};
\draw [->] (4,0) to [out=90, in=-10] (2.5,2) to [out=170,in=-90](1,4) node[above]{$j_1$};
{\color{red}
\draw [-] (0,4)--(5,4);
\draw [-] (0,0)--(5,0);
\draw [-] (0,0)--(0,4);
\draw [-] (5,0)--(5,4);
}
\draw (1.3,1.8)node[left]{$1$};
\draw (3.4,1.7)node[below left]{$2$};
\draw (1.9,2.3)node[above]{$3$};
\draw (4,3.2)node[left]{$4$};
\end{scope}
\end{tikzpicture}
\end{align}

We now consider the transfer matrix under a homogeneous choice of parameters.
It takes the following form:
\begin{align*}
  T_G(x,y) &= \sum_{i_1,j_1,j_2}  T_{\mathrm{i},\mathrm{j}}^{\mathrm{i},\mathrm{j}}\, x^{i_1}y^{j_1+j_2}
\\
&=r^4 + s^4 xy^2 + B_{0,1}\,y + B_{0,2}\,y^2+B_{1,0}\, x + B_{1,1}\, xy ,
\\[1mm]
B_{0,1} &=e^{2\uu_2 + \uu_3 - \ww_1 + \ww_2} f g^2 r 
+ e^{\uu_2 + 2\uu_4 - \ww_3 + \ww_4} f g^2 r 
+ e^{\uu_4 - \ww_2} g r^2 \\
& \quad + e^{\uu_1 + 2\uu_3 + \ww_3} f g^2 r^2 
+ e^{\uu_3 - \ww_1 + \ww_2} g r^2 s 
+ e^{\uu_2 - \ww_3 + \ww_4} g r^2 s 
+ e^{\uu_1 + \ww_3} g r^3 s,
\\
B_{0,2} &= e^{\uu_1 + 2\uu_3 + \uu_4 - \ww_2 + \ww_3} f g^3 
+ e^{\uu_3 + \uu_4 - \ww_1} g^2 s 
+ e^{\uu_1 + \uu_2 + 2\uu_4 + \ww_4} f g^3 s 
\\
& \quad + e^{\uu_1 + \uu_4 - \ww_2 + \ww_3} g^2 r s 
+ e^{\uu_1 + \uu_2 + \ww_4} g^2 r s^2.
\\
B_{1,0}&=e^{\uu_1 + 2\uu_2 + \uu_4 + \ww_2 - \ww_3} f^3 g 
+ e^{\uu_1 + \uu_2 - \ww_4} f^2 r 
+ e^{2\uu_1 + \uu_3 + \uu_4 + \ww_1} f^3 g r 
\\
& \quad + e^{\uu_1 + \uu_4 + \ww_2 - \ww_3} f^2 r s 
+ e^{\uu_3 + \uu_4 + \ww_1} f^2 r^2 s.
\\
B_{1,1} &=e^{2\uu_1 + \uu_3 + \ww_1 - \ww_2} f^2 g s 
+ e^{\uu_2 + 2\uu_3 + \ww_3 - \ww_4} f^2 g s 
+ e^{\uu_1 - \ww_3} f s^2 
+ e^{2\uu_2 + \uu_4 + \ww_2} f^2 g s^2 \\
& \quad + e^{\uu_3 + \ww_1 - \ww_2} f r s^2 
+ e^{\uu_2 + \ww_3 - \ww_4} f r s^2 
+ e^{\uu_4 + \ww_2} f r s^3.
\end{align*}
Again, this $G$ is admissible, and the coefficients of $T_G(x,y)$ commute.

\section{Free parafermion and relativistic Toda}

\subsection{Representations of $T_G(x,y)$}

Let $\mathcal{L}(r_k,s_k,f_k,g_k;q)$ ($k=1,\ldots, K$)
be the operator (\ref{L1})--(\ref{L3}) assigned to 
the vertex $k$ of an admissible graph $G$ on a torus.
We have constructed the transfer matrix $T_G(x,y)$ in the tensor product of 
$q$-Weyl algebra $\mathcal{W}(q)^{\otimes K}$.
Given a set of representations
\begin{align}
\rho_k : \: \mathcal{W}(q) \rightarrow \mathrm{End}(F_k)
\end{align}
for $k=1,\ldots, K$, we have the representation of the transfer matrix 
\begin{align}
\mathbb{T}_G(x,y)=(\rho_1\otimes \cdots \otimes \rho_{K})(T_G(x,y)) 
\in \mathrm{End}(F_1\otimes \cdots \otimes F_{K}),
\end{align} 
which satisfies the commutativity 
\begin{align}
[\mathbb{T}_G(x,y), \mathbb{T}_G(u,w)]=0
\end{align}
 by the construction.
A typical representation of $\mathcal{W}(q)$ is the one with 
$F = \bigoplus_{m \in \Z}\C(q)|m\rangle$, and 
\begin{align}
\text{$\uu$-diagonal representation}:& \;  
e^{\pm \uu} |m\rangle = q^{\pm m}|m\rangle,
\quad 
e^{\pm\ww} |m\rangle = |m \pm 1\rangle,
\label{urep}\\
\text{$\ww$-diagonal representation}:& \;  
e^{\pm \uu} |m\rangle = |m\mp 1\rangle,
\quad 
e^{\pm\ww} |m\rangle = q^{\pm m}|m\rangle.
\label{wrep}
\end{align}
It can further be pulled back by any automorphism of $\mathcal{W}(q)$ in \cite[Rem.5.7]{IKSTY}
specified by an element of $\mathrm{SL}(2,\Z)$ depending on $k=1,\ldots, K$.
(The above two representations are connected in this way. 
The representation of $\mathcal{W}(-q)$ in (\ref{uwp}) is obtained similarly with 
the replacement $q \rightarrow -q$.)

According to the prescription explained in (\ref{Malpha})--(\ref{posc}),
one can also consider a representation of $\mathcal{L}(r_k,s_k,f_k,g_k;q)$ on the nonnegative mode subspace 
$\bigoplus_{m \in \Z_{\ge 0}}\C(q)|m\rangle$, which factors through the 
$q$-oscillator algebra by imposing the constraint $f_kg_k = -q^{-1}r_ks_k$ on  the parameters $r_k,s_k,f_k,g_k$.

All the representations mentioned thus far are infinite dimensional.
One can construct a finite dimensional $\mathbb{T}_G(x,y)$
by specializing $q$ to a root of unity.
Let $N\ge 2$ be an  integer and $q$ be an $N$ th roof of unity.
Then $\mathcal{W}(q)$ has an $N$ dimensional representation obtained 
by regarding $m$ in (\ref{urep})--(\ref{wrep}) as taking values in $\Z_N := \Z/ N \Z$.
The ``smallest" choice is $N=2$, where $e^\uu$ and $e^\ww$ in (\ref{urep}) for example
can be identified with the Pauli matrices $\sigma^z$ and $\sigma^x$, respectively.
The resulting transfer matrix $\mathbb{T}_G(x,y) \in \mathrm{End}((\C^2)^{\otimes K})$
remains commutative, and defines an integrable ``two-dimensional XZ model" with local weights
\begin{align}
\mathscr{L}^{00}_{00} = r,\;\;  \mathscr{L}^{11}_{11} = s,\;\; \mathscr{L}^{10}_{10} = f \sigma^z,\;\;
\mathscr{L}^{01}_{01} = g \sigma^z, 
\;\;  \mathscr{L}^{10}_{01} = \sigma^x , \;\; \mathscr{L}^{01}_{10} =(rs-fg)\sigma^x.
\end{align}

\subsection{Free parafermion}\label{ss:FP}

As an application of the construction in the previous subsection, 
we demonstrate that the transfer matrix $\mathbb{T}_G(x,y)$, for a suitable choice of $G$ and parameters, 
yields a generating function of 
commuting quantum Hamiltonians in the celebrated {\em free parafermion model}.
It is a non-Hermitian $\Z_N$ spin chain on one dimensional lattice 
introduced in \cite{B89} through its connection to the Chiral Potts model.
The model is known to exhibit a remarkable spectrum that generalizes the $N=2$ quantum Ising model.
Further developments on the model are discussed in~\cite{F14, AP20},  
and a recent survey is available in~\cite{BHL23}.

For an integer $L \ge 2$, 
let $G$ be a $2$ by $L$ square lattice whose arrows and vertices are 
specified as in Figure \ref{fig:FP1}. 
The graph $G$ is admissible.
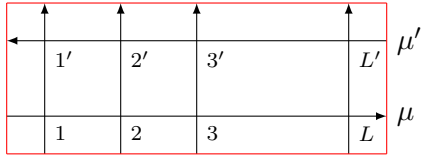
\begin{figure}[H]
\[
\begin{tikzpicture}
\begin{scope}[>=latex,xshift=0pt]
\draw [<-] (-0.5,1)--(4.5,1);
\draw [->] (-0.5,0)--(4.5,0);
\draw [->] (0,-0.5)--(0,1.5);
\draw [->] (1,-0.5)--(1,1.5);
\draw [->] (2,-0.5)--(2,1.5);
\draw [->] (4,-0.5)--(4,1.5);
%
\draw(4.5,1) node[right]{$\mu'$};
\draw(4.5,0) node[right]{$\mu$};
{\color{red}
\draw [-] (-0.5,1.5)--(4.5,1.5);
\draw [-] (-0.5,-0.5)--(4.5,-0.5);
\draw [-] (-0.5,-0.5)--(-0.5,1.5);
\draw [-] (4.5,-0.5)--(4.5,1.5);
}
\draw(0,1) node[below right]{$\scriptstyle{1'}$} ; \draw(0,0) node[below right]{$\scriptstyle{1}$} ;
\draw(1,1) node[below right]{$\scriptstyle{2'}$} ; \draw(1,0) node[below right]{$\scriptstyle{2}$} ;
\draw(2,1) node[below right]{$\scriptstyle{3'}$} ; \draw(2,0) node[below right]{$\scriptstyle{3}$} ;
\draw(4,1) node[below right]{$\scriptstyle{L'}$} ; \draw(4,0) node[below right]{$\scriptstyle{L}$} ;

\end{scope}
\end{tikzpicture}
\]
\caption{An admissible graph $G$ related to free parafermion model.
Vertices are labeled with $1,\ldots, L$ and $1',\ldots, L'$. 
The parameters $\mu$ and $\mu'$ are boundary magnetic fields mentioned in Remark \ref{re:inho}.}
\label{fig:FP1}
\end{figure}

Let 
 $\LL(r_i,s_i, f_i,g_i ;q)$  and  $\LL(r_{i'},s_{i'}, f_{i'}, g_{i'} ;q)$
 be the operators (\ref{L1}) associated with the vertices $i$ and $i'$ of $G$, respectively.
We focus on a quantized {\em five} vertex model corresponding to the following specialization of the parameters:
\begin{align}\label{fg0}
\LL(r_i=1,s_i=1, f_i=0, g_i=g ;q),\quad 
\LL(r_{i'}=1,s_{i'}=1, f_{i'}=f, g_{i'}=0 ;q)\quad (1 \le i \le L).
\end{align}
Here the key condition is 
$f_i=0$ and $g_{i'}=0$, while the homogeneous choice of the remaining parameters is made purely for  
simplicity of presentation. 
In fact, the argument below extends straightforwardly to the inhomogeneous case.
More importantly, we further introduce the ``boundary magnetic fields"  $\mu, \mu'$ explained in Remark \ref{re:inho} (ii) 
on the two rightmost horizontal edges as shown in Figure \ref{fig:FP1}.

For simplicity we first consider the case $\mu=\mu'=0$.
Then, the local states 
on the rightmost (and hence also the leftmost, due to the periodic boundary conditions)
edges are restricted to zero. 
With these settings,  the transfer matrix $T_G(x,y)$ 
becomes independent of $x$, allowing us to denote it as $T_G(y)$.
It is a polynomial of degree $L$ in $y$, which still satisfies the commutativity $[T_G(y), T_G(y')]=0$.
The coefficient of $y^1$ of $T_G(y)$ consists of $2L-1$ configurations.
For $L=3$, they are given as Figure~\ref{fig:FP2}.
  
\begin{figure}[H]
\[
\begin{tikzpicture}
\begin{scope}[>=latex,xshift=0pt]
\draw [<-] (-0.4,0.6)--(1.6,0.6);
\draw [->] (-0.4,0)--(1.6,0);
\draw [very thick,  ->] (0,-0.4)--(0,1);
\draw [->] (0.6,-0.4)--(0.6,1);
\draw [->] (1.2,-0.4)--(1.2,1);
{\color{red}
\draw [-] (-0.4,1)--(1.6,1);
\draw [-] (-0.4,-0.4)--(1.6,-0.4);
\draw [-] (-0.4,-0.4)--(-0.4,1);
\draw [-] (1.6,-0.4)--(1.6,1);
}
\draw (-0.4,-0.8) node[right] {$h_1=f g X_1$};
\end{scope}
\end{tikzpicture}
\qquad
\begin{tikzpicture}
\begin{scope}[>=latex,xshift=0pt]
\draw [<-] (-0.4,0.6)--(1.6,0.6);
\draw [->] (-0.4,0)--(1.6,0);
\draw [->] (0,-0.4)--(0,1);
\draw [very thick,  ->] (0.6,-0.4)--(0.6,1);
\draw [->] (1.2,-0.4)--(1.2,1);
{\color{red}
\draw [-] (-0.4,1)--(1.6,1);
\draw [-] (-0.4,-0.4)--(1.6,-0.4);
\draw [-] (-0.4,-0.4)--(-0.4,1);
\draw [-] (1.6,-0.4)--(1.6,1);
}
\draw (-0.4,-0.8)  node[right] {$h_3=fgX_2$};
\end{scope}
\end{tikzpicture}
\qquad 
\begin{tikzpicture}
\begin{scope}[>=latex,xshift=0pt]
\draw [<-] (-0.4,0.6)--(1.6,0.6);
\draw [->] (-0.4,0)--(1.6,0);
\draw [->] (0,-0.4)--(0,1);
\draw [->] (0.6,-0.4)--(0.6,1);
\draw  [very thick,  ->](1.2,-0.4)--(1.2,1);
{\color{red}
\draw [-] (-0.4,1)--(1.6,1);
\draw [-] (-0.4,-0.4)--(1.6,-0.4);
\draw [-] (-0.4,-0.4)--(-0.4,1);
\draw [-] (1.6,-0.4)--(1.6,1);
}
\draw (-0.4,-0.8)  node[right] {$h_5=fg X_3$};
\end{scope}
\end{tikzpicture}
\qquad 
\begin{tikzpicture}
\begin{scope}[>=latex,xshift=0pt]
\draw [<-] (-0.4,0.6)--(1.6,0.6);
\draw [->] (-0.4,0)--(1.6,0);
\draw [->] (0,-0.4)--(0,1);
\draw [->] (0.6,-0.4)--(0.6,1);
\draw [->] (1.2,-0.4)--(1.2,1);
\draw [very thick,  -](0,-0.4)--(0,0);
\draw [very thick,  -](0,0)--(0.6,0);
\draw [very thick,  -](0.6,0)--(0.6,0.6);
\draw [very thick,  -](0.6,0.6)--(0,0.6);
\draw [very thick,  ->](0,0.6)--(0,1);
{\color{red}
\draw [-] (-0.4,1)--(1.6,1);
\draw [-] (-0.4,-0.4)--(1.6,-0.4);
\draw [-] (-0.4,-0.4)--(-0.4,1);
\draw [-] (1.6,-0.4)--(1.6,1);
}
\end{scope}
\draw (-0.5,-0.8) node[right] {$h_2=Z_1^{-1}Z_2$};
\end{tikzpicture}
\qquad 
\begin{tikzpicture}
\begin{scope}[>=latex,xshift=0pt]
\draw [<-] (-0.4,0.6)--(1.6,0.6);
\draw [->] (-0.4,0)--(1.6,0);
\draw [->] (0,-0.4)--(0,1);
\draw [->] (0.6,-0.4)--(0.6,1);
\draw [->] (1.2,-0.4)--(1.2,1);
\draw [very thick, -](0.6,-0.4)--(0.6,0);
\draw [very thick, -](0.6,0)--(1.2,0);
\draw [very thick, -](1.2,0)--(1.2,0.6);
\draw [very thick, -](1.2,0.6)--(0.6,0.6);
\draw [very thick, ->](0.6,0.6)--(0.6,1);
{\color{red}
\draw [-] (-0.4,1)--(1.6,1);
\draw [-] (-0.4,-0.4)--(1.6,-0.4);
\draw [-] (-0.4,-0.4)--(-0.4,1);
\draw [-] (1.6,-0.4)--(1.6,1);
}
\draw (-0.5,-0.8) node[right] {$h_4=Z_2^{-1}Z_3$};
\end{scope}
\end{tikzpicture}
\]
\caption{Configurations of a q-5v model contributing $y^1$ term to $T_G(y)$ for $L=3$.
Thick (resp. thin) edges correspond to the local state $1$ (resp. $0$).}
\label{fig:FP2}
\end{figure}
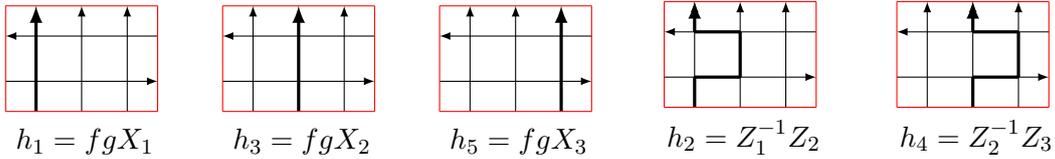

\noindent
Here we have set 
\begin{align}\label{XZ}
X_i = e^{\uu_i+\uu_{i'}},\quad 
Z_i = \e^{\ww_i+\ww_{i'}},
\end{align}
which satisfy the commutation relation
\begin{align}\label{qXZ}
X_i Z_j = q^{2\delta_{i,j}}Z_jX_i.
\end{align}
For general $L$, we similarly define $h_1, \ldots, h_{2L-1}$ as
\begin{align}
h_{2i-1} = fg X_{i},\quad h_{2i} = Z_i^{-1}Z_{i+1}
\end{align}
by using (\ref{XZ}).
They satisfy the commutation relations
\begin{align}\label{hh}
h_a h_{a+1} = q^{-2}h_{a+1}h_a,\quad h_a h_b = h_b h_a \; \,\text{for }\, | a-b | >1.
\end{align}
From Figure \ref{fig:6v} and the specialization (\ref{fg0}),
it is not difficult to see that the transfer matrix is expressed as 
\begin{align}
T_G(y) &= \sum_{m=0}^L y^m J^{(m)},
\label{Tfp}\\
J^{(m)}&= \sum_{1 \le b_i \prec \cdots \prec b_m \le 2L-1} h_{b_1} \cdots h_{b_m} \qquad (J^{(0)}=1),
\label{Jfp}
\end{align} 
where $b\prec b'$ means $b+2 \le b'$.
The commutativity of $T_G(y)$ implies  $[J^{(m)}, J^{(m')}]=0$ for all $1 \le m,m' \le L$.
In particular, regarding 
\begin{align}\label{HJ1}
H:=J^{(1)} = \sum_{a=1}^{2L-1}h_a = \sum_{i=1}^{L-1} Z_i^{-1}Z_{i+1} 
+ fg \sum_{i=1}^LX_i
\end{align}
as the Hamiltonian of the system, $J^{(2)}, \ldots, J^{(L)}$ serve as mutually commuting 
(non-local) conserved quantities. 

The construction thus far involves $2L$ copies of the $q$-Weyl algebra, generated by 
$e^{\pm \uu_i}, e^{\pm \ww_i}$, $e^{\pm \uu_{i'}},  e^{\pm \ww_{i'}}\, (1 \le i \le L)$,
which are assigned to the vertices of $G$ shown in Figure \ref{fig:FP1}.
However, the expressions (\ref{Tfp})--(\ref{Jfp}) reveal that 
$T_G(y)$ actually depends only on its subalgebra generated by 
$X_i^{\pm 1}, Z^{\pm 1}_i\,(1 \le i \le L)$ defined in (\ref{XZ}).
This subalgebra is isomorphic to 
$\mathcal{W}(q^2)^{\otimes L}$;  See (\ref{qXZ}).    
Thus, we may consider a representation $\mathbb{T}_G(y)$ of $T_G(y)$ 
on the space $\sum_{{\bf m} \in \Z^L}\C|{\bf m} \rangle$,
where $Z_i$ acts diagonally:
\begin{align}\label{XZL}
X_i |{\bf m}\rangle = |{\bf m} -{\bf e}_i\rangle,
\qquad 
Z_i |{\bf m}\rangle = q^{2m_i}|{\bf m}\rangle,
\end{align}
where ${\bf m} = (m_1,\ldots, m_n) \in \Z^L$ and 
${\bf e}_i = (0,\ldots, \overset{i}{1},\ldots, 0)$ is the $i$th unit vector in $\Z^L$.

The Hamiltonian of the free parafermion model is obtained by specializing $q^2$ to a primitive $N$th root of unity 
in (\ref{HJ1}), and regarding the integer array ${\bf m} \in \Z^L$ as 
${\bf m} \in (\Z_N)^L$.
The state space of the model is then naturally identified with the $N^L$ dimensional space
$\sum_{{\bf m} \in (\Z_N)^L}\C|{\bf m} \rangle$.
Up to convention, our $H$ (\ref{HJ1}) and  $J^{(m)}$ (\ref{Jfp}) reproduce \cite[eq. (1)]{BHL23} and 
\cite[eq.(41)]{F14}, respectively.
Compare also (\ref{hh}) and \cite[eq.(39)]{F14}.

Now let us consider the effect of reintroducing the magnetic fields $\mu$ and $\mu'$.
Figure \ref{fig:FP3} illustrates some representative additional configurations that contribute to 
$T_G(x,y)$ for $L=3$.

\begin{figure}[H]
\[
\begin{tikzpicture}
\begin{scope}[>=latex,xshift=0pt]
\draw [<-] (-0.4,0.6)--(1.6,0.6);
\draw [->] (-0.4,0)--(1.6,0);
\draw [very thick,  -] (1.2,-0.4)--(1.2,0);
\draw [very thick,  ->] (1.2,0)--(1.6,0);
\draw [very thick,  -] (-0.3,0)--(0,0);
\draw [very thick,  -] (0,0)--(0,0.6);
\draw [very thick,  ->] (0,0.6)--(-0.4,0.6);
\draw [very thick,  -] (1.6,0.6)--(1.2,0.6);
\draw [very thick,  ->] (1.2,0.6)--(1.2,1);
\draw [->] (0,-0.4)--(0,1);
\draw [->] (0.6,-0.4)--(0.6,1);
\draw [->] (1.2,-0.4)--(1.2,1);
{\color{red}
\draw [-] (-0.4,1)--(1.6,1);
\draw [-] (-0.4,-0.4)--(1.6,-0.4);
\draw [-] (-0.4,-0.4)--(-0.4,1);
\draw [-] (1.6,-0.4)--(1.6,1);
}
\draw (-0.8,-0.8) node[right] {$h_6=\mu\mu' Z_3^{-1}Z_1$};
\end{scope}
\end{tikzpicture}
\qquad
\begin{tikzpicture}
\begin{scope}[>=latex,xshift=0pt]
\draw [<-] (-0.4,0.6)--(1.6,0.6);
\draw [->] (-0.4,0)--(1.6,0);
\draw [->] (0,-0.4)--(0,1);
\draw [very thick,  ->] (1.6,0.6)--(-0.4,0.6);
\draw [very thick,  ->] (0,-0.4)--(0,1);
\draw [very thick,  ->] (0.6,-0.4)--(0.6,1);
\draw [very thick,  ->] (1.2,-0.4)--(1.2,1);
\draw [->] (1.2,-0.4)--(1.2,1);
{\color{red}
\draw [-] (-0.4,1)--(1.6,1);
\draw [-] (-0.4,-0.4)--(1.6,-0.4);
\draw [-] (-0.4,-0.4)--(-0.4,1);
\draw [-] (1.6,-0.4)--(1.6,1);
}
\draw (0.6,-0.8)  node{$\mu'g^3e^{\uu_1+\uu_2+\uu_3}$};
\end{scope}
\end{tikzpicture}
\qquad \;\;
\begin{tikzpicture}
\begin{scope}[>=latex,xshift=0pt]
\draw [<-] (-0.4,0.6)--(1.6,0.6);
\draw [->] (-0.4,0)--(1.6,0);
\draw [->] (0,-0.4)--(0,1);
\draw [->] (0.6,-0.4)--(0.6,1);
\draw [very thick,  <-] (1.6,0)--(-0.4,0);
\draw [very thick,  ->] (0,-0.4)--(0,1);
\draw [very thick,  ->] (0.6,-0.4)--(0.6,1);
\draw [very thick,  ->] (1.2,-0.4)--(1.2,1);
{\color{red}
\draw [-] (-0.4,1)--(1.6,1);
\draw [-] (-0.4,-0.4)--(1.6,-0.4);
\draw [-] (-0.4,-0.4)--(-0.4,1);
\draw [-] (1.6,-0.4)--(1.6,1);
}
\draw (0.6,-0.8)  node{$\mu f^3e^{\uu_{1'}+\uu_{2'}+\uu_{3'}}$};
\end{scope}
\end{tikzpicture}
\qquad \;\;
\begin{tikzpicture}
\begin{scope}[>=latex,xshift=0pt]
\draw [<-] (-0.4,0.6)--(1.6,0.6);
\draw [->] (-0.4,0)--(1.6,0);
\draw [->] (0,-0.4)--(0,1);
\draw [->] (0.6,-0.4)--(0.6,1);
\draw [very thick,  ->](1.6,0.6)--(-0.4,0.6);
\draw [very thick,  <-] (1.6,0)--(-0.4,0);
\draw [very thick,  ->] (0,-0.4)--(0,1);
\draw [very thick,  ->] (0.6,-0.4)--(0.6,1);
\draw [very thick,  ->] (1.2,-0.4)--(1.2,1);
{\color{red}
\draw [-] (-0.4,1)--(1.6,1);
\draw [-] (-0.4,-0.4)--(1.6,-0.4);
\draw [-] (-0.4,-0.4)--(-0.4,1);
\draw [-] (1.6,-0.4)--(1.6,1);
}
\draw (0.6,-0.8)  node{$\mu\mu'$};
\end{scope}
\end{tikzpicture}
\]
\caption{Examples of extra configurations allowed for nonzero $\mu, \mu'$ for $L=3$.
Weights are shown without the spectral parameters $x,y$.}
\label{fig:FP3}
\end{figure}
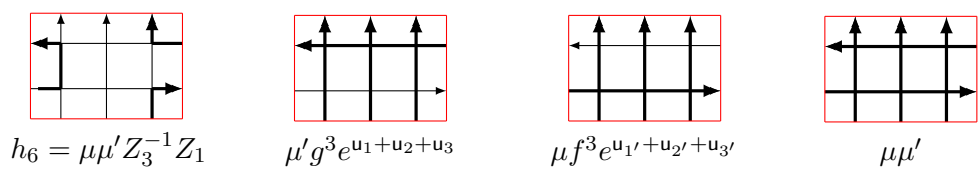

\noindent
For general $L$, we define $h_{2L}= \mu\mu' Z^{-1}_LZ_1$.
The commutation relations among $h_1,\ldots, h_{2L}$ are given by (\ref{hh}),
this time interpreting the indices of $h_a$ as elements of $\Z_{2L}$.
The commuting transfer matrix $T_G(x,y)$ in the presence of  nonzero magnetic fields is given by 
\begin{align}
T_G(x,y) &=  x^{-1}y^L \mu' g^L e^{\uu_1+\cdots + \uu_L} + 
x y^L\mu f^L e^{\uu_{1'}+\cdots + \uu_{L'}}+ 
\sum_{m=0}^Ly^m \tilde{J}^{(m)},
\label{fptg1}
\\
\tilde{J}^{(m)} &= \sum_{\substack{1 \le b_1 \prec \cdots \prec b_m \le 2L
\\  (b_1,b_L)\neq (1,2L)}}h_{b_1}\cdots h_{b_m}.
\label{fptg2}
\end{align}
Notably, the contribution $\mu\mu'$ from the rightmost configuration in Figure \ref{fig:FP3}
has been incorporated as the term $h_2h_4\cdots h_{2L}$ in $\tilde{J}^{(L)}$.

An interesting direction for further study is to formulate 
analogous models for various admissible graphs $G$ and 
to explore whether there exist associated Hamiltonians that exhibit elegant spectral structures.

\subsection{Relativistic quantum Toda chain}
\label{sec:Toda}

Another quantum integrable system that can be formulated as the
q-6v model is the relativistic quantum Toda chain.
The Hamiltonian of the $N$-particle relativistic quantum Toda chain is
\cite{Suris90}
\begin{equation}\label{Hrel}
  H_{\text{Toda}} = \sum_{i=1}^N (e^{\pp_i} + g^2 e^{\pp_i + \qq_{i+1} - \qq_i}),
\end{equation}
where $g$ is a real constant and $\qq_i$, $\pp_i$ satisfy the
canonical commutation relation
\begin{equation}
  [\qq_i, \pp_j] = i\hbar_{\text{Toda}} \delta_{ij}.
\end{equation}
We will consider the closed chain and impose the periodic boundary
condition: $\qq_{i+N} = \qq_i$, $\pp_{i+N} = \pp_i$.

The conserved charges of the relativistic quantum Toda chain are most
conveniently expressed in new variables,
\begin{equation}
  c_i = g^2 e^{\pp_i + \qq_{i+1} - \qq_i},
  \qquad
  d_i = e^{\pp_i},
\end{equation}
satisfying 
\begin{align}\label{q-cd}
c_i c_{i+1} = e^{i \hbar_{\text{Toda}}} c_{i+1}c_i, 
\quad 
c_i d_{i} = e^{-i \hbar_{\text{Toda}}} d_{i}c_i,
\quad 
c_i d_{i+1} = e^{i \hbar_{\text{Toda}}} d_{i+1}c_i.  
\end{align}
The Hamiltonian \eqref{Hrel} is simply written as $H_{\text{Toda}}=\sum_{i=1}^N (c_i+d_i)$.
The total momentum $\mathsf{P} = \sum_{i=1}^N \pp_i$ is a conserved quantity commuting with $H_{\text{Toda}}$. 
Let us consider the sector in which $\mathsf{P} = P \in \R$.  
Note that the operators $c_i$, $d_i$ leave
this sector invariant, where they satisfy the relations
\begin{equation}
  g^{-2N} e^{-\frac{N-2}{2} i \hbar_{\text Toda}} c_1 c_2 \cdots c_N = d_1 d_2 \cdots d_N = e^P.
\end{equation}
The claim is that this sector of the relativistic quantum Toda chain
is a sector of the q-6v model on the admissible graph
 $G$ in Figure~\ref{fig:closed-Toda}.

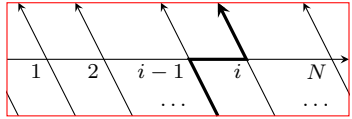
\begin{figure}[H]
  \begin{tikzpicture}[>=stealth, scale=0.75, font=\scriptsize]
    \draw[->] (0.3,1) -- (6.3,1);

    \begin{scope}
      \clip (0.3,0) rectangle (6.3,2);
      \draw[->] (0.5,0) -- (-0.5,2);
      \draw[->] (6.5,0) -- (5.5,2);
    \end{scope}
    
    \draw[->] (1.5,0) -- (0.5,2);
    \draw[->] (2.5,0) -- (1.5,2);
    \draw[->] (4,0) -- (3,2);
    \draw[->] (5,0) -- (4,2);
    \node at (3.25,0.2) {$\dots$};
    \node at (5.75,0.2) {$\dots$};
    
    \node[below left=-2pt] at (1,1) {$1$};
    \node[below left=-2pt] at (2,1) {$2$};
    \node[below left=-2pt] at (3.5,1) {$i-1$};
    \node[below left=-2pt] at (4.5,1) {$i$};
    \node[below left=-2pt] at (6,1) {$N$};

    \draw[very thick, ->] (4,0) -- (3.5,1) -- (4.5,1) -- (4,2);

    \draw[red] (0.3,0) rectangle (6.3,2);
  \end{tikzpicture}
  \caption{The admissible graph for the closed relativistic Toda
    chain.  The configuration shown here contributes to the
    Hamiltonian.}
  \label{fig:closed-Toda}
\end{figure}

For the q-6v model on the graph $G$ in Figure \ref{fig:closed-Toda}, 
without loss of generality let us set
\begin{equation}
  r_i = f_i = 1
\end{equation}
by rescaling the operator $\LL(r,s,f,g;q)$ as
\begin{equation}
  \LL \to r^{-1} \LL,
  \quad
  (e^{\uu}, e^{\ww}) \to (f^{-1} e^{\uu}, r^{-1} e^{\ww}),
  \quad
  (r,s,f,g) \to (r,rs, rf, rfg).
\end{equation}
The transfer matrix $T_G(x,y)$ can be expanded as
\begin{equation}\label{Trel}
  T_G(x,y) = \sum_{m=-1}^1 \sum_{n=0}^N T_{m,n} x^m y^n,
\end{equation}
with the nontrivial coefficients for $n = 0$ and $1$ being
\begin{equation}
  T_{1,0} = \prod_{i=1}^N e^{\uu_i},
  \qquad
  T_{0,1}
  =
  \sum_{i=1}^N
  (s_i e^{\ww_i-\ww_{i-1}} + g_i e^{2\uu_i + \ww_i - \ww_{i-1}}).
\end{equation}
Indeed, in the sector with $T_{1,0} = e^U$ with $U \in \R$, we have
$T_{0,1} = H_{\text{Toda}}$ under the identification
\begin{equation}\label{cd-uw}
  c_i = g_{i+1} e^{2\uu_{i+1} + \ww_{i+1} - \ww_i},
  \quad
  d_i = s_{i+1} e^{\ww_{i+1}-\ww_i},
  \quad
  \hbar_{\text{Toda}} = 2i\hbar,
  \quad
  g^{2N} = e^{2U} \prod_{i=1}^N \frac{g_i}{s_i},
  \quad
  e^P = \prod_{i=1}^N s_i.
\end{equation}
Thus the q-6v model defines the relativistic quantum Toda chain. 
We will show in \S \ref{sec:6v-d} that the transfer matrix
$T_G(x,y)$ \eqref{Trel} can be expressed as the determinant of the Kasteleyn matrix $K(x,y)$ for the corresponding dimer model,
\begin{equation}\label{Krel}
  K(x,y) =
  \begin{pmatrix}
    ys_2 e^{\ww_2-\ww_1} - 1 &  e^{\uu_2 + \ww_2 - \ww_1} &&& x^{-1} yg_1 e^{\uu_1} \\
    yg_2 e^{\uu_2} & ys_3 e^{\ww_3-\ww_2} - 1 &  e^{\uu_3 + \ww_3 - \ww_2} \\
    & yg_3 e^{\uu_3} & \ddots & \ddots \\
    && \ddots & \\
    &&&& e^{\uu_N + \ww_N - \ww_{N-1}} \\
    x e^{\uu_1 + \ww_1 - \ww_N} &&& yg_N e^{\uu_N} & ys_1 e^{\ww_1-\ww_N} - 1
  \end{pmatrix}.
\end{equation}
More precisely, the coefficient $K_{m,n}$ of $x^m y^n$ in
$\det_\star \! K(x,y)$ is related to $T_{m,n}$ by
\begin{equation}
  T_{m,n} = (-1)^{N + mn + m + n} K_{m,n},
\end{equation}
where  $\det_{\star} $ is a Weyl ordered determinant, which will be defined in \eqref{detK}.

We remark that the matrix \eqref{Krel} is a natural quantization of the Lax matrix for the (classical) relativistic Toda chain introduced in \cite{BR89},
\begin{equation}
  L(\mu,\nu)
  =
  \begin{pmatrix}
    \nu d_1 - 1 & \nu c_1 &&& \mu^{-1} \\
    1 & \nu d_2 - 1 & \nu c_2 \\
    & 1 & \ddots & \ddots \\
    &&&& \nu c_{N-1} \\
    \mu\nu c_N &&& 1 & \nu d_N - 1
  \end{pmatrix}.
\end{equation}
Now $c_i$ and $d_i$ satisfy the Poisson brackets.
\begin{align}
\{c_i, c_{i+1}\} = c_i c_{i+1},
\quad 
\{ c_i,  d_{i}\} = c_i d_{i},
\quad 
\{c_i, d_{i+1}\} = c_i d_{i+1},  
\end{align}
corresponding to the classical counterpart of \eqref{q-cd}. 
This Lax matrix is related to the (classical limit of) Kasteleyn matrix \eqref{Krel}
by a gauge transformation $D^{-1}K(x,y)D = L(x/D_N, y)$ with 
\begin{equation}
  D = \mathrm{diag}(D_1, D_2, \dots, D_N),
  \qquad
  D_i = \prod_{j=1}^i yg_j e^{\uu_j},
\end{equation}
and the identification \eqref{cd-uw}.
It is expected that the quantum model given by \eqref{Krel} is equivalent to the conventional relativistic quantum Toda chain given by a $2$ by $2$ Lax matrix. We note that  another realization of the relativistic Toda chain  in terms of dimer model is studied in \cite{EFS12}.

\section{Relation to dimer models}
\label{sec:dimer}

\subsection{Dimer models}

A \emph{bipartite graph} is a graph such that each vertex is colored
either black or white and each edge connects exactly one black and one
white vertex.  The edges of a bipartite graph are canonically oriented
from a black to a white vertex.  A \emph{perfect matching} $M$ of a
bipartite graph $\Gamma$ is a subset of the edges of $\Gamma$ such
that each vertex of $\Gamma$ is contained in exactly one edge of $M$.
The set of all perfect matchings of $\Gamma$ will be denoted by
$\MM(\Gamma)$.

For a bipartite graph $\Gamma$ embedded in a surface $\Sigma$ and
$M_1$, $M_2 \in \mathcal{M}(\Gamma)$, the sum of all edges in $M_1$
and the sum of all edges in $M_2$ have the same boundary, namely the
sum of all black vertices minus the sum of all white vertices.  Their
difference is therefore a $1$-cycle of $\Sigma$, and we will write
$[M_1 - M_2]$ to denote its homology class in $H_1(\Sigma)$.

The dimer model is a statistical mechanics model defined on a
bipartite graph $\Gamma$ on a surface $\Sigma$ such that the states
are represented by the perfect matchings of $\Gamma$.  Note that for
perfect matchings to exist, there must be an equal number of black and
white vertices.  Each edge $e$ of $\Gamma$ is assigned a local
Boltzmann weight $w(e)$, which is usually a positive real number in
the context of physical statistical mechanics.  We will take $\Sigma$
to be a torus and consider the quantized case in which $w(e)$ is an
element of the product of some copies of the $q$-Weyl algebra.

Choose a perfect matching $M_0$ of $\Gamma$.  The partition function
of the dimer model on a torus is given by
\begin{equation}
  Z_\Gamma(x, y)
  =
  \sum_{M \in \mathcal{M}(\Gamma)}
  \prod_{e \in M} w(e)
  x^{\langle [M - M_0], [A]\rangle} y^{\langle [M - M_0], [B]\rangle},
\end{equation}
where $x$, $y$ are parameters, $[A]$, $[B]$ are generators of
$H_1(\Sigma)$ represented by $1$-cycles $A$, $B$ of the torus, and
$\langle\ ,\ \rangle$ is the intersection pairing.  The product of
local Boltzmann weights is understood to be Weyl ordered:
\begin{equation}\label{Weyl-o}
  e^{a\uu} \star e^{b\ww} = e^{b\ww} \star e^{a\uu} = e^{a\uu + b\ww}.
\end{equation}
Up to multiplication by monomials in $x$ and $y$, the partition
function is independent of the choice of the reference perfect
matching $M_0$.

To be concrete, and to connect with the previous discussions on the
q-6v model, let us fix a fundamental domain of the
torus (as in the definition \eqref{pT_G} of the monodromy matrix) and
take $A$ to be the east boundary oriented upward and $B$ to be the
north boundary oriented to the left.  Then, a cycle $\gamma$
traversing the fundamental domain from west to east horizontally has
$\langle[\gamma], [A]\rangle = 1$ and
$\langle[\gamma], [B]\rangle = 0$, while a cycle $\gamma'$ traversing
the fundamental domain from south to north vertically has
$\langle[\gamma'], [A]\rangle = 0$ and
$\langle[\gamma'], [B]\rangle = 1$.

The partition function $Z_\Gamma(x, y)$ can be calculated from a
\emph{Kasteleyn matrix} $K(x,y)$ whose rows and columns are indexed by
the black and white vertices of $\Gamma$, respectively.  The
construction of $K(x,y)$ requires a choice of sign $\sigma(e) = \pm 1$
for each edge $e$.  The sign assignment must satisfy the condition that
for every face with $L$ edges, the product of signs of its edges is
equal to $(-1)^{L/2+1}$.  The matrix element $K(x,y)_{bw}=K_{bw}$ for a black
vertex $b$ and a white vertex $w$, connected by an edge $e$, is given
by
\begin{equation}
  K(x,y)_{bw} = \sigma(e) w(e) x^{\langle e, A\rangle} y^{\langle e, B\rangle}.
\end{equation}
The pairing $\langle e, A\rangle$ equals $+1$ or $-1$ if $e$ crosses
$A$ from left to right or from right to left, and $0$ otherwise;
$\langle e, B\rangle$ is defined in a similar manner.

For an $N$ by $N$ Kasteleyn matrix $K(x,y)$, define its determinant $\det_\star \! K(x,y)$ using the Weyl ordering \eqref{Weyl-o} as
\begin{equation}\label{detK}
  {\det}_{\star}  K(x,y) = \sum_{\sigma \in \mathfrak{G}_N} (-1)^{\mathrm{sgn}(\sigma)} K_{b_1, w_{\sigma(1)}} \star K_{b_2, w_{\sigma(2)}} \star \cdots \star 
K_{b_N, w_{\sigma(N)}}. 
\end{equation}
This is a polynomial in $x$ and $y$, and we write $\det_\star \! K(x,y) = \sum_{m,n} K_{m,n} x^m y^n$.

According to a classic theorem of Kasteleyn \cite{K}, adapted to the
case of a bipartite graph on a torus \cite{KOS}, the partition
function of the dimer model can be expressed as
\begin{equation}
  Z_\Gamma(x,y)
  =
  \sum_{m,n} \pm K_{m,n} x^{m - m_0} y^{n - n_0} \,,
\end{equation}
where $m_0 = \sum_{m \in M_0} \langle m, A\rangle$ and
$n_0 = \sum_{m \in M_0} \langle m, B\rangle$.

\subsection{Quantized six-vertex model as a dimer model}
\label{sec:6v-d}

In order to reformulate the q-6v model as a dimer model,
we replace every vertex of the wiring diagram $G$ locally with a
bipartite graph as follows:%
\footnote{Given a bipartite graph, one can construct a quiver such
  that its underlying undirected graph is dual to the bipartite graph
  and the arrows surrounding a white (resp.\ black) vertex of the
  bipartite graph are oriented clockwise (resp.\ counterclockwise)
  with weight $1/2$.  For the bipartite graph for the quantized
  six-vertex model, the associated quiver is the symmetric butterfly
  quiver studied in \cite{IKSTY}.  Similarly, for the bipartite graph
  for the quantized five-vertex model for $f = 0$, the associated
  quiver is the Fock--Goncharov quiver \cite{IKT1}.}
\begin{equation}
  \label{6vdim}
  \begin{tikzpicture}[>=latex, baseline={(0,0.9)}]
    \draw[->] (0,1) --(2,1);
    \draw[->] (1,0) --(1,2);
  \end{tikzpicture}
  \quad \to \quad
  \begin{tikzpicture}[font=\scriptsize, baseline={(0,0.9)}]
    \draw[-] (0,1) -- node[above=-2pt, xshift=-2pt] {$1$} (0.5,1);
    \draw[-] (1.5,1) -- node[above=-2pt, xshift=4pt] {$e^{-\ww}$} (2,1);
    \draw[-] (1,0) -- node[right=-2pt, yshift=-2pt] {$1$} (1,0.5);
    \draw[-] (1,1.5) -- node[right=-2pt, yshift=2pt] {$1$} (1,2);

    \draw[-] (0.5,1) -- node[above left=-5pt] {$se^\ww$} (1,1.5)
    -- node[above right=-3pt] {$ge^\uu$} (1.5,1)
    -- node[below right=-3pt] {$r$} (1,0.5)
    -- node[below left=-4pt] {$fe^{\uu+\ww}$} cycle;

    \draw[fill=black] (0.5,1) circle[radius=2pt];
    \draw[fill=black] (1.5,1) circle[radius=2pt];
    \draw[fill=white] (1,1.5) circle[radius=2pt];
    \draw[fill=white] (1,0.5) circle[radius=2pt];
  \end{tikzpicture}
  \
  .
\end{equation}
Here $e^\uu$, $e^\ww$ are the generators of the copy of the $q$-Weyl
algebra $\mathcal{W}(q)$ assigned to the vertex of $G$ in question.
If the resulting graph has two vertices of the same color connected by
an edge, we introduce another vertex of the opposite color between
them (and may subsequently shrink the 2-valent vertices).  Then we
obtain a bipartite graph, which we call $\Gamma(G)$.  An edge
connecting a white vertex and a black one may have two weights
assigned.  In that case it is understood that the product of the two
is the weight of that edge.

There are seven perfect matchings for the local piece of the bipartite
graph on the right-hand side of \eqref{6vdim}, listed as states
1--5, 6a and 6b in Table~\ref{tab:6vd}.  States 1--5 correspond
to five configurations in the q-6v model, whereas states
6a and 6b account for the two terms in the weight for the remaining
configuration.

\begin{table}
  \small
  \centering
  \renewcommand{\arraystretch}{2}
  \begin{tabular}{lccccccc}
    State & 1 & 2 & 3 & 4 & 5 & 6a & 6b
    \\
    Perfect matching
    &
    \begin{tikzpicture}[baseline={(0,0.5)}, scale=0.7]
      \draw[-, very thick, blue] (0,1) -- (0.5,1);
      \draw[-] (1.5,1) -- (2,1);
      \draw[-] (1,0) -- (1,0.5);
      \draw[-, very thick, blue] (1,1.5) -- (1,2);

      \draw[-] (0.5,1) -- (1,1.5);
      \draw[-] (1,1.5) -- (1.5,1);
      \draw[-, very thick, blue] (1.5,1) -- (1,0.5);
      \draw[-] (1,0.5) -- (0.5,1);
      
      \draw[fill=black] (0.5,1) circle[radius=3pt];
      \draw[fill=black] (1.5,1) circle[radius=3pt];
      \draw[fill=white] (1,1.5) circle[radius=3pt];
      \draw[fill=white] (1,0.5) circle[radius=3pt];
    \end{tikzpicture}
    &
    \begin{tikzpicture}[baseline={(0,0.5)}, scale=0.7]
      \draw[-] (0,1) -- (0.5,1);
      \draw[-, very thick, blue] (1.5,1) -- (2,1);
      \draw[-, very thick, blue] (1,0) -- (1,0.5);
      \draw[-] (1,1.5) -- (1,2);

      \draw[-, very thick, blue] (0.5,1) -- (1,1.5);
      \draw[-] (1,1.5) -- (1.5,1);
      \draw[-] (1.5,1) -- (1,0.5);
      \draw[-] (1,0.5) -- (0.5,1);
      
      \draw[fill=black] (0.5,1) circle[radius=3pt];
      \draw[fill=black] (1.5,1) circle[radius=3pt];
      \draw[fill=white] (1,1.5) circle[radius=3pt];
      \draw[fill=white] (1,0.5) circle[radius=3pt];
    \end{tikzpicture}
    &
    \begin{tikzpicture}[baseline={(0,0.5)}, scale=0.7]
      \draw[-] (0,1) -- (0.5,1);
      \draw[-, very thick, blue] (1.5,1) -- (2,1);
      \draw[-] (1,0) -- (1,0.5);
      \draw[-, very thick, blue] (1,1.5) -- (1,2);

      \draw[-] (0.5,1) -- (1,1.5);
      \draw[-] (1,1.5) -- (1.5,1);
      \draw[-] (1.5,1) -- (1,0.5);
      \draw[-, very thick, blue] (1,0.5) -- (0.5,1);
      
      \draw[fill=black] (0.5,1) circle[radius=3pt];
      \draw[fill=black] (1.5,1) circle[radius=3pt];
      \draw[fill=white] (1,1.5) circle[radius=3pt];
      \draw[fill=white] (1,0.5) circle[radius=3pt];
    \end{tikzpicture}
    &
    \begin{tikzpicture}[baseline={(0,0.5)}, scale=0.7]
      \draw[-, very thick, blue] (0,1) -- (0.5,1);
      \draw[-] (1.5,1) -- (2,1);
      \draw[-, very thick, blue] (1,0) -- (1,0.5);
      \draw[-] (1,1.5) -- (1,2);

      \draw[-] (0.5,1) -- (1,1.5);
      \draw[-, very thick, blue] (1,1.5) -- (1.5,1);
      \draw[-] (1.5,1) -- (1,0.5);
      \draw[-] (1,0.5) -- (0.5,1);
      
      \draw[fill=black] (0.5,1) circle[radius=3pt];
      \draw[fill=black] (1.5,1) circle[radius=3pt];
      \draw[fill=white] (1,1.5) circle[radius=3pt];
      \draw[fill=white] (1,0.5) circle[radius=3pt];
    \end{tikzpicture}
    &
    \begin{tikzpicture}[baseline={(0,0.5)}, scale=0.7]
      \draw[-, very thick, blue] (0,1) -- (0.5,1);
      \draw[-, very thick, blue] (1.5,1) -- (2,1);
      \draw[-, very thick, blue] (1,0) -- (1,0.5);
      \draw[-, very thick, blue] (1,1.5) -- (1,2);

      \draw[-] (0.5,1) -- (1,1.5);
      \draw[-] (1,1.5) -- (1.5,1);
      \draw[-] (1.5,1) -- (1,0.5);
      \draw[-] (1,0.5) -- (0.5,1);
      
      \draw[fill=black] (0.5,1) circle[radius=3pt];
      \draw[fill=black] (1.5,1) circle[radius=3pt];
      \draw[fill=white] (1,1.5) circle[radius=3pt];
      \draw[fill=white] (1,0.5) circle[radius=3pt];
    \end{tikzpicture}
    &
    \begin{tikzpicture}[baseline={(0,0.5)}, scale=0.7]
      \draw[-] (0,1) -- (0.5,1);
      \draw[-] (1.5,1) -- (2,1);
      \draw[-] (1,0) -- (1,0.5);
      \draw[-] (1,1.5) -- (1,2);

      \draw[-, very thick, blue] (0.5,1) -- (1,1.5);
      \draw[-] (1,1.5) -- (1.5,1);
      \draw[-, very thick, blue] (1.5,1) -- (1,0.5);
      \draw[-] (1,0.5) -- (0.5,1);
      
      \draw[fill=black] (0.5,1) circle[radius=3pt];
      \draw[fill=black] (1.5,1) circle[radius=3pt];
      \draw[fill=white] (1,1.5) circle[radius=3pt];
      \draw[fill=white] (1,0.5) circle[radius=3pt];
    \end{tikzpicture}
    &
    \begin{tikzpicture}[baseline={(0,0.5)}, scale=0.7]
      \draw[-] (0,1) -- (0.5,1);
      \draw[-] (1.5,1) -- (2,1);
      \draw[-] (1,0) -- (1,0.5);
      \draw[-] (1,1.5) -- (1,2);

      \draw[-] (0.5,1) -- (1,1.5);
      \draw[-, very thick, blue] (1,1.5) -- (1.5,1);
      \draw[-] (1.5,1) -- (1,0.5);
      \draw[-, very thick, blue] (1,0.5) -- (0.5,1);
      
      \draw[fill=black] (0.5,1) circle[radius=3pt];
      \draw[fill=black] (1.5,1) circle[radius=3pt];
      \draw[fill=white] (1,1.5) circle[radius=3pt];
      \draw[fill=white] (1,0.5) circle[radius=3pt];
    \end{tikzpicture}
    \\[5mm]
    $1$-chain
    &
    \begin{tikzpicture}[baseline={(0,0.5)}, scale=0.7, >=latex]
      \draw[-] (0,1) -- (0.5,1);
      \draw[-] (1.5,1) -- (2,1);
      \draw[-] (1,0) -- (1,0.5);
      \draw[-] (1,1.5) -- (1,2);

      \draw[-] (0.5,1) -- (1,1.5);
      \draw[-] (1,1.5) -- (1.5,1);
      \draw[-] (1.5,1) -- (1,0.5);
      \draw[-] (1,0.5) -- (0.5,1);
      
      \draw[fill=black] (0.5,1) circle[radius=3pt];
      \draw[fill=black] (1.5,1) circle[radius=3pt];
      \draw[fill=white] (1,1.5) circle[radius=3pt];
      \draw[fill=white] (1,0.5) circle[radius=3pt];
    \end{tikzpicture}
    &
    \begin{tikzpicture}[baseline={(0,0.5)}, scale=0.7, >=latex]
      \draw[-, very thick] (0,1) -- (0.5,1);
      \draw[->] (1.5,1) -- (2,1);
      \draw[-, very thick] (1.5,1) -- (1.9,1);
      \draw[-, very thick] (1,0) -- (1,0.5);
      \draw[->] (1,1.5) -- (1,2);
      \draw[-, very thick] (1,1.5) -- (1,1.9);

      \draw[-, very thick] (0.5,1) -- (1,1.5);
      \draw[-] (1,1.5) -- (1.5,1);
      \draw[-, very thick] (1.5,1) -- (1,0.5);
      \draw[-] (1,0.5) -- (0.5,1);
      
      \draw[fill=black] (0.5,1) circle[radius=3pt];
      \draw[fill=black] (1.5,1) circle[radius=3pt];
      \draw[fill=white] (1,1.5) circle[radius=3pt];
      \draw[fill=white] (1,0.5) circle[radius=3pt];
    \end{tikzpicture}
    &
    \begin{tikzpicture}[baseline={(0,0.5)}, scale=0.7, >=latex]
      \draw[-, very thick] (0,1) -- (0.5,1);
      \draw[->] (1.5,1) -- (2,1);
      \draw[-, very thick] (1.5,1) -- (1.9,1);
      \draw[-] (1,0) -- (1,0.5);
      \draw[-] (1,1.5) -- (1,2);

      \draw[-] (0.5,1) -- (1,1.5);
      \draw[-] (1,1.5) -- (1.5,1);
      \draw[-, very thick] (1.5,1) -- (1,0.5);
      \draw[-, very thick] (1,0.5) -- (0.5,1);
      
      \draw[fill=black] (0.5,1) circle[radius=3pt];
      \draw[fill=black] (1.5,1) circle[radius=3pt];
      \draw[fill=white] (1,1.5) circle[radius=3pt];
      \draw[fill=white] (1,0.5) circle[radius=3pt];
    \end{tikzpicture}
    &
    \begin{tikzpicture}[baseline={(0,0.5)}, scale=0.7, >=latex]
      \draw[-] (0,1) -- (0.5,1);
      \draw[-] (1.5,1) -- (2,1);
      \draw[-, very thick] (1,0) -- (1,0.5);
      \draw[->] (1,1.5) -- (1,2);
      \draw[-, very thick] (1,1.5) -- (1,1.9);

      \draw[-] (0.5,1) -- (1,1.5);
      \draw[-, very thick] (1,1.5) -- (1.5,1);
      \draw[-, very thick] (1.5,1) -- (1,0.5);
      \draw[-] (1,0.5) -- (0.5,1);
      
      \draw[fill=black] (0.5,1) circle[radius=3pt];
      \draw[fill=black] (1.5,1) circle[radius=3pt];
      \draw[fill=white] (1,1.5) circle[radius=3pt];
      \draw[fill=white] (1,0.5) circle[radius=3pt];
    \end{tikzpicture}
    &
    \begin{tikzpicture}[baseline={(0,0.5)}, scale=0.7, >=latex]
      \draw[-] (0,1) -- (0.5,1);
      \draw[->] (1.5,1) -- (2,1);
      \draw[-, very thick] (1.5,1) -- (1.9,1);
      \draw[-, very thick] (1,0) -- (1,0.5);
      \draw[-] (1,1.5) -- (1,2);

      \draw[-] (0.5,1) -- (1,1.5);
      \draw[-] (1,1.5) -- (1.5,1);
      \draw[-, very thick] (1.5,1) -- (1,0.5);
      \draw[-] (1,0.5) -- (0.5,1);
      
      \draw[fill=black] (0.5,1) circle[radius=3pt];
      \draw[fill=black] (1.5,1) circle[radius=3pt];
      \draw[fill=white] (1,1.5) circle[radius=3pt];
      \draw[fill=white] (1,0.5) circle[radius=3pt];
    \end{tikzpicture}
    &
    \begin{tikzpicture}[baseline={(0,0.5)}, scale=0.7, >=latex]
      \draw[-, very thick] (0,1) -- (0.5,1);
      \draw[-] (1.5,1) -- (2,1);
      \draw[-] (1,0) -- (1,0.5);
      \draw[->] (1,1.5) -- (1,2);
      \draw[-, very thick] (1,1.5) -- (1,1.9);

      \draw[-, very thick] (0.5,1) -- (1,1.5);
      \draw[-] (1,1.5) -- (1.5,1);
      \draw[-] (1.5,1) -- (1,0.5);
      \draw[-] (1,0.5) -- (0.5,1);
      
      \draw[fill=black] (0.5,1) circle[radius=3pt];
      \draw[fill=black] (1.5,1) circle[radius=3pt];
      \draw[fill=white] (1,1.5) circle[radius=3pt];
      \draw[fill=white] (1,0.5) circle[radius=3pt];
    \end{tikzpicture}
    &
    \begin{tikzpicture}[baseline={(0,0.5)}, scale=0.7, >=latex]
      \draw[-, very thick] (0,1) -- (0.5,1);
      \draw[-] (1.5,1) -- (2,1);
      \draw[-] (1,0) -- (1,0.5);
      \draw[->] (1,1.5) -- (1,2);
      \draw[-, very thick] (1,1.5) -- (1,1.9);

      \draw[-] (0.5,1) -- (1,1.5);
      \draw[-, thick] (1,1.5) -- (1.5,1);
      \draw[-, thick] (1.5,1) -- (1,0.5);
      \draw[-, thick] (1,0.5) -- (0.5,1);

      \draw[fill=black] (0.5,1) circle[radius=3pt];
      \draw[fill=black] (1.5,1) circle[radius=3pt];
      \draw[fill=white] (1,1.5) circle[radius=3pt];
      \draw[fill=white] (1,0.5) circle[radius=3pt];
    \end{tikzpicture}
    \\[5mm]
    q-6v configuration
    &
    \begin{tikzpicture}[baseline={(0,0.5)}, scale=0.7, >=latex]
      \draw[->] (0,1) -- (2,1);
      \draw[->] (1,0) -- (1,2);
    \end{tikzpicture}
    &
    \begin{tikzpicture}[baseline={(0,0.5)}, scale=0.7, >=latex]
      \draw[->] (0,1) -- (2,1);
      \draw[->] (1,0) -- (1,2);
      \draw[-, very thick] (0,1) -- (1.9,1);
      \draw[-, very thick] (1,0) -- (1,1.9);
    \end{tikzpicture}
    &
    \begin{tikzpicture}[baseline={(0,0.5)}, scale=0.7, >=latex]
      \draw[->] (0,1) -- (2,1);
      \draw[->] (1,0) -- (1,2);
      \draw[-, very thick] (0,1) -- (1.9,1);
    \end{tikzpicture}
    &
    \begin{tikzpicture}[baseline={(0,0.5)}, scale=0.7, >=latex]
      \draw[->] (0,1) -- (2,1);
      \draw[->] (1,0) -- (1,2);
      \draw[-, very thick] (1,0) -- (1,1.9);
    \end{tikzpicture}
    &
    \begin{tikzpicture}[baseline={(0,0.5)}, scale=0.7, >=latex]
      \draw[->] (0,1) -- (2,1);
      \draw[->] (1,0) -- (1,2);
      \draw[-, very thick] (1,1) -- (1.9,1);
      \draw[-, very thick] (1,0) -- (1,1);
    \end{tikzpicture}
    &
    \begin{tikzpicture}[baseline={(0,0.5)}, scale=0.7, >=latex]
      \draw[->] (0,1) -- (2,1);
      \draw[->] (1,0) -- (1,2);
      \draw[-, very thick] (0,1) -- (1,1);
      \draw[-, very thick] (1,1) -- (1,1.9);
    \end{tikzpicture}
    &
    \begin{tikzpicture}[baseline={(0,0.5)}, scale=0.7, >=latex]
      \draw[->] (0,1) -- (2,1);
      \draw[->] (1,0) -- (1,2);
      \draw[-, very thick] (0,1) -- (1,1);
      \draw[-, very thick] (1,1) -- (1,1.9);
    \end{tikzpicture}
    \\[3mm]
    Weight
    &
    $r$ & $s$ & $fe^\uu$ & $ge^\uu$ & $e^{-\ww}$ & $rse^\ww$ & $fge^{2\uu+\ww}$
    \\[3mm]
  \end{tabular}

  \caption{The local state configurations for the dimer model on the
    bipartite graph obtained from a wiring diagram.  State 1 is chosen
    as the reference perfect matching.  The third row lists the
    1-chain given by the difference of the perfect matching in
    question and the reference perfect matching.  The fourth row shows
    the corresponding configuration of the q-6v model,
    where thick edges indicate that their states are $1$.}
  \label{tab:6vd}
\end{table}

For the bipartite graph $\Gamma(G)$ obtained from a wiring diagram $G$
for the q-6v model, the condition for a perfect matching
$M$, that each vertex must be contained in exactly one edge of $M$, is
equivalent to the charge conservation rule.  For example, suppose that
the wiring diagram contains a horizontal wire directed to the right,
crossed by two vertical wires directed upward:
\begin{equation}
  \begin{tikzpicture}[>=latex, baseline={(0,0.9)}]
    \draw[->] (0,1) --(3,1);
    \draw[->] (1,0) --(1,2);
    \draw[->] (2,0) --(2,2);
  \end{tikzpicture}
  \quad \to \quad
  \begin{tikzpicture}[font=\scriptsize, baseline={(0,0.9)}]
    \draw[-] (0,1) -- (0.5,1);
    \draw[-] (1.5,1) -- (2,1);
    \draw[-] (1,0) -- (1,0.5);
    \draw[-] (1,1.5) -- (1,2);

    \draw[-] (0.5,1) -- (1,1.5)
    -- (1.5,1)
    -- (1,0.5)
    -- cycle;

    \draw[fill=black] (0.5,1) circle[radius=2pt];
    \draw[fill=black] (1.5,1) circle[radius=2pt];
    \draw[fill=white] (1,1.5) circle[radius=2pt];
    \draw[fill=white] (1,0.5) circle[radius=2pt];

    \begin{scope}[shift={(2,0)}]
      \draw[-] (0,1) -- (0.5,1);
      \draw[-] (1.5,1) -- (2,1);
      \draw[-] (1,0) -- (1,0.5);
      \draw[-] (1,1.5) -- (1,2);
      
      \draw[-] (0.5,1) -- (1,1.5)
      -- (1.5,1)
      -- (1,0.5)
      -- cycle;
      
      \draw[fill=black] (0.5,1) circle[radius=2pt];
      \draw[fill=black] (1.5,1) circle[radius=2pt];
      \draw[fill=white] (1,1.5) circle[radius=2pt];
      \draw[fill=white] (1,0.5) circle[radius=2pt];
    \end{scope}

    \draw[fill=white] (2,1) circle[radius=2pt];
  \end{tikzpicture}
  \ .
\end{equation}
The corresponding bipartite graph contains two copies of the local
piece shown in \eqref{6vdim}, connected horizontally side by side with
an additional white vertex.  If the local piece on the left is in
state 1 in Table~\ref{tab:6vd}, then the one on the right can only be
in states 1, 4 and 5, which match the three local configurations of
the q-6v model where the in-states on the horizontal
wire and the left vertical wire are $0$.

Indeed, if states 6a and 6b are combined as a single state, there is a
one-to-one correspondence between the perfect matchings of the dimer
model and the configurations of the q-6v model.  Hence,
the partition function of the dimer model on $\Gamma(G)$ coincides
with the transfer matrix $T_G(x,y)$ of the q-6v model:
\begin{equation}
  Z_{\Gamma(G)}(x,y) = T_G(x,y) .
\end{equation}

The bipartite graph $\Gamma(G)$ can be reduced to a simpler graph by
the ``shrinking of 2-valent vertex'' move, under which the partition
function of a dimer model is invariant.  This move transforms the
bipartite graph and the edge weights as follows:
\begin{equation}
  \begin{tikzpicture}[font=\scriptsize, baseline={(0,-0.1)}]
    \draw (0,0) -- node[above=-2pt]{$a$} (0.5,0);
    \draw (0.5,0) -- node[above=-2pt]{$b$} (1,0);

    \draw (0,0) --  ++(140:0.5) node[above left=-4pt]{$l_1$};
    \draw (0,0) --  ++(170:0.5) node[left=-2pt]{$l_2$};
    \draw (0,0) -- ++(240:0.5) node[below left=-4pt, xshift=2pt]{$l_m$};
    \node[rotate=-60] at (-0.35,-0.15) {$...$};

    \draw (1,0) -- ++(-50:0.5) node[below right=-3pt]{$r_1$};
    \draw (1,0) -- ++(-20:0.5) node[right=-2pt, yshift=-1pt]{$r_2$};
    \draw (1,0) -- ++(50:0.5) node[above right=-3pt]{$r_n$};
    \node[rotate=-70] at (1.35,0.1) {$...$};

    \draw[fill=black] (0,0) circle[radius=2pt];
    \draw[fill=white] (0.5,0) circle[radius=2pt];
    \draw[fill=black] (1,0) circle[radius=2pt];
  \end{tikzpicture}
  \to 
  \begin{tikzpicture}[font=\scriptsize, baseline={(0,-0.1)}]
    \draw (0,0) --  ++(140:0.5) node[above left=-4pt]{$b \star l_1$};
    \draw (0,0) --  ++(170:0.5) node[left=-2pt]{$b \star l_2$};
    \draw (0,0) -- ++(240:0.5) node[below left=-4pt, xshift=2pt]{$b \star l_m$};
    \node[font=\normalsize, rotate=-60] at (-0.35,-0.15) {$...$};

    \draw (0,0) -- ++(-50:0.5) node[below right=-3pt]{$a \star r_1$};
    \draw (0,0) -- ++(-20:0.5) node[right=-2pt, yshift=-1pt]{$a \star r_2$};
    \draw (0,0) -- ++(50:0.5) node[above right=-3pt]{$a \star r_n$};
    \node[font=\normalsize, rotate=-70] at (0.35,0.1) {$...$};

    \draw[fill=black] (0,0) circle[radius=2pt];
  \end{tikzpicture}
  \
  ,
  \qquad
  \begin{tikzpicture}[font=\scriptsize, baseline={(0,-0.1)}]
    \draw (0,0) -- node[above=-2pt]{$a$} (0.5,0);
    \draw (0.5,0) -- node[above=-2pt]{$b$} (1,0);

    \draw (0,0) --  ++(140:0.5) node[above left=-4pt]{$l_1$};
    \draw (0,0) --  ++(170:0.5) node[left=-2pt]{$l_2$};
    \draw (0,0) -- ++(240:0.5) node[below left=-4pt, xshift=2pt]{$l_m$};
    \node[rotate=-60] at (-0.35,-0.15) {$...$};

    \draw (1,0) -- ++(-50:0.5) node[below right=-3pt]{$r_1$};
    \draw (1,0) -- ++(-20:0.5) node[right=-2pt, yshift=-1pt]{$r_2$};
    \draw (1,0) -- ++(50:0.5) node[above right=-3pt]{$r_n$};
    \node[rotate=-70] at (1.35,0.1) {$...$};

    \draw[fill=white] (0,0) circle[radius=2pt];
    \draw[fill=black] (0.5,0) circle[radius=2pt];
    \draw[fill=white] (1,0) circle[radius=2pt];
  \end{tikzpicture}
  \to 
  \begin{tikzpicture}[font=\scriptsize, baseline={(0,-0.1)}]
    \draw (0,0) --  ++(140:0.5) node[above left=-4pt]{$b \star l_1$};
    \draw (0,0) --  ++(170:0.5) node[left=-2pt]{$b \star l_2$};
    \draw (0,0) -- ++(240:0.5) node[below left=-4pt, xshift=2pt]{$b \star l_m$};
    \node[font=\normalsize, rotate=-60] at (-0.35,-0.15) {$...$};

    \draw (0,0) -- ++(-50:0.5) node[below right=-3pt]{$a \star r_1$};
    \draw (0,0) -- ++(-20:0.5) node[right=-2pt, yshift=-1pt]{$a \star r_2$};
    \draw (0,0) -- ++(50:0.5) node[above right=-3pt]{$a \star r_n$};
    \node[font=\normalsize, rotate=-70] at (0.35,0.1) {$...$};

    \draw[fill=white] (0,0) circle[radius=2pt];
  \end{tikzpicture}
  \
  .
\end{equation}
In the original bipartite graph, each perfect matching either covers
the edge with weight $a$ and one of the edges with weights $r_1$,
\dots, $r_n$, or covers the edge with weight $b$ and one of the edges
with weights $l_1$, \dots, $l_m$.  In the transformed bipartite graph,
the obvious corresponding perfect matching has the same contribution
to the partition function.

As an example, let us take as a wiring diagram a square lattice on
$\R^2$, where the wires intersect at the integral points.  See
Figure~\ref{fig:q6v-dimer}.  The replacement \eqref{6vdim} gives the
first bipartite graph, and shrinking the 2-valent vertices results in
the second bipartite graph.  If we assign a canonical pair
$(\uu_{ij}, \ww_{ij})$ and parameters $r_{ij}$, $s_{ij}$,
$f_{ij}$, $g_{ij}$ to the intersection $(i,j) \in \Z^2$, the weights
of the four edges around $(i,j)$ in the second dimer model are given
as follows:
\begin{equation}
    \begin{tikzpicture}[>=stealth, scale=1.5, baseline={(0,0)}]
     \draw[black!30, ->] (0,-0.8) -- (0,0.8);
     \draw[black!30, ->] (-0.8,0) -- (0.8,0);

     \draw[thick] (-0.5,0)
     -- node[above left=-4pt] {$s_{ij} e^{\ww_{ij} - \ww_{i-1,j}}$}
     (0,0.5)
     -- node[above right=-2pt] {$g_{ij} e^{\uu_{ij}}$}
     (0.5,0)
     -- node[below right=-2pt] {$r_{ij}$}
     (0,-0.5)
     -- node[below left=-4pt] {$f_{ij} e^{\uu_{ij} + \ww_{ij} - \ww_{i-1,j}}$}
     cycle;

     \draw[fill=black] (-0.5,0) circle[radius=1.5pt];
     \draw[fill=white] (0,0.5) circle[radius=1.5pt];
     \draw[fill=black] (0.5,0) circle[radius=1.5pt];
     \draw[fill=white] (0,-0.5) circle[radius=1.5pt];

     \node[font=\scriptsize] at (0,0) {$(i,j)$};
  \end{tikzpicture}
\end{equation}

\begin{figure}[H]
  \centering
  
  \begin{tikzpicture}[>=stealth, scale=1.5]
    \draw[dashed] (-0.6,0.15) rectangle (2.4,2.15);
    \clip (-0.6,0.15) rectangle (2.4,2.15);
    
    \foreach \x in {0,1,2,3} {
      \draw[->] (\x,0) -- (\x,2.15);
    }      

    \foreach \y in {0,1,2} {
      \draw[->] (-0.6,{\y+0.5}) -- (2.4,{\y+0.5});
    }
  \end{tikzpicture}
  \qquad
  \begin{tikzpicture}[>=stealth, scale=1.5]
    \draw[dashed] (-0.6,0.15) rectangle (2.4,2.15);
    \clip (-0.6,0.15) rectangle (2.4,2.15);
    
    \foreach \x in {0,1,2,3} {
      \draw[black!30, ->] (\x,0) -- (\x,2.15);
    }      

    \foreach \y in {0,1,2} {
      \draw[black!30, ->] (-0.6,{\y+0.5}) -- (2.4,{\y+0.5});
    }
      
    \foreach \x in {-1,0,1,2,3,4} {
      \foreach \y in {0,1,2,3} {
        \draw[very thick, blue] (\x-0.5,\y+0.5) -- ++(0.25,0);
        \draw[thick] (\x+0.25,\y+0.5) -- ++(0.25,0);
        \draw[thick] (\x,\y) -- ++(0,0.25);
        \draw[very thick, blue] (\x,\y+0.75) -- ++(0,0.25);
        
        \draw[thick] (\x,\y+0.25) -- (\x-0.25,\y+0.5) -- (\x,\y+0.75) -- (\x+0.25,\y+0.5);
        \draw[very thick, blue] (\x+0.25,\y+0.5) -- (\x,\y+0.25);
        
        \draw[fill=white] (\x,\y+0.25) circle[radius=1.5pt];
        \draw[fill=white] (\x,\y+0.75) circle[radius=1.5pt];
        \draw[fill=black] (\x-0.25,\y+0.5) circle[radius=1.5pt];
        \draw[fill=black] (\x+0.25,\y+0.5) circle[radius=1.5pt];
        
        \draw[fill=white] (\x-0.5,\y+0.5) circle[radius=1.5pt];
        \draw[fill=black] (\x,\y) circle[radius=1.5pt];
      }
    }
  \end{tikzpicture}
  \qquad
  \begin{tikzpicture}[>=stealth, scale=1.5]
    \draw[dashed] (-0.6,0.15) rectangle (2.4,2.15);
    \clip (-0.6,0.15) rectangle (2.4,2.15);

    \foreach \x in {0,1,2,3,4} {
      \draw[black!30, ->] (\x,0) -- (\x,2.15);
    }      

    \foreach \y in {0,1,2} {
      \draw[black!30, ->] (-0.6,{\y+0.5}) -- (2.4,{\y+0.5});
    }

    \foreach \x in {-1,0,1,2,3,4} {
      \foreach \y in {0,1,2,3} {
        \draw[thick] (\x,\y) -- (\x-0.5,\y+0.5)
        -- (\x,\y+1) -- (\x+0.5,\y+0.5);
        \draw[very thick, blue] (\x+0.5,\y+0.5) -- (\x,\y);
        
        \draw[fill=black] (\x-0.5,\y+0.5) circle[radius=1.5pt];
        \draw[fill=white] (\x,\y) circle[radius=1.5pt];
      }
    }

  \end{tikzpicture}

  \caption{The bipartite graph for the q-6v model on a
    plane. The reference perfect matching is indicated by the blue
    edges.}
  \label{fig:q6v-dimer}
\end{figure}
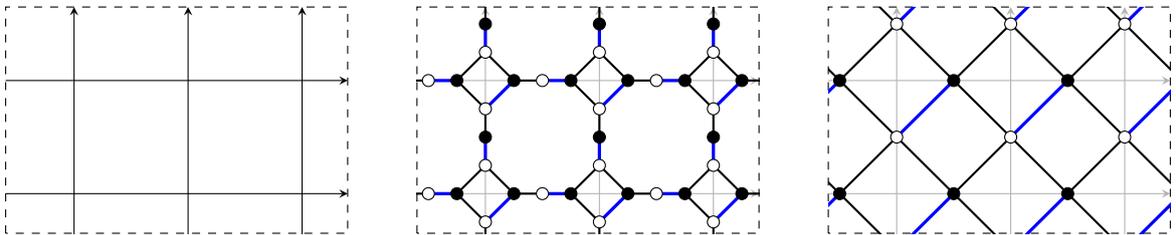

Taking the quotient of $\R^2$ by translations generated by two
integral vectors, we obtain a bipartite graph corresponding to a
wiring diagram on a torus.  Let us consider the case in which the two
vectors are $(M,M)$ and $(N,0)$.  We choose the fundamental domain of
the torus in such a way that the reference perfect matching does not
have edges crossing the boundaries, as in Figure~\ref{fig:Toda-dimer}.
In order to write down the Kasteleyn matrix, it is convenient to
introduce new coordinates $(a,b) = (i-j,j)$ and use them to label the
canonical pairs and the parameters assigned at the intersection
$(i,j)$ as well as the black vertex to the right of it and the white
vertex below it.  For the sign assignment we choose the edges in the
reference perfect matching to have $-1$ and all remaining edges to
have $+1$.  Then, the Kasteleyn matrix is given by
\begin{multline}
  K(x,y)_{ab,cd}
  =
  -\delta_{ac} \delta_{bd} r_{ab}
  + \delta_{ac} \delta_{b+1,d} y^{\delta_{bM}} s_{a+1,b} e^{\ww_{a+1,b} - \ww_{ab}}
  \\
  + \delta_{a+1,c} \delta_{bd} x^{\delta_{aN}}f_{a+1,b} e^{\uu_{a+1,b} + \ww_{a+1,b} - \ww_{ab}}
  + \delta_{a,c+1} \delta_{b+1,d} x^{-\delta_{a1} \delta_{bM}} y^{\delta_{bM}} g_{ab} e^{\uu_{ab}},
\end{multline}
where the indices $a$, $c$ are defined modulo $N$ and $b$, $d$ are
defined modulo $M$.  The relativistic quantum Toda chain discussed in
\S \ref{sec:Toda} is the case $M = 1$.

\begin{figure}[H]
  \centering

  \begin{tikzpicture}[>=stealth]
    \draw[dashed] (-0.6,0.15) rectangle (4.9,3.15);
    \clip (-0.6,0.15) rectangle (4.9,3.15);

    \foreach \x in {0,1,2,3,4} {
      \draw[black!30, ->] (\x,0) -- (\x,3.15);
    }      

    \foreach \y in {0,1,2} {
      \draw[black!30, ->] (-0.6,{\y+0.5}) -- (4.9,{\y+0.5});
    }

    \foreach \x in {-1,0,1,2,3,4,5} {
      \foreach \y in {0,1,2,3} {
        \draw[thick] (\x,\y) -- (\x-0.5,\y+0.5)
        -- (\x,\y+1) -- (\x+0.5,\y+0.5);
        \draw[very thick, blue] (\x+0.5,\y+0.5) -- (\x,\y);
        
        \draw[fill=black] (\x-0.5,\y+0.5) circle[radius=1.5pt];
        \draw[fill=white] (\x,\y) circle[radius=1.5pt];
      }
    }

    \draw[thick, red, shift={(-0.5,0.3)}] (0,0.5) -- (2,2.5) -- (5,2.5) -- (3,0.5) -- cycle;

    \node at (0.5,1.5) [right, font=\tiny] {$11$};
    \node at (1.5,1.5) [right, font=\tiny] {$21$};
    \node at (2.5,1.5) [right, font=\tiny] {$31$};
    \node at (1.5,2.5) [right, font=\tiny] {$12$};
    \node at (2.5,2.5) [right, font=\tiny] {$22$};
    \node at (3.5,2.5) [right, font=\tiny] {$32$};

    \node at (0,1) [right, font=\tiny] {$11$};
    \node at (1,1) [right, font=\tiny] {$21$};
    \node at (2,1) [right, font=\tiny] {$31$};
    \node at (1,2) [right, font=\tiny] {$12$};
    \node at (2,2) [right, font=\tiny] {$22$};
    \node at (3,2) [right, font=\tiny] {$32$};
  \end{tikzpicture}
  \caption{The bipartite graph on a torus for a dimer model
    generalizing the relativistic quantum Toda chain, for $(M,N)=(2.3)$.}
  \label{fig:Toda-dimer}
\end{figure}
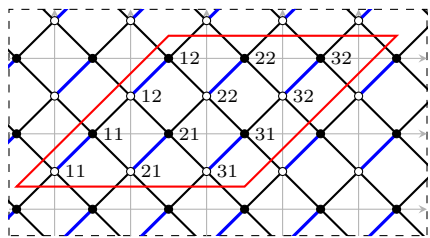

\subsection{Quantized five-vertex model as a dimer model}

When the parameter $f = 0$, only five of the six local configurations
in Figure \ref{fig:6v} contribute to the transfer matrix of the
q-6v model.  In the corresponding bipartite graph
\eqref{6vdim}, setting $f = 0$ amounts to removing the edge with
weight $fe^{\uu+\ww}$.  After this edge is removed, the edges with
weight $1$ connected to the removed edge become redundant; whether or
not they are covered by a perfect matching is correlated with the
adjacent edges with weight $se^\ww$ and $r$.  Therefore, the
reduced model can be reformulated as a dimer model on a simpler
bipartite graph, obtained from the wiring diagram by the following
replacement of the vertices:
\begin{equation}
  \label{q6vdim}
  \begin{tikzpicture}[>=latex, baseline={(0,0.9)}]
    \draw[->] (0,1) --(2,1);
    \draw[->] (1,0) --(1,2);
  \end{tikzpicture}
  \quad \to \quad 
  \begin{tikzpicture}[font=\scriptsize, baseline={(0,0.9)}]
    \draw[-] (0,1.25) -- node[above=-2pt, xshift=-2pt] {$se^\ww$} (0.75,1.25);
    \draw[-] (1.25,0.75) -- node[above=-2pt, xshift=4pt] {$e^{-\ww}$} (2,0.75);
    \draw[-] (1.25,0) -- node[right=-2pt, yshift=-2pt] {$r$} (1.25,0.75);
    \draw[-] (0.75,1.25) -- node[right=-2pt, yshift=4pt] {$1$} (0.75,2);

    \draw[-] (0.75,1.25) -- node[above right=-3pt] {$ge^\uu$} (1.25,0.75);

    \draw[fill=white] (0.75,1.25) circle[radius=2pt];
    \draw[fill=black] (1.25,0.75) circle[radius=2pt];
  \end{tikzpicture}
  \
  .
\end{equation}
See Table \ref{q5vdim} for the perfect matchings and the
corresponding configurations of the quantized five-vertex model.

\begin{table}[H]
  \small
  \centering
  \renewcommand{\arraystretch}{2}
  \begin{tabular}{lccccc}
    Perfect matching
    &
  \begin{tikzpicture}[baseline={(0,0.5)}, scale=0.7]
    \draw[-] (0,1.25) -- (0.75,1.25);
    \draw[-] (1.25,0.75) -- (2,0.75);
    \draw[-, very thick, blue] (1.25,0) -- (1.25,0.75);
    \draw[-, very thick, blue] (0.75,1.25) -- (0.75,2);

    \draw[-] (0.75,1.25) -- (1.25,0.75);

    \draw[fill=white] (0.75,1.25) circle[radius=3pt];
    \draw[fill=black] (1.25,0.75) circle[radius=3pt];
  \end{tikzpicture}
    &
  \begin{tikzpicture}[baseline={(0,0.5)}, scale=0.7]
    \draw[-, very thick, blue] (0,1.25) -- (0.75,1.25);
    \draw[-, very thick, blue] (1.25,0.75) -- (2,0.75);
    \draw[-] (1.25,0) -- (1.25,0.75);
    \draw[-] (0.75,1.25) -- (0.75,2);

    \draw[-] (0.75,1.25) -- (1.25,0.75);

    \draw[fill=white] (0.75,1.25) circle[radius=3pt];
    \draw[fill=black] (1.25,0.75) circle[radius=3pt];
  \end{tikzpicture}
    &
  \begin{tikzpicture}[baseline={(0,0.5)}, scale=0.7]
    \draw[-] (0,1.25) -- (0.75,1.25);
    \draw[-] (1.25,0.75) -- (2,0.75);
    \draw[-] (1.25,0) -- (1.25,0.75);
    \draw[-] (0.75,1.25) -- (0.75,2);

    \draw[-, very thick, blue] (0.75,1.25) -- (1.25,0.75);

    \draw[fill=white] (0.75,1.25) circle[radius=3pt];
    \draw[fill=black] (1.25,0.75) circle[radius=3pt];
  \end{tikzpicture}
    &
  \begin{tikzpicture}[baseline={(0,0.5)}, scale=0.7]
    \draw[-] (0,1.25) -- (0.75,1.25);
    \draw[-, very thick, blue] (1.25,0.75) -- (2,0.75);
    \draw[-] (1.25,0) -- (1.25,0.75);
    \draw[-, very thick, blue] (0.75,1.25) -- (0.75,2);

    \draw[-] (0.75,1.25) -- (1.25,0.75);

    \draw[fill=white] (0.75,1.25) circle[radius=3pt];
    \draw[fill=black] (1.25,0.75) circle[radius=3pt];
  \end{tikzpicture}
    &
  \begin{tikzpicture}[baseline={(0,0.5)}, scale=0.7]
    \draw[-, very thick, blue] (0,1.25) -- (0.75,1.25);
    \draw[-] (1.25,0.75) -- (2,0.75);
    \draw[-, very thick, blue] (1.25,0) -- (1.25,0.75);
    \draw[-] (0.75,1.25) -- (0.75,2);

    \draw[-] (0.75,1.25) -- (1.25,0.75);

    \draw[fill=white] (0.75,1.25) circle[radius=3pt];
    \draw[fill=black] (1.25,0.75) circle[radius=3pt];
  \end{tikzpicture}
    \\[5mm]
    q-5v configuration
    &
    \begin{tikzpicture}[baseline={(0,0.5)}, scale=0.7, >=latex]
      \draw[->] (0,1) -- (2,1);
      \draw[->] (1,0) -- (1,2);
    \end{tikzpicture}
    &
    \begin{tikzpicture}[baseline={(0,0.5)}, scale=0.7, >=latex]
      \draw[->] (0,1) -- (2,1);
      \draw[->] (1,0) -- (1,2);
      \draw[-, very thick] (0,1) -- (1.9,1);
      \draw[-, very thick] (1,0) -- (1,1.9);
    \end{tikzpicture}
    &
    \begin{tikzpicture}[baseline={(0,0.5)}, scale=0.7, >=latex]
      \draw[->] (0,1) -- (2,1);
      \draw[->] (1,0) -- (1,2);
      \draw[-, very thick] (1,0) -- (1,1.9);
    \end{tikzpicture}
    &
    \begin{tikzpicture}[baseline={(0,0.5)}, scale=0.7, >=latex]
      \draw[->] (0,1) -- (2,1);
      \draw[->] (1,0) -- (1,2);
      \draw[-, very thick] (1,1) -- (1.9,1);
      \draw[-, very thick] (1,0) -- (1,1);
    \end{tikzpicture}
    &
    \begin{tikzpicture}[baseline={(0,0.5)}, scale=0.7, >=latex]
      \draw[->] (0,1) -- (2,1);
      \draw[->] (1,0) -- (1,2);
      \draw[-, very thick] (0,1) -- (1,1);
      \draw[-, very thick] (1,1) -- (1,1.9);
    \end{tikzpicture}
    \\
    Weight
    &
    $r$ & $s$ & $ge^\uu$ & $e^{-\ww}$ & $rse^\ww$
  \end{tabular}

  \caption{The perfect matchings and the corresponding configurations
    of the quantized five-vertex model for $f=0$.}
  \label{q5vdim}
\end{table}

One can also formulate a simpler
version of the quantized five-vertex model and the equivalent 
dimer model by considering the specialization $s=0$ in the q-6v model.
From Table \ref{tab:6vd}, the allowed vertices and their weights, along with the 
corresponding perfect matchings are provided as follows:

\begin{table}[H]
  \small
  \centering
  \renewcommand{\arraystretch}{2}
  \begin{tabular}{lccccc}
    Perfect matching
    &
  \begin{tikzpicture}[baseline={(0,0.5)}, scale=0.7]
    \draw[-] (0,0.75) -- (0.75,0.75);
    \draw[-] (1.25,1.25) -- (2,1.25);
    \draw[-] (0.75,0) -- (0.75,0.75);
    \draw[-] (1.25,1.25) -- (1.25,2);

    \draw[-, very thick, blue] (0.75,0.75) -- (1.25,1.25);

    \draw[fill=white] (0.75,0.75) circle[radius=3pt];
    \draw[fill=black] (1.25,1.25) circle[radius=3pt];
  \end{tikzpicture}
    &
  \begin{tikzpicture}[baseline={(0,0.5)}, scale=0.7]
    \draw[-, very thick, blue] (0,0.75) -- (0.75,0.75);
    \draw[-, very thick, blue] (1.25,1.25) -- (2,1.25);
    \draw[-] (0.75,0) -- (0.75,0.75);
    \draw[-] (1.25,1.25) -- (1.25,2);

    \draw[-] (0.75,0.75) -- (1.25,1.25);

    \draw[fill=white] (0.75,0.75) circle[radius=3pt];
    \draw[fill=black] (1.25,1.25) circle[radius=3pt];
  \end{tikzpicture}
    &
  \begin{tikzpicture}[baseline={(0,0.5)}, scale=0.7]
    \draw[-] (0,0.75) -- (0.75,0.75);
    \draw[-] (1.25,1.25) -- (2,1.25);
    \draw[-, very thick, blue] (0.75,0) -- (0.75,0.75);
    \draw[-, very thick, blue] (1.25,1.25) -- (1.25,2);

    \draw[-] (0.75,0.75) -- (1.25,1.25);

    \draw[fill=white] (0.75,0.75) circle[radius=3pt];
    \draw[fill=black] (1.25,1.25) circle[radius=3pt];
  \end{tikzpicture}
    &
  \begin{tikzpicture}[baseline={(0,0.5)}, scale=0.7]
    \draw[-] (0,0.75) -- (0.75,0.75);
    \draw[-, very thick, blue] (1.25,1.25) -- (2,1.25);
    \draw[-, very thick, blue] (0.75,0) -- (0.75,0.75);
    \draw[-] (1.25,1.25) -- (1.25,2);

    \draw[-] (0.75,0.75) -- (1.25,1.25);

    \draw[fill=white] (0.75,0.75) circle[radius=3pt];
    \draw[fill=black] (1.25,1.25) circle[radius=3pt];
  \end{tikzpicture}
    &
  \begin{tikzpicture}[baseline={(0,0.5)}, scale=0.7]
    \draw[-, very thick, blue] (0,0.75) -- (0.75,0.75);
    \draw[-] (1.25,1.25) -- (2,1.25);
    \draw[-] (0.75,0) -- (0.75,0.75);
    \draw[-, very thick, blue] (1.25,1.25) -- (1.25,2);

    \draw[-] (0.75,0.75) -- (1.25,1.25);

    \draw[fill=white] (0.75,0.75) circle[radius=3pt];
    \draw[fill=black] (1.25,1.25) circle[radius=3pt];
  \end{tikzpicture}
    \\[5mm]
    q-5v configuration
    &
    \begin{tikzpicture}[baseline={(0,0.5)}, scale=0.7, >=latex]
      \draw[->] (0,1) -- (2,1);
      \draw[->] (1,0) -- (1,2);
    \end{tikzpicture}
    &
    \begin{tikzpicture}[baseline={(0,0.5)}, scale=0.7, >=latex]
      \draw[->] (0,1) -- (2,1);
      \draw[->] (1,0) -- (1,2);
      \draw[-, very thick] (0,1) -- (1.9,1);
    \end{tikzpicture}
    &
    \begin{tikzpicture}[baseline={(0,0.5)}, scale=0.7, >=latex]
      \draw[->] (0,1) -- (2,1);
      \draw[->] (1,0) -- (1,2);
      \draw[-, very thick] (1,0) -- (1,1.9);
    \end{tikzpicture}
    &
    \begin{tikzpicture}[baseline={(0,0.5)}, scale=0.7, >=latex]
      \draw[->] (0,1) -- (2,1);
      \draw[->] (1,0) -- (1,2);
      \draw[-, very thick] (1,1) -- (1.9,1);
      \draw[-, very thick] (1,0) -- (1,1);
    \end{tikzpicture}
    &
    \begin{tikzpicture}[baseline={(0,0.5)}, scale=0.7, >=latex]
      \draw[->] (0,1) -- (2,1);
      \draw[->] (1,0) -- (1,2);
      \draw[-, very thick] (0,1) -- (1,1);
      \draw[-, very thick] (1,1) -- (1,1.9);
    \end{tikzpicture}
    \\
    Weight
    &
    $r$ & $fe^\uu$ & $ge^\uu$ & $e^{-\ww}$ & $fge^{2\uu+\ww}$
  \end{tabular}

  \caption{The perfect matchings and the corresponding configurations
    of the quantized five-vertex model for $s=0$.}
  \label{5vd2}

\end{table}

\noindent
The equivalent dimer model is obtained by replacing the local vertex configurations
with the perfect matchings as shown in Table \ref{5vd2}.
Their weights are identified by the following local rule:
\begin{equation}
  \label{dim2}
  \begin{tikzpicture}[>=latex, baseline={(0,0.9)}]
    \draw[->] (0,1) --(2,1);
    \draw[->] (1,0) --(1,2);
  \end{tikzpicture}
  \quad \to \quad
  \begin{tikzpicture}[font=\scriptsize, baseline={(0,0.9)}]
    \draw[-] (0,0.75) -- node[above=-2pt, xshift=-2pt] {$fe^\uu$} (0.75,0.75);
    \draw[-] (1.25,1.25) -- node[above=-2pt, xshift=4pt] {$1$} (2,1.25);
    \draw[-] (0.75,0) -- node[right=-2pt, yshift=-2pt] {$e^{-\ww}$} (0.75,0.75);
    \draw[-] (1.25,1.25) -- node[right=-2pt, yshift=4pt] {$ge^{\uu+\ww}$} (1.25,2);

    \draw[-] (0.75,0.75) -- node[above left=-3pt] {$r$} (1.25,1.25);

    \draw[fill=white] (0.75,0.75) circle[radius=2pt];
    \draw[fill=black] (1.25,1.25) circle[radius=2pt];
  \end{tikzpicture}
  \
  .
\end{equation}
Due to the construction, the resulting graph on which the dimer model is defined 
becomes bipartite automatically without the need to introduce additional vertices 
as was required in the previous case of the $f=0$ five-vertex model.
Furthermore, the configurations of the five-vertex model and the dimer model 
can be directly translated to one another without mediating their comparison via a reference 
perfect matching.
This simpler version of the correspondence between the five-vertex and the dimer models 
may be regarded as a quantization of the one studied in \cite{W68,FW70}.


\end{document}